\documentclass{article}

\usepackage[nonatbib,final]{neurips_2020}

\usepackage[utf8]{inputenc} 
\usepackage[T1]{fontenc}    
\usepackage{hyperref}       
\usepackage{url}            
\usepackage{booktabs}       
\usepackage{amsfonts, amsthm}       
\usepackage{nicefrac}       
\usepackage{microtype}      
\usepackage{mathrsfs}

\usepackage[british]{babel}
\usepackage[mono=false]{libertine}
\usepackage[T1]{fontenc}
\usepackage{amsthm}
\usepackage[font=small]{caption}
\usepackage{enumitem}
\usepackage{bm}

\newtheorem{thm}{Theorem}
\newtheorem{lemma}{Lemma} 
\newtheorem{proposition}{Proposition}
\newtheorem{corollary}{Corollary}

\newtheorem{defi}{Definition}

\newenvironment{talign*}
 {\csname align*\endcsname}
 {\endalign}

\usepackage{titlesec}
\makeatletter
\newcommand{\setappendix}{Appendix~\thesection:~~}
\newcommand{\setsection}{\thesection~~}
\titleformat{\section}{\bfseries\LARGE}{%
  \ifnum\pdfstrcmp{\@currenvir}{appendices}=0
  \setappendix
  \else
  \setsection
\fi}{0em}{}
\makeatother
\usepackage[titletoc]{appendix}

\usepackage{xcolor}

\usepackage{cite}
\usepackage[pdftex]{graphicx}
\graphicspath{./figures}
\usepackage{array}
\usepackage{mathtools}
\usepackage{amssymb,amsfonts,amsmath}
\usepackage{dsfont}
\usepackage{stmaryrd}
\usepackage{bm}
\usepackage{notations}
\usepackage{footmisc}
\DefineFNsymbols{mySymbols}{{\ensuremath\dagger}{\ensuremath\Diamond}{\ensuremath{\ast}}{\ensuremath{\star}}}
\setfnsymbol{mySymbols}
\usepackage{setspace}                    

\def \({\left(}
\def \){\right)}
\def \[{\left[}
\def \]{\right]}
\def \nn{\nonumber \\}

\newcommand{\bY}{{\bm {Y}}}

\newcommand{\bU}{{\bm {U}}}

\newcommand{\bW}{{\bm {W}}}
\newcommand{\bZ}{{\bm {Z}}}
\newcommand{\bh}{{\bm {h}}}

\newcommand{\bA}{{\bm {A}}}

\newcommand{\bX}{{\bm {X}}}
\newcommand{\bx}{{\bm {x}}}

\newcommand{\by}{{\bm {y}}}

\newcommand{\bc}{{\bm {c}}}

\newcommand{\e}{\text {e}}

\DeclareMathAlphabet{\varmathbb}{U}{bbold}{m}{n}

\newcommand{\abs}[1]{\lvert#1\rvert}
 
  \newcommand{\ep}{\epsilon}

\title{All-or-nothing statistical and computational phase transitions in sparse spiked matrix estimation}

\author{%
  Jean Barbier \\
  International Center for Theoretical Physics \\
  Strada Costiera 11, 34151 Trieste, Italy  \\
  \texttt{jbarbier@ictp.it} \\
  \And
  Nicolas Macris \\
  Ecole Polytechnique F\'ed\'erale de Lausanne \\
  CH 1015 Lausanne, Switzerland \\
  \texttt{nicolas.macris@epfl.ch} \\
  \And
  Cynthia Rush \\
  Department of Statistics, Columbia University\\
  New York, NY 10025 \\
  \texttt{cynthia.rush@columbia.edu}
}

\begin{document}

\maketitle

\begin{abstract}
We determine statistical and computational limits for estimation of a rank-one matrix (the spike) corrupted by an additive gaussian noise matrix, in  a sparse limit, where the underlying hidden vector (that constructs the rank-one matrix) has a number of non-zero components that scales sub-linearly with the total dimension of the vector, and the signal-to-noise ratio tends to infinity at an appropriate speed. We prove explicit low-dimensional variational formulas for the asymptotic mutual information between the spike and the observed noisy matrix and analyze the approximate message passing algorithm in the sparse regime. For Bernoulli and Bernoulli-Rademacher distributed vectors, and 
when the sparsity and signal strength satisfy an appropriate scaling relation, we find all-or-nothing phase transitions for the asymptotic minimum and algorithmic mean-square errors. These jump from their maximum possible value to zero, at well defined signal-to-noise thresholds whose asymptotic values we determine exactly. In the asymptotic regime the statistical-to-algorithmic gap diverges indicating that sparse recovery is hard for approximate message passing. 
\end{abstract}

\section{Introduction and setting}\label{sec:intro}

In modern machine learning and high dimensional statistics one often faces regression, classification, or estimation tasks, where the dimension of the feature vectors is much larger than the effective underlying dimensionality of the structure at hand. For example, hand-written MNIST digits are presented as vectors consisting of $28\times 28$ pixels, in other words, they are binary vectors with $784$ dimensions, whereas \cite{costa_2004_38676,AudibertHein2005} estimate their effective dimension to be in the orders of $10$'s. Similarly the ISOMAP face database consists of images ($256$ levels of gray) of size $64\times 64$, i.e., vectors in 
$\mathbb{R}^{4096}$, whereas the correct intrinsic dimension is only $3$ (for the vertical, horizontal pause and lighting direction). Natural images, which are generally sparse in a wavelet basis \cite{Mallat_book_1999}, are another popular example of low effective dimensionality. For natural images, a very simple model of low-dimensional structure, namely vectors with a sparse number of non-zero components, has proven immensely useful for studying these types of data structures and has led to the development of the whole area of compressed sensing \cite{CandesRombergTao_2006,Donoho_CompressedSensing2006}. Similarly, matrix completion can be performed successfully when the number of sampled matrix elements is much smaller than the total number of elements, as long as one assumes the matrix is low-rank \cite{candes2009exact}. 

These and other developments have amply justified the  ``bet on sparsity principle'', which, in a nutshell, says that intrinsic low-dimensionality is often a crucial ingredient for the interpretability of high dimensional statistical models  \cite{hastie_09_elements-of.statistical-learning, RishGrabarnik2014}. In this context, it is of great importance to determine computational limits of estimation and to establish fundamental information theoretical (i.e., statistical) limits as benchmarks. Broadly speaking, exact results in the direction of computational or information theoretic limits usually fall in two categories. 
The first direction, traditional in statistics and computer science,  derives  {\it finite size} bounds on thresholds marking the onset of feasible signal recovery or learning
\cite{HastieTibshiraniWainwright, wainwright2019high}. Such results usually leave out exact constants or do not always give the exact asymptotics. 
The second approach, is an {\it average case} approach (in the spirit of the statistical mechanics treatment of high dimensional systems), that models feature vectors by a {\it random ensemble}, taken as a set of random vectors with independently identically distributed (i.i.d.) components, and a small but fixed fraction of non-zero components. For example, the distribution might be a Bernoulli distribution, denoted ${\rm Ber}(\rho_n)$ with $0< \rho_n < 1$ and $\rho_n\to \rho > 0$ fixed, as the dimension of the vectors $n\to +\infty$.
In Bayesian settings with known priors and hyper-parameters this approach has been highly successful, yielding exact formulas for the mutual information and minimum mean-square error (MMSE), as well as exact expressions (with constants) for statistical and computational message passing phase transition thresholds in the limit of {\it infinite dimensions} \cite{Zdeborov2016}. While the mathematical analysis of this approach is well developed in compressed sensing, generalized linear estimation, or rank-one noisy matrix and tensor estimation \cite{barbier_allerton_RLE,9079920,private,barbier2017phase,2016arXiv161103888L,2017arXiv170200473M,XXT,BarbierM17a,BarbierMacris2019,2017arXiv170108010L,2017arXiv170910368B,mourrat2019hamilton}, the cited works all fall short of addressing the ``true'' sparse limit where $\rho_n \to 0$ instead of the limit being fixed (i.e., $\rho_n\to \rho > 0$) as $n\to + \infty$; {to be more precise we manage to tackle the regime $\rho_n=\Omega(n^{-\beta})$ for $\beta\in[0,1/6)$ for the information-theoretic analysis, and $\rho_n=\Omega((\ln n)^{-\alpha})$ for any positive  fixed $\alpha$ for the algorithmic results. The terminology ``true sparsity'' is employed in order to emphasize this contrast.} To the best of our knowledge the only works addressing this ``true'' sparse limit, in the average case approach for statistical phase transitions, are \cite{david2017high, reeves2019all} which consider linear regression. 

In this work, we address the issue of ``true'' sparsity in the average case approach for the problem of rank-one matrix estimation from noisy observations of the entries. 
Low-rank matrix estimation (or factorization) is an important problem with numerous applications in image processing, principal component analysis (PCA), machine learning, DNA microarray data, and tensor decompositions. We determine information theoretic limits of the problem as well as computational limits of an approximate message passing algorithm \cite{Kabashima_2003,bayati2011dynamics, Donoho10112009, krz12, MontChap11, Rangan11} for signal estimation in the case of a noisy symmetric rank-one matrix model. Let us now introduce the model.

{\bf Setting:} In the {\it sparse spiked Wigner matrix model} we consider a sparse signal-vector $\bX = (X_1, \ldots, X_n)\in \mathbb{R}^n$ with i.i.d.\ components distributed according to 
$P_{X,n} = \rho_n p_X+(1-\rho_n)\delta_0$,
where $\delta_0$ is the Dirac mass at zero and $\rho_n \in (0,1]^\mathbb{N}$ is a sequence of weights that will eventually tend to $0$; the signal has in expectation a {\it sub-linear} number $n\rho_n$ of non-zero components. For the distribution $p_X$ we assume that $i)$ it is independent of $n$, $ii)$ it has finite support in an interval $[-S,S]$, $iii)$ it has second moment equal to $1$ (without loss of generality). One has access to the symmetric data matrix $\bW\in \mathbb{R}^{n\times n}$ with noisy entries
\begin{align}\label{WSM}
\bW = \sqrt{\frac{\lambda_n}{n}} \bX\otimes \bX +  \bZ\,, \ \ \text{or componentwise} \ \ W_{ij} = \sqrt{\frac{\lambda_n}{n}} X_iX_j +  Z_{ij}\,,  \quad 1\le i<j\le n
\end{align}
where $\lambda_n >0$ controls the strength of the signal and the noise is i.i.d.\ gaussian $Z_{ij}\sim{\cal N}(0,1)$ for $i<j$ and symmetric $Z_{ij}=Z_{ji}$. Notice that the matrix $\bW$ can be viewed as a sum of a gaussian matrix from the Wigner ensemble perturbed by a rank-one matrix, $\bX \bX^\intercal$ (the ``spike''). We focus, in particular,
on binary $\bX$ generated with i.i.d.\ Bernoulli entries $X_i \sim P_{X,n}= {\rm Ber}(\rho_n)$, or Bernoulli-Rademacher entries, $X_i\sim P_{X,n}=(1-\rho_n)\delta_0 + \rho_n\frac12(\delta_{-1}+\delta_1)$. In the Bayesian setting, we suppose that the prior $P_{X,n}$ and hyper-parameters are known.
As we will see, when $\rho_n\to 0$, non-trivial estimation is possible only when $\lambda_n\to +\infty$.

The goal is to estimate the sparse spike $\bX \otimes \bX$ from the data $\bW$. In the spiked Wigner model with linear sparsity, a class of polynomial-time algorithms, referred to as approximate message passing or AMP, have been shown to provide Bayes-optimal signal estimation for some problem settings asymptotically as $n \rightarrow +\infty$ \cite{deshpande2015finding, deshpande2014information, montanari2017estimation}. Moreover, AMP algorithms have been applied successfully for signal recovery to a number of other low-rank matrix estimation problems\cite{vila2015hyperspectral, fletcher2018iterative, parker2014bilinear, montanari2015non} and, based on bold conjectures from the statistical physics literature, it is suggested that the estimation performance of AMP is the 
best among polynomial-time algorithms. Again, AMP is also provably optimal in some parameters regimes.
In this work, we study the properties of an AMP algorithm designed for signal estimation for the spiked Wigner matrix model in the sub-linear sparsity regime and compare its performance to benchmarks established by the information theoretic limits. This analysis provides a better understanding of the computational vs.\ theoretical gaps posed by the problem.


{\bf Some background and related work:} In recent years, there has been much progress in understanding such spiked matrix models, which have played a crucial role in the analysis of threshold phenomena in high-dimensional statistical models for almost two decades, but most of this work has focused on standard settings, by which we mean problem settings where the distribution 
$P_X$ is fixed independent of the problem dimension $n$. This means that the expected number of non-zero components of $\bX$, even if ``small'', will scale {\it linearly} with $n$. Early rigorous results found in \cite{baik2005phase} determined the location of the information theoretic phase transition point in a spiked covariance model using spectral methods, and \cite{peche2006largest,feral2007largest} did the same for the Wigner case. More recently, the information theoretic limits and those of hypothesis testing have been derived, with the additional structure of sparse vectors, for large but finite sizes \cite{amini2009,pmlr-v75-brennan18a,gamarnik2019overlap}. A lot of efforts have also been devoted to computational aspects of sparse PCA with many remarkable results \cite{deshpande2014information,Cai2015,krauthgamer2015,JMLR:v17:15-160,wang2016,pmlr-v30-Berthet13,Ma:2015:SLB:2969239.2969419,gamarnik2019overlap,pmlr-v75-brennan18a}. The picture that has emerged is that the information theoretic and computational phase transition regimes are not on the same scale and that the computational-to-statistical gap diverges in the limit of vanishing sparsity. 
However, the exact thresholds with constants as well as the behaviour of the mean-square errors remained unknown.

Using heuristic methods from the statistical physics of spin glass theory (the so-called replica method \cite{mezard2009information}), the authors of \cite{2017arXiv170100858L} observed an interesting phenomenology of the information theoretical and computational limits with sharp phase transitions as $n\to +\infty$. The rigorous mathematical theory of these phase transitions is now largely under control. On one hand, an approximate message passing algorithm for signal recovery can be rigorously analyzed via its state evolution \cite{bolthausen2014iterative,bayati2011dynamics,Donoho10112009}, and on the other hand, the asymptotic mutual information per variable between the hidden spike and data matrices has been rigorously computed in a series of works using various methods (cavity method, spatial coupling, interpolation methods, PDE techniques) \cite{korada2009exact,krzakala2016mutual,XXT,2016arXiv161103888L,2017arXiv170108010L,2017arXiv170200473M,BarbierM17a,BarbierMacris2019,2017arXiv170910368B,el2018estimation,barbier2019mutual,mourrat2019hamilton}. The information theoretic phase transitions are then signaled by singularities, as a function of the signal strength, in the limit of the mutual information per variable when $n\to +\infty$.
The phase transition also manifests itself as a jump discontinuity in the minimum mean-square error (MMSE)\footnote{This is the generic singularity and one speaks of a first order transition. In special cases the MMSE may be continuous with a higher discontinuous derivative of the mutual information.}. 
Once the mutual information is known, it is usually possible to deduce the MMSE using so-called  I-MMSE relations \cite{GuoShamaiVerdu_IMMSE,guo2011estimation}. Essentially, the MMSE can be accessed by differentiating the mutual information with respect to the signal-to-noise strength. Closed form expressions for the asymptotic mutual information therefore allow
to benchmark the fundamental information theoretical limits of estimation. We also point the reader towards the works \cite{perry2018optimality,alaoui2017finite,alaoui2018detection} which derive limits of detecting the presence of a spike in a noisy matrix, rather than estimating it.

Finally, similar phase transitions in sub-linear sparsity regimes for binary signals have been studied in the context of high-dimensional linear regression or compressed sensing for support recovery \cite{david2017high,reeves2019all}. These works focus on the MMSE and prove the occurrence of the $0-1$ phase transition, which they called an ``all-or-nothing'' phenomenon. We note that our approach is technically very different in that it determines the variational expressions for mutual informations and finds the transitions as a consequence. Moreover these works do not deal with algorithmic phase transitions, while we consider here the one of AMP.

{\bf Our contributions:} We provide new results in sparse limits along two main lines:
\begin{itemize}
\item 
The exact statistical threshold for the sharp all-or-nothing statistical transition at the level of the MMSE. This follows from a rigorous derivation of the mutual information in the form of a variational problem.
\item 
The AMP algorithmic threshold and all-or-nothing transition at the level of the AMP mean-square error. This follows from 
a ``finite sample'' analysis of the approximate message passing algorithm, allowing to rigorously track its performance in sparse regimes.
\end{itemize} 
Let us explain these contributions in detail.

In this work, we identify the correct {\it scaling regimes} of vanishing sparsity and diverging signal strength in which non-trivial information theoretic and algorithmic AMP phase transitions occur. Moreover, we determine the statistical-to-algorithmic gap in the scaling regime. These scalings, thresholds, as well as formulas for the mutual information, were first heuristically and numerically derived in \cite{2017arXiv170100858L} using the non-rigorous replica method of spin-glass theory and the state evolution equations for AMP. However, it must be stressed that, not only were these calculations far from rigorous, but more importantly the limit $n\to +\infty$ is taken first for a fixed parameter $\rho_n=\rho$, and the sparse limit $\rho\to 0_+$ is taken only after. Although the thresholds found in this way agree with our derivations, this is far from evident a priori. In contrast, our results are entirely rigorous and valid in the truly sparse limit. Therefore the picture found in \cite{2017arXiv170100858L} is fully vindicated. In addition, we also establish that the MMSE and AMP phase transitions are of the all-or-nothing type, a novelty of the present work. 

The information theoretic analysis is done via the adaptive interpolation method \cite{BarbierM17a,BarbierMacris2019,2017arXiv170910368B}, first introduced in the non-sparse matrix estimation problems, to provide for the sparse limit, closed form expressions of the mutual information in terms of low-dimensional variational expressions (theorem~\ref{thm:ws} in section \ref{sec:matrixWinfotheorresults}). That the adaptive interpolation method can be extended to the sparse limit is interesting and not a priori obvious.
Using the I-MMSE relation and the solution of the variational problems for Bernoulli and Bernoulli-Rademacher distributions of the sparse signal, we then find that the MMSE displays an all-or-nothing phase transition (corollary~\ref{Cor1:MMSEwigner}) and we determine the exact threshold (with constants). 

A useful property of AMP is that in the large system limit $n \rightarrow + \infty$, its performance can be exactly characterized and  rigorously analyzed through its so-called {\it state evolution}. When $\rho_n\to \rho>0$, the validity of the state evolution analysis for AMP for low-rank matrix estimation follows from the standard AMP theory \cite{bayati2011dynamics,Donoho10112009} (with some additional work needed to deal with technicalities relating to the algorithm's initialization \cite{montanari2017estimation}), however, in the sub-linear sparsity regime considered here, proving the validity of the state evolution characterization requires a new and non-trivial analysis using ``finite sample'' techniques, first developed in \cite{RushVenkataramanan}. We find that the algorithmic MSE, denoted ${\rm MSE}_{\rm AMP}$ displays an all-or-nothing transition as well and we determine the scaling of the threshold (the constant being obtained numerically). Interestingly, the transition is on a very different signal-to-noise scale as compared to the MMSE (theorem \ref{AMP-theorem} found in section \ref{sec:matrixWalgo}).

Let us describe in a bit more detail the sparse regimes we study and the corresponding thresholds.
To gain some intuition, we first note that for sub-linear sparsity, phase transitions can appear only if the signal strength tends to infinity. This can be seen from the following heuristic argument: 
notice that the total signal-to-noise ratio per non-zero component\footnote{In more detail, this is equal to the signal-to-noise ratio per observation $(\lambda_n/n)\rho_n^2$ times the number of observations $\Theta(n^2)$ divided by the expected number of non-zero components $\rho_n n$.}
scales as $(\lambda_n/n)\rho_n^2n^2/(\rho_n n) = \lambda_n \rho_n$, meaning that $\lambda_n \to +\infty$ is necessary in order to have enough energy to estimate the non-zero components. 
Our analysis shows that non-trivial information theoretic and AMP phase transitions occur at different scales:
\begin{itemize}
\item
{\bf Statistical phase transition regime:}
While our results are more general (see appendix A and theorem 3) our main interest is in a regime of the form
\begin{align}\label{mainregime}
\lambda_n = 4\gamma \vert \ln \rho_n\vert\rho_n^{-1}, \qquad \rho_n = \Omega(n^{-\beta}), 
\end{align}
for $\beta, \gamma \in \mathbb{R}_{\ge 0}$ and $\beta$ small enough. We prove that in this regime a phase transition occurs as function of $\gamma$. 
\item
{\bf Algorithmic AMP phase transition regime:}
We control the performance of AMP for a number of time-iterations $t = o(\frac{\ln n}{\ln\ln n})$ and rigorously prove that the all-or-nothing transition occurs for 
\begin{align}
\lambda_n = w \rho_n^{-2}, \qquad \rho_n = \Omega((\ln n)^{-\alpha}), 
\end{align}
where $w, \alpha \in \mathbb{R}_{\ge 0}$ are fixed constants (note that we can take any $\alpha >1$). Controlling the AMP iterations in this regime is already highly non-trivial, however, we conjecture that the result still holds when $\rho_n = \Omega(n^{-\beta})$ for $\beta >0$ small enough, but refining the analysis in appendix K to find the stronger result is left for future work.
%
\end{itemize}
The relation  $\lambda_n \sim \rho_n^{-2}$ for the AMP threshold was obtained in \cite{2017arXiv170100858L} based on a stability analysis of the linearized state evolution. However, we recall that in their setting $\rho_n=\rho$, $n\to +\infty$, and not only is the sparse limit $\rho\to 0_+$ taken after the high-dimensional limit, but also the AMP iterations are not controlled. In appendix G in the supplementary material we provide a simpler alternative argument that does not require linearizing the recursion.

%

We focus in particular on binary signals with $P_{X,n}$ equal to ${\rm Ber}(\rho_n)$ or Bernoulli-Rademacher $(1-\rho_n)\delta_0+\rho_n\frac12(\delta_{-1}+\delta_1)$. For these distributions we prove the existence of all-or-nothing transitions for the MMSE and ${\rm MSE}_{\rm AMP}$ {for the specific sparsity regimes stated above.} 
This is illustrated in figures \ref{fig:MMSEandMI} and \ref{fig:MMSEandMSEamp}, found in sections \ref{sec:matrixWinfotheorresults} and \ref{sec:matrixWalgo}, which display, for the Bernoulli prior, the explicit asymptotic values to which the finite $n$ mutual information and MMSE converge. The results are similar for the Bernoulli-Rademacher distribution. In figure \ref{fig:MMSEandMI},
we see that as $\rho_n\to 0_+$ the (suitably normalized) mutual information approaches the broken line with an angular point at $\lambda /\lambda_c(\rho_n) =1$ where 
$\lambda_c(\rho_n) = 4 \vert \ln\rho_n\vert /\rho_n$. Moreover the (suitably normalized) MMSE tends to its maximum possible value $1$ for $\lambda /\lambda_c(\rho_n) <1$, develops a jump discontinuity at $\lambda /\lambda_c(\rho_n)=1$, and takes the value $0$ when $\lambda /\lambda_c(\rho_n)>1$ as $\rho_n\to 0$. In figure \ref{fig:MMSEandMSEamp}, we observe the same behavior for ${\rm MSE}_{\rm AMP}$ as a function 
of $\lambda/\lambda_{\rm AMP}(\rho_n)$, but now the algorithmic threshold is $\lambda_{\rm AMP}(\rho_n) = 1/(e \rho_n^2)$, where the constant $1/e$ is approximated numerically. Note that the same asymptotic behavior is observed in the related problem of finding a small hidden community in a graph, see figure 5 in \cite{montanari2015finding}.
\section{Statistical phase transition}\label{sec:matrixWinfotheorresults}

The phase transition manifests itself as a singularity (more precisely a discontinuous first order derivative) in the mutual information
$I(\bX\otimes \bX;\bW)=H(\bW)-H(\bW|\bX\otimes \bX)$. Note that because the data $\bW$ depends on $\bX$ only through $\bX\otimes \bX$ we have $H(\bW|\bX\otimes \bX)=H(\bW|\bX)$ and therefore $I(\bX\otimes \bX; \bW)=I(\bX; \bW)$. From now on we use the form $I(\bX; \bW)$. 

To state the result, we define the {\it potential function}:
\begin{align}\label{26}
i_n^{\rm pot}(q, \lambda,\rho) \equiv \frac{\lambda}{4} (q-\rho)^2+I_n(X;\sqrt{\lambda q}X+Z)\, ,
\end{align}
where $I_n(X;\sqrt{\lambda q}X+Z)$ is the mutual information for a scalar gaussian channel, with $X\sim P_{X,n}$ and $Z\sim{\cal N}(0,1)$. The mutual information $I_n$ is indexed by $n$ because of its dependence on $P_{X,n}$.
\begin{thm}[Mutual information for the sparse spiked Wigner model]\label{thm:ws}
Let the sequences $\lambda_n$ and $\rho_n$ verify \eqref{mainregime} with $\beta\in [0, 1/6)$ and $\gamma >0$. There exists $C>0$ independent of $n$ such that 
\begin{align}\label{mainbound}
\frac{1}{\rho_n|\ln\rho_n|}\Big|\frac{1}{n}I(\bX;\bW) - \inf_{q\in [0,\rho_n]} i^{\rm pot}_n(q,\lambda_n,\rho_n)\Big| 
\le C\frac{(\ln n)^{1/3}}{n^{(1-6\beta)/7}}\,.
\end{align}	
\end{thm}

\begin{figure}[t!]
\centering
\includegraphics[trim={0 0 0 1.2cm},clip,width=0.45\linewidth]{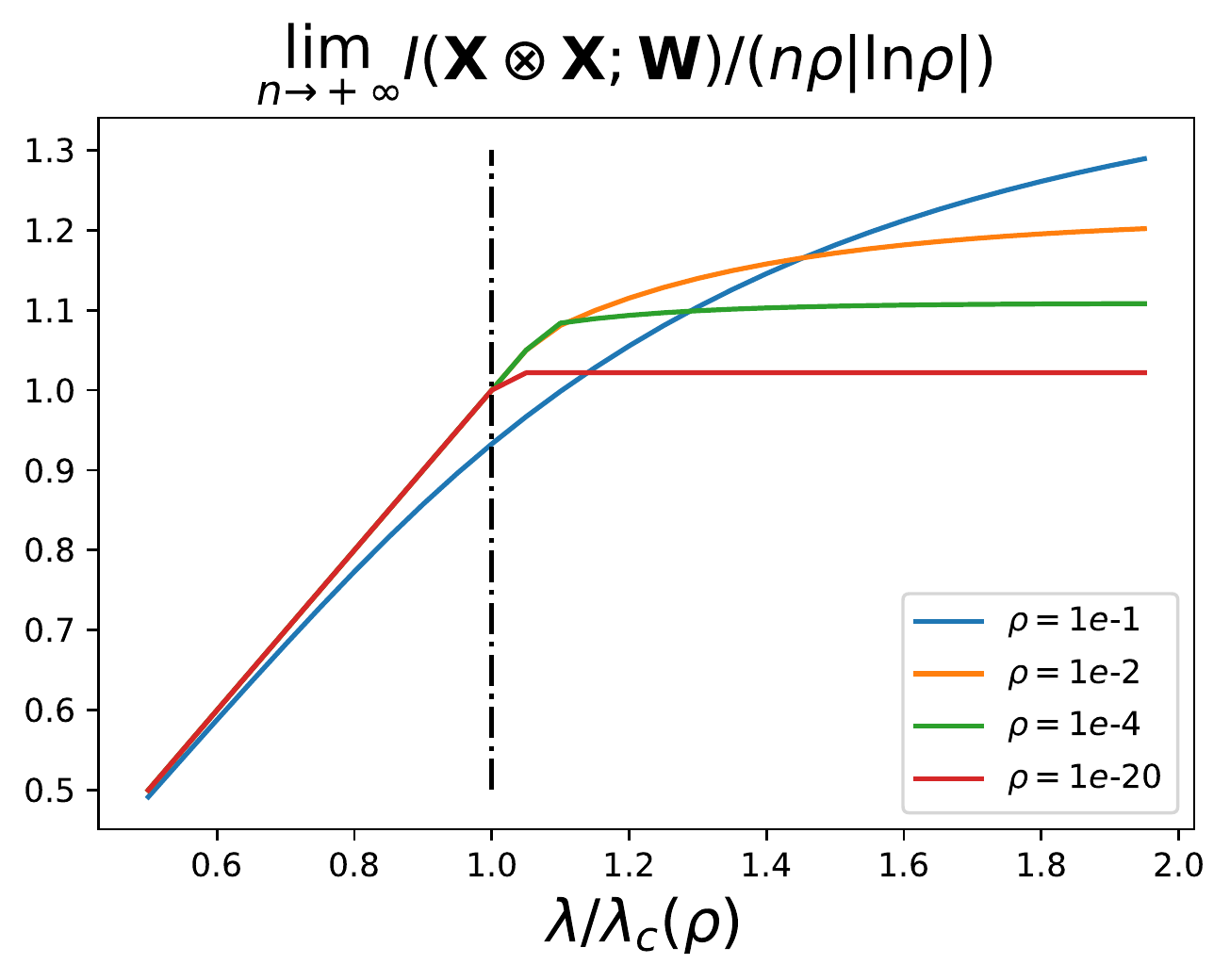}
\includegraphics[trim={0 0 0 1.2cm},clip,width=0.45\linewidth]{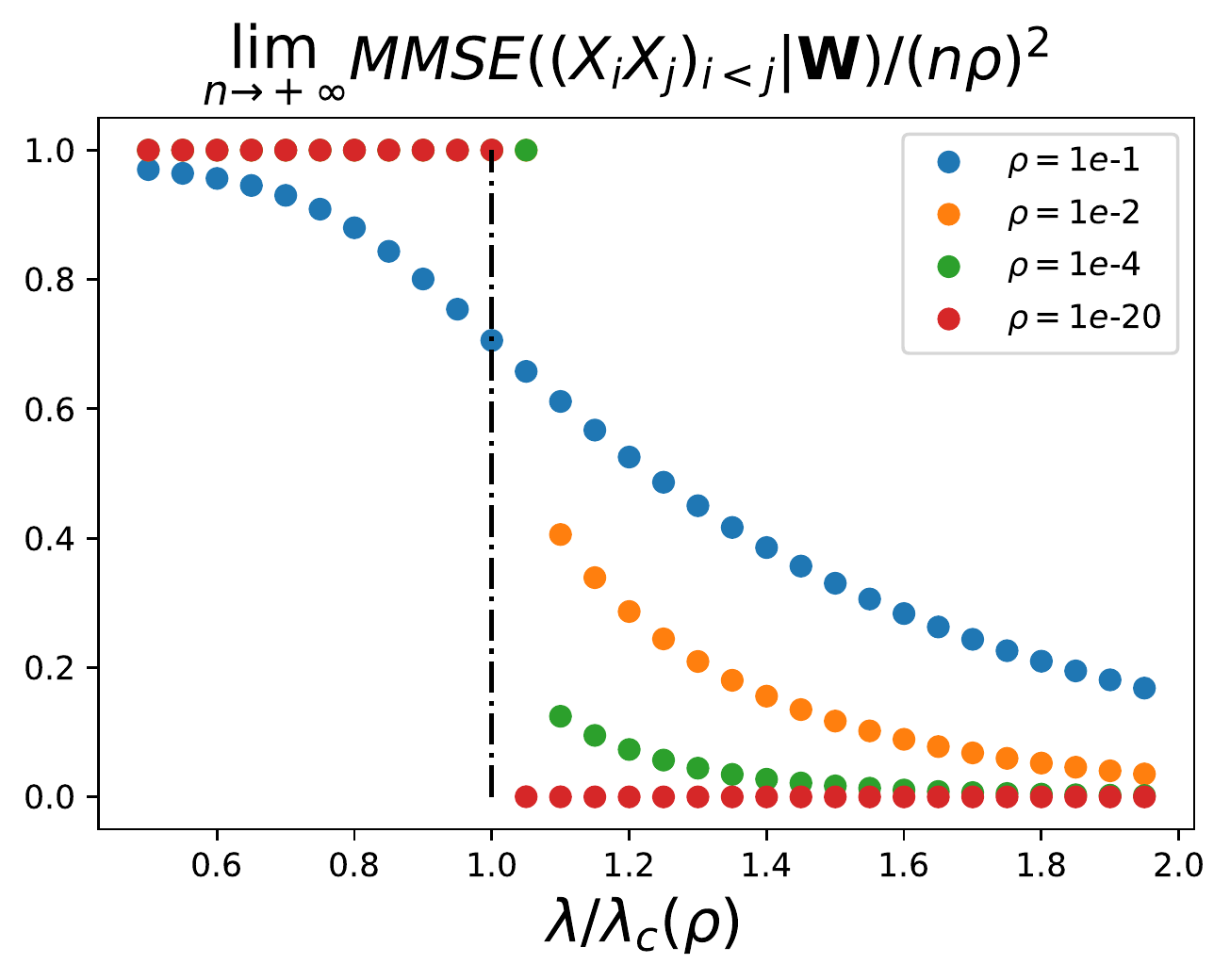}
 \caption{\footnotesize A sequence of suitably normalized asymptotic mutual information $(\rho |\ln \rho|)^{-1}\inf_{q\in [0,\rho]} i^{\rm pot}_n(q,\lambda,\rho)$ (left) and associated minimum mean-square error (MMSE) $\rho^{-2} \frac{d}{d\lambda}\inf_{q\in [0,\rho]} i^{\rm pot}_n(q,\lambda,\rho)$ (right) curves as a function of $\lambda/\lambda_c(\rho)$ with $\lambda_c(\rho) = 4\vert \ln \rho\vert/\rho$ for the model $X_i \sim {\rm Ber}(\rho)$ and various $\rho=\rho_n$ values (that can be converted to signal sizes through $\rho_n=\Omega(n^{-\beta})$ given a sparsity scaling $\beta$) using the potential function defined in \eqref{26}. These are the curves towards which, respectively, the finite size mutual information $(n\rho|\ln\rho|)^{-1}I(\bX;\bW)$ and minimum mean-square error $(n\rho)^{-2}{\rm MMSE}((X_iX_j)_{i<j}|\bW)$, converge: see theorem~\ref{thm:ws} and corollary~\ref{Cor1:MMSEwigner}. In the sparse limit $\rho\to 0$, the MMSE curves approach a $0$--$1$ phase transition with the discontinuity at $\lambda = \lambda_c(\rho)$. This corresponds to an angular point for the mutual information (by the I-MMSE relation).}
\label{fig:MMSEandMI}
\end{figure}

%
The mutual information is thus given, to leading order, by a one-dimensional variational problem 
$$
I(\bX;\bW) = n\rho_n \vert \ln \rho_n \vert \inf_{q\in [0,\rho_n]} i^{\rm pot}_n(q,\lambda_n,\rho_n) + {\rm correction\,\,terms}\,.
$$
The factor $\rho_n|\ln \rho_n|$ is related to the entropy (in nats) of the support of the signal given by
$- n (\rho_n\ln\rho_n + (1-\rho_n)\ln(1-\rho_n)),$ which behaves like $n \rho_n\vert \ln\rho_n\vert$ for $\rho_n\to 0_+$. In particular, for both the {\it Bernoulli} and {\it Bernoulli-Rademacher} distributions an analytical solution of the variational problem, given in appendix F, shows that $(\rho_n |\ln \rho_n|)^{-1} \inf_{q\in [0,\rho_n]} i^{\rm pot}_n(q,\lambda_n,\rho_n)$ tends to the singular function $\gamma \mathbb{I}(\gamma \le 1)+\mathbb{I}(\gamma \geq 1)$
as $n\to +\infty$ and $\rho_n\to 0$, where we recall that $\lambda_n = 4\gamma \vert \ln \rho_n\vert\rho_n^{-1}$. See figure~\ref{fig:MMSEandMI}. {Let us mention that $\beta<1/6$ is probably not a fundamental limit to the validity of the result but rather is an artefact of the sub-optimality of our proof technique.}

We now turn to the consequences for the MMSE. It is convenient to work with the ``matrix'' MMSE defined as
${\rm MMSE}((X_iX_j)_{i<j}|\bW) \equiv \mathbb{E} \Vert (X_iX_j)_{i<j} - \mathbb{E}[(X_iX_j)_{i<j} \vert \bW]\Vert_{\rm F}^2$. 
This quantity satisfies the I-MMSE relation \cite{GuoShamaiVerdu_IMMSE, guo2011estimation} (see also appendix I for a self-contained derivation),
$$
\frac{d}{d\lambda_n} \frac{1}{n} I(\bX; \bW) = \frac{1}{2n^2}{\rm MMSE}((X_iX_j)_{i<j}|\bW)\,.
$$
In appendix J we prove:
\begin{corollary}[Minimum mean-square error for the sparse spiked Wigner model]\label{Cor1:MMSEwigner} 
Let $\frac{1}{2}{m}_n(\lambda,\rho_n)\equiv \rho_n^{-2} \frac{d}{d\lambda}\inf_{q\in [0,\rho_n]} i^{\rm pot}_n(q,\lambda,\rho_n)$. Let $\epsilon > 0$ and sequences $\lambda_n$ and $\rho_n$ verifying \eqref{mainregime} with $\beta\in [0, 1/13)$. There exists $C'>0$ independent of $n$ such that
\begin{align*}
 {m}_n(\lambda_n+\epsilon,\rho_n) \! -\ \! \frac{C^\prime}{\epsilon} \frac{(\ln n)^{4/3}}{n^{(1-13\beta)/7}}
 \le 
 \frac{{\rm MMSE}((X_iX_j)_{i<j}\vert \bW)}{(n\rho_n)^2} 
 \le{m}_n(\lambda_n-\epsilon,\rho_n) \! + \! \frac{C^\prime}{\epsilon} \frac{(\ln n)^{4/3}}{n^{(1-13\beta)/7}}.
\end{align*}
\end{corollary}
Concretely the derivative $(d/d{\lambda_n})\inf_{q\in [0,\rho_n]} i^{\rm pot}_n(q,\lambda_n,\rho_n)$ is computed, using the envelope theorem \cite{milgrom2002envelope}, as $(\partial/\partial{\lambda_n}) i^{\rm pot}_n(q_n^*,\lambda_n,\rho)$ where $q_n^*=q_n^*(\lambda_n,\rho_n)$ is the solution of the variational problem, which is unique almost everywhere (except at the phase transition point, see e.g.\ \cite{barbier2017phase} for such proofs). For Bernoulli and Bernoulli-Rademacher distributions, we easily compute the limiting behavior $m_n(\lambda_n,\rho_n)$ from the solution of the variational problem stated above, and find that $(n\rho_n)^{-2}{\rm MMSE}((X_iX_j)_{i<j}\vert \bW)$ tends to $\mathbb{I}(\gamma \le 1)$ as $n\to +\infty$.

Figure \ref{fig:MMSEandMI} shows the mutual information and MMSE computed from the numerical solution of the variational problem for a sequence of ${\rm Ber}(\rho_n)$ distributions. We check that the limiting curves are indeed approached as $\rho_n \to 0$ and, in particular, the suitably rescaled MMSE displays the all-or-nothing transition at $\lambda / \lambda_c(\rho_n) =1$ as $n\to +\infty$ with $\lambda_c(\rho_n) = 4\vert\ln\rho_n\vert/ \rho_n$. For the Bernoulli-Rademacher distribution the transition location is the same, suggesting that the hardness of the inference is only related, for discrete priors, to the recovery of the support. For more generic distributions than these two cases the situation is richer.
Although one generically observes phase transitions {\it in the same scaling regime}, the limiting curves appear to be more complicated than the simple staircase shape and the jumps are not necessarily located at $\gamma=1$. A classification of these transitions is an interesting problem that is out of the scope of this paper.
\section{AMP algorithmic phase transition}\label{sec:matrixWalgo}

Approximate message passing (AMP) is a low complexity algorithm that  iteratively updates estimates of the unknown signal, which, in the case of the spiked Wigner model is $\bX$, from the noisy data $\bW$. The iterative estimates are denoted $\{\bx^t\}_{t \geq 1}$.  Let $\bA \equiv \bW/\sqrt{n}$ and initialize  with $f_0(\bx^{0})$ independent of $\bW$, such that $\langle f_0(\bx^{0}), \bX \rangle > 0$.  Then let $\bx^{1} = \bA  f_0(\bx^{0})$, and for $t \geq 1$, compute
\begin{equation}
\bx^{t+1} = \bA  f_t(\bx^{t}) - \mathsf{b}_t  f_{t-1}(\bx^{t-1}), \qquad \mathsf{b}_t  = \frac{1}{n} \sum_{i=1}^n f'_t(x^{t}_i),
\label{eq:AMP}
\end{equation}
where the scalar function $f_t: \mathbb{R} \rightarrow \mathbb{R}$ is applied elementwise to vector input, i.e., $f_t(\bx) = (f_t(x_1), \ldots, f_t(x_n))$ for a vector $\bx \in \mathbb{R}^n$, and its exact value is given in what follows (in \eqref{eq:denoiser}). We refer to the functions $\{f_t\}_{t\geq 0}$ as ``denoisers'', for reasons that will become clear momentarily. Notice that \eqref{eq:AMP} gives both matrix estimates $\hat{\bX} \hat{\bX}^\intercal = f_t(\bx^t) [f_t(\bx^t)]^\intercal$ and signal estimates $\hat{\bX} = f_t(\bx^t)$.

A key property of AMP is that, asymptotically as $n \rightarrow \infty$, a deterministic, scalar recursion referred to as \emph{state evolution} exactly characterizes its performance, in the sense that the estimates $x^t_i$ converge to random variables with mean and variance governed by the state evolution.  For the sub-linear sparsity regime, we introduce an $n$-dependent state evolution, reflecting that our sparsity level $\rho_n$ and signal strength $\lambda_n$ both now change as $n$ grows. We will show, based on measure concentration arguments, that the usual asymptotic characterization also gives a finite sample approximation, meaning that for any $n$ fixed but large, $x^t_i$ is \emph{approximately} distributed as a $x^t_i \overset{d}{\approx} \mu_t^n X_0^n + \sqrt{\tau^n_t} Z$ where $\mu_t^n$ and $\tau^n_t$ are characterized by the state evolution below with $X_0^n \sim P_{X, n}$ independent of standard gaussian $Z$.
The $n$-dependent state evolution is defined as follows: for $t \geq 1$,
\begin{align}
\mu^n_1 =  \sqrt{\lambda_n} \langle f_0(\bx^0), \bX \rangle/n, \qquad & \qquad \tau^n_1 = \| f_0(\bx^0) \|^2/n,
\label{eq:state_evolution_start} \\
\mu^n_{t+1}  = \sqrt{\lambda_n}\, \mathbb{E}\left\{X_0^n f_t\big(\mu^n_{t} X_0^n + \sqrt{\tau^n_t} Z\big)\right\}, \qquad & \qquad\tau^n_{t+1}  =  \mathbb{E}\Big\{\big[f_t\big(\mu^n_{t} X_0^n + \sqrt{\tau^n_t} Z\big)\big]^2\Big\},
\label{eq:state_evolution}
\end{align}
where we include the $n$ superscript to emphasize the dependence.  

A well-motivated choice of denoiser functions $\{f_t\}_{t\geq 0}$ are the conditional expectation denoisers. Namely,  given that we have knowledge of the prior distribution of the signal elements, and considering the approximate characterization of the estimate $x^t_i$ via the state evolution, the Bayes-optimal way to update our signal estimate at any iteration is the following: 
for $t \geq 1$,
\begin{equation}
 f_t(x) = \mathbb{E}\big \{X_0^n  \mid  \mu_t^n X_0^n + \sqrt{\tau^n_t} Z = x\big\},
\label{eq:denoiser}
\end{equation}
 with $X_0^n \sim P_{X, n}$ independent of standard gaussian $Z$. Strictly speaking, $ f_t(\cdot)$ also has an $n$-dependency, so to be consistent we should label $ f_t(\cdot) \equiv  f^n_t(\cdot)$, however we drop this for simplicity.  With this choice of denoiser function, the state evolution \eqref{eq:state_evolution} simplifies: by the Law of Total Expectation, $ \mathbb{E}\{X_0^n f_t(\mu^n_{t} X_0^n + \sqrt{\tau^n_t} Z)\} = \mathbb{E}\{[f_t(\mu^n_{t} X_0^n + \sqrt{\tau^n_t} Z)]^2\},$
 thus $ \mu_t^n = \sqrt{\lambda_n} \tau^n_{t}$, and so
\begin{equation}
\label{eq:state_evolution2}
\tau^n_{t+1} = \mathbb{E}\Big\{\big[ \mathbb{E}\big \{X_0^n \mid  \sqrt{\lambda_n} \tau^n_{t} X_0^n + \sqrt{\tau^n_t} Z\big\}\big]^2\Big\}.
\end{equation}

\begin{figure}[t!]
\centering
\includegraphics[trim={0 0 0 1.2cm},clip,width=0.45\linewidth]{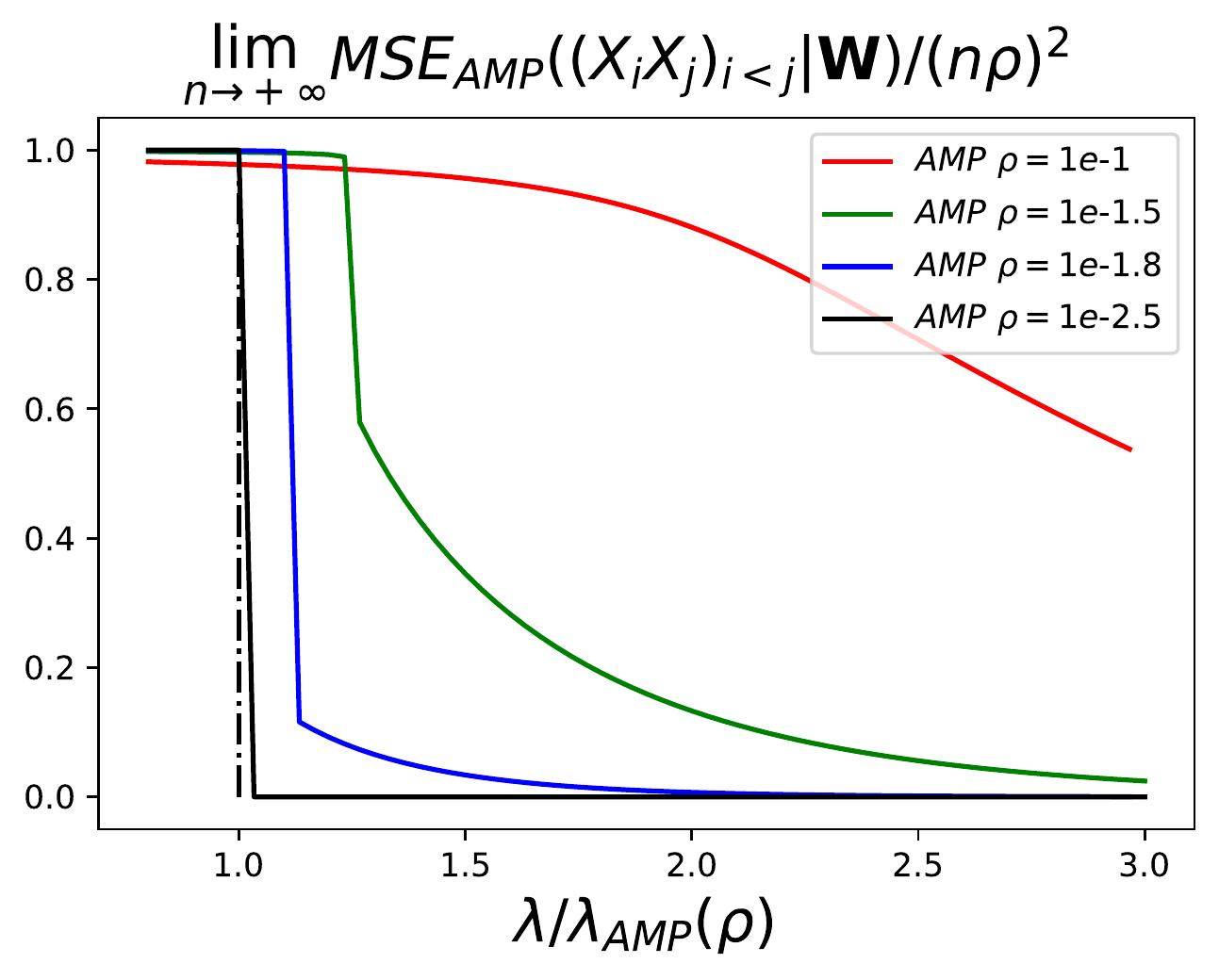}
\includegraphics[trim={0 0 0 1.2cm},clip,width=0.45\linewidth]{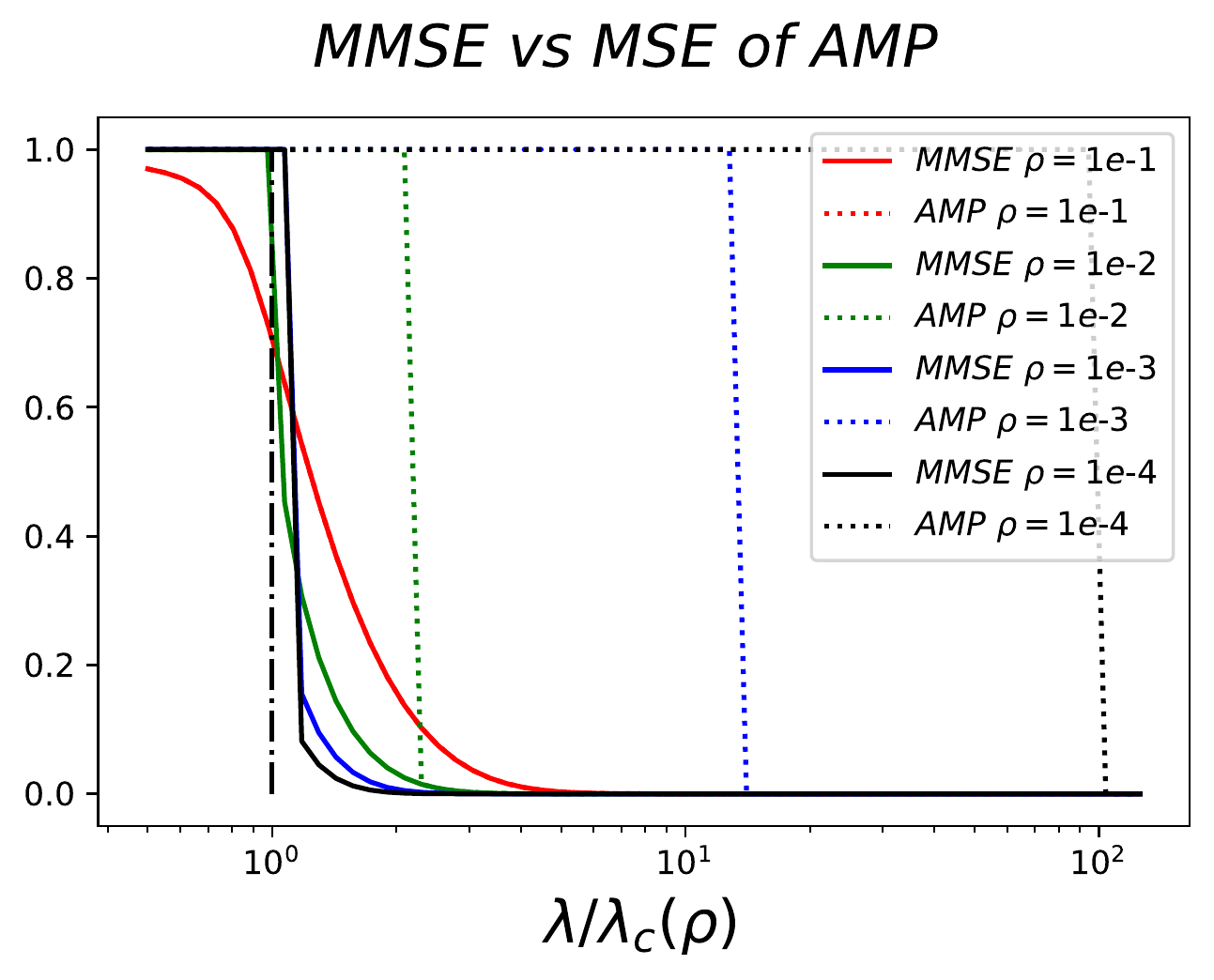}
 \caption{\footnotesize Left: The mean-square error towards which the suitably normalized matrix-MSE of the AMP algorithm, ${\rm MSE}_{\rm AMP}$, converges 
for various sparsity levels, see theorem~\ref{AMP-theorem}. An all-or-nothing transition appears as $\rho=\rho_n\to0$ at $\lambda_{\rm AMP}(\rho)=1/(e \rho)^2$. Comparing to figure~\ref{fig:MMSEandMI} the transition becomes sharper much faster as $\rho$ decreases. Right: Horizontal axis is on a log scale. The statistical-to-algorithmic gap diverges as $\rho\to 0$.}
\label{fig:MMSEandMSEamp}
\end{figure}

The performance guarantees given by the state evolution are stated informally in what follows, with a more formal result given in appendix K. The proof 
extends and refines\footnote{The result in  \cite{RushVenkataramanan} is a general AMP algorithm with a  ``rectangular'' structure that does not cover the ``symmetric'' AMP in \eqref{eq:AMP}. However, extensions of this result to the symmetric case are straightforward, as discussed in \cite[Section 1]{RushVenkataramanan}, but technical. Moreover, the dependence on $n$ for the state evolution requires that these values are tracked carefully through the proof, whereas this was not done in \cite{RushVenkataramanan}, as these values were assumed to be universal constants.  For simplicity of exposition in this document, we do not elaborate further on these technicalities at this point and put these details in appendix K.} the finite sample analysis of AMP given in \cite[Theorem 1]{RushVenkataramanan}. These guarantees concern the convergence of the empirical distribution of $x^t_i$ to its approximating distribution determined by the state evolution and specifically apply to the AMP algorithm using the denoiser in \eqref{eq:denoiser}.  For all order $2$ pseudo-Lipschitz functions\footnote{For any $n,m \in \mathbb{N}_{>0}$, a function $\phi : \mathbb{R}^n\to \mathbb{R}^m$ is \emph{pseudo-Lipschitz of order $2$} if there exists a constant $L>0$ such that $\left|\left|\phi(\bx)-\phi(\by)\right|\right|\leq L\left(1+ \|\bx\|+ \|\by\|\right) \|\bx-\by\|$ for $\bx, \by \in \mathbb{R}^n$.}, denoted $\psi: \mathbb{R}^2 \rightarrow \mathbb{R}$ with Lipschitz constant $L_{\psi}> 0$, we have that for $\epsilon \in (0,1)$ and $t \geq 1$, 
 \begin{equation}
 \begin{split}
&\mathbb{P} \Big(\Big| \frac{1}{n} \sum_{i=1}^n \psi(X_i, f_{t}(x^{t}_i))\! -\!  \mathbb{E}\Big\{\psi\big(X_0^n,  f_{t}(\mu^n_{t} X_0^n + \sqrt{\tau^n_t} Z)\big) \Big\} \Big| \geq \epsilon \Big)   \leq C C_{t} \exp\Big\{\frac{-c c_{t} n \epsilon^2}{L_{\psi}^2  \gamma_n^{t}}\Big\}
 \label{eq:finite_sample}
 \end{split}
 \end{equation}
 where $\bX = (X_1, \ldots, X_n)$ is the true signal and $C, C_t, c, c_t$ are universal constants not depending on $n$ or $\epsilon$, but with $C_t, c_t$ depending on the iteration $t$ and whose exact value is given in theorem~\ref{AMP-theorem}. Finally, $\gamma_n^{t}$ characterizes the way the bound depends on the state evolution parameters and its exact value is given in \eqref{eq:gamma_def}.  
We want to consider, specifically, the vector-MSE and matrix-MSE of AMP, namely $\frac{1}{n}\|\bX -  f_t(\bx^t)\|^2$ and $\frac{1}{n^2}\|\bX \bX^\intercal -   f_t(\bx^t)[ f_t(\bx^t)]^\intercal\|_F^2$,  for any $t \geq 1$.
 \begin{thm}[Finite sample state evolution]\label{AMP-theorem}
Consider  AMP in \eqref{eq:AMP} using the conditional expectation denoiser in \eqref{eq:denoiser}. Then for $\epsilon \in (0,1)$ and $t \geq 1,$ let $\textsf{bound}_t \equiv C C_{t} \exp\{ {-c c_{t} n \epsilon^2}/{  \gamma_n^{t}}\},$ then
 \begin{align}
&\mathbb{P} \Big(\Big| \frac{1}{n}\|\bX -  f_{t}(\bx^{t})\|^2-  (\rho_n -  \tau^n_{t+1})\Big| \geq \epsilon \Big) \leq \textsf{bound}_t ,  \label{eq:finite_sample_vector1}\\
&\mathbb{P} \Big(\Big| \frac{1}{n^2}\|\bX \bX^\intercal -  f_{t}(\bx^{t}) [ f_{t}(\bx^{t})]^\intercal\|_F^2 -  ( \rho_n^2 -  (  \tau^n_{t+1})^2)\Big| \geq \epsilon \Big) \leq \textsf{bound}_t ,  \label{eq:finite_sample_matrix1}
 \end{align}
 where $X_0^n \sim P_{X, n}$ and $\tau_t^n$ is defined in \eqref{eq:state_evolution2}.  The values $C, c$ are universal constants not depending on $n$ or $\epsilon$ with $C_t, c_t$ given by
$C_t =C_1^{t} (t!)^{C_2},  c_t = [c_1^{t} (t!)^{c_2}]^{-1}$. Finally, 
\begin{equation}
\begin{split}
\gamma_n^{t} &\equiv \lambda_n^{2t-1}  (\nu^n +  \tau^n_{1})  (\nu^n +  \tau^n_{1} +  \tau^n_{2})  \cdots  (\nu^n + \sum_{i=1}^{t} \tau^n_{i}) \\
&\hspace{4cm} \times \max\{1, \hat{\textsf{b}}_1\}  \max\{1, \hat{\textsf{b}}_2\} \cdots  \max\{1, \hat{\textsf{b}}_{t-1}\},
\label{eq:gamma_def}
\end{split}
\end{equation}
where $\nu^n$ is the variance factor of sub-Gaussian $\bX^n$ which equals $12 \rho_n$ for $P_{X, n}={\rm Ber}(\rho_n)$ (see lemma~\ref{lem:subgauss}) and $\hat{\textsf{b}}_t = \mathbb{E}\{f'_t(\mu^n_t X^n_0 + \sqrt{\tau_t^n} Z)\}$.
 \end{thm}
Theorem~\ref{AMP-theorem} follows from the finite sample guarantees given in \eqref{eq:finite_sample}, and, in appendix K, we discuss in more detail the proof of theorem~\ref{AMP-theorem} and result~\ref{eq:finite_sample}. We make a few remarks on the result here.

\textbf{Remark 1: $\rho_n$ normalization and all-or-nothing transition.}
To be consistent with the previously stated results, we could renormalize the MSEs as follows and the result still holds as
\begin{align*}
&\mathbb{P} \Big(\Big| \frac{1}{\rho_n n}\|\bX -  f_{t}(\bx^{t})\|^2-  \Big(1 - \frac{\tau^n_{t+1}}{ \rho_n} \Big)\Big| \geq \epsilon \Big) \leq C C_{t} \exp\{ {-c c_{t} n \rho_n^2 \epsilon^2}/{  \gamma_n^{t}}\}, \\
&\mathbb{P} \Big(\Big| \frac{1}{(\rho_n n)^2}\|\bX \bX^\intercal -  f_{t}(\bx^{t}) [ f_{t}(\bx^{t})]^\intercal\|_F^2 -  \Big(1 -  \Big(\frac{\tau^n_{t+1}}{ \rho_n}\Big)^2\Big)\Big| \geq \epsilon \Big) \leq C C_{t} \exp\{ {-c c_{t} n \rho_n^4 \epsilon^2}/{  \gamma_n^{t}}\}.
 \end{align*}
 In appendix G we show that $\tau^n_{t+1}/\rho_n \to 0$ for $\lambda_n\rho_n^{2} \to 0$ and $\tau^n_{t+1}/\rho_n \to 1$ for $\lambda_n\rho_n^{2} \to +\infty$. This is consistent with the numerics on figure~\ref{fig:MMSEandMSEamp} where we see a transition for 
$\lambda_n \rho_n^{2} = 1/e^2$. 

\textbf{Remark 2: AMP regime and statistical-to-algorithmic gap.}
We apply theorem~\ref{AMP-theorem} in the regime where  $t = o( \frac{\ln n}{\ln\ln n})$, which, as discussed in \cite{RushVenkataramanan}, is the regime where the state evolution predictions are meaningful with respect to the values of $C_t, c_t$ and the constraints they specify on how large $t$ can be compared to the dimension $n$.  In our work, we also have constraints related to the $ \gamma_n^{t}$ value in \eqref{eq:gamma_def} that appears in the denominator of the rate of concentration. Considering these constraints, we apply theorem~\ref{AMP-theorem} for signal strength and sparsity scaling like $\lambda_n \rho_n^2= w$ 
and $\rho_n = \Omega((\ln n)^{-\alpha})$ with $w, \alpha\in \mathbb{R}_{+}$, and  show that the above probabilities indeed tend to zero as $n\to +\infty$. Appendix L provides the details of this calculation. 

{Note that since theorem~\ref{thm:ws} and corollary~\ref{Cor1:MMSEwigner} hold for $\rho_n=\Omega(n^{-\beta})$ and thus for $\rho_n = \Omega((\ln n)^{-\alpha})$ as well, then both the statistical and algorithmic transitions (and therefore the statistical-to-computational gap) are proven for $\rho_n=\Omega((\ln n)^{-\alpha})$.}

\textbf{Remark 3: $\lambda_n, \tau^n$ dependence.}
The $\lambda_n$ dependence in $\gamma_n^t$ defined in \eqref{eq:gamma_def} comes from the (pseudo-) Lipschitz constants $L_f$ in~\eqref{eq:finite_sample}. The dependence on the Lipschitz constants, and on the state evolution parameters $\tau^n_t$, was not stated explicitly in the original concentration bound in \cite[Theorem~1]{RushVenkataramanan} as the authors assume these values do not change with $n$ and, thus, can be absorbed into the universal constants. 
By examining the proof of \cite[Theorem 1]{RushVenkataramanan}, one gets that the dependence takes the form in \eqref{eq:gamma_def}. More details on how we arrive at the rates in theorem~\ref{AMP-theorem} can be found in appendix~K.
%

\textbf{Remark 4: Algorithm initialization.} We assume that the AMP algorithm in \eqref{eq:AMP} was initialized  with $f_0(\bx^{0})$ independent of $\bW$ such that $\langle f_0(\bx^{0}), \bX \rangle > 0$.  The second condition ensures that $\mu^n_0 \neq 0$ (which would mean $\mu^n_t = 0$ for all $t \geq 0$). If $P_{X, n}$ is ${\rm Ber}(\rho_n)$, one could use, for example, $f_0(\bx^{0}) = \mathbf{1}$, since the mean of the signal elements is positive. However, if $P_{X, n}$ is Bernoulli-Rademacher, a more complicated initialization procedure is needed since initializing in this way would cause the algorithm to get stuck in an unstable fixed point.
We refer the reader to \cite{montanari2017estimation} for a discussion of an appropriate spectral initialization for this setting. However, such an initialization violates the assumption of independence with $\bW$. The theoretical idea in \cite{montanari2017estimation} that allows one to get around this dependence is to analyze AMP  in \eqref{eq:AMP} with a matrix $\widetilde{\bA}$ that is an \emph{approximate} representation of the conditional distribution of $\bA$ given the initialization, and then to show that with high probability the two algorithms will be close each other.  We believe that  incorporating these ideas with the finite sample guarantee in \eqref{eq:finite_sample} would be straightforward, and theorem~\ref{AMP-theorem} could be extended to the setting of AMP with a spectral initialization. 


%
\section*{Broader impact}

One cannot underestimate the relevance of sparse estimation in modern technology, and although this work is valid within the limits of a theoretical model, it participates towards better fundamental understanding of necessary resources in terms of energy and quantity of data when this data is sparse. Besides radical transitions in behaviour under small changes of control parameters, we also show that an estimation task can become computationally hard or impossible, even with (practically) unbounded signal strengths. Broadly speaking, such results provide guidelines for better design and less wasteful engineering systems.

%

\section*{Acknowledgments}
J.B.\ acknowledges discussions with Galen Reeves during his visit of Duke University. C.R.\ acknowledges support from NSF CCF \#1849883 and N.M. from Swiss National Foundation for Science grant number 200021E 17554.

\bibliographystyle{unsrt_abbvr}      
\bibliography{refs} 

\begin{thebibliography}{10}

\bibitem{costa_2004_38676}
J.~Costa and A.~Hero.
\newblock Learning intrinsic dimension and intrinsic entropy of
  high-dimensional datasets.
\newblock In {\em European Signal Processing Conference (EUSIPCO), Vienna,
  Austria, 2004}. Zenodo, sep 2004.

\bibitem{AudibertHein2005}
M.~Hein and J.-Y. Audibert.
\newblock Intrinsic dimensionality estimation of submanifolds in rd.

\bibitem{Mallat_book_1999}
S.~Mallat.
\newblock {\em A Wavelet Tour of Signal Processing (Third Edition): The Sparse
  Way}.
\newblock Academic Press, Boston, third edition, 2009.

\bibitem{CandesRombergTao_2006}
E.~J. Cand\`es, J.~K. Romberg, and T.~Tao.
\newblock Stable signal recovery from incomplete and inaccurate measurements.
\newblock {\em Communications on Pure and Applied Mathematics},
  59(8):1207--1223, 2006.

\bibitem{Donoho_CompressedSensing2006}
D.~L. {Donoho}.
\newblock Compressed sensing.
\newblock {\em IEEE Transactions on Information Theory}, 52(4):1289--1306,
  2006.

\bibitem{candes2009exact}
E.~Cand{\`e}s and B.~Recht.
\newblock Exact matrix completion via convex optimization.
\newblock {\em Foundations of Computational mathematics}, 9(6):717--772, 2009.

\bibitem{hastie_09_elements-of.statistical-learning}
T.~Hastie, R.~Tibshirani, and J.~Friedman.
\newblock {\em The elements of statistical learning: data mining, inference and
  prediction}.
\newblock Springer, 2 edition, 2009.

\bibitem{RishGrabarnik2014}
I.~Rish and G.~Grabarnik.
\newblock {\em Sparse Modeling: Theory, Algorithms, and Applications}.
\newblock CRC Press, Inc., USA, 1st edition, 2014.

\bibitem{HastieTibshiraniWainwright}
T.~Hastie, R.~Tibshirani, and M.~Wainwright.
\newblock {\em Statistical Learning with Sparsity: The Lasso and
  Generalizations}.
\newblock Chapman and Hall/CRC, 2015.

\bibitem{wainwright2019high}
M.~Wainwright.
\newblock {\em High-Dimensional Statistics: A Non-Asymptotic Viewpoint}.
\newblock Cambridge Series in Statistical and Probabilistic Mathematics.
  Cambridge University Press, 2019.

\bibitem{Zdeborov2016}
L.~Zdeborov{\'a} and F.~Krzakala.
\newblock Statistical physics of inference: thresholds and algorithms.
\newblock {\em Advances in Physics}, 65(5), Aug 2016.

\bibitem{barbier_allerton_RLE}
J.~Barbier, M.~Dia, N.~Macris, and F.~Krzakala.
\newblock {The Mutual Information in Random Linear Estimation}.
\newblock In {\em 54th Annual Allerton Conference on Communication, Control,
  and Computing}, September 2016.

\bibitem{9079920}
J.~{Barbier}, N.~{Macris}, M.~{Dia}, and F.~{Krzakala}.
\newblock Mutual information and optimality of approximate message-passing in
  random linear estimation.
\newblock {\em IEEE Transactions on Information Theory}, 2020.

\bibitem{private}
G.~Reeves and H.~D. Pfister.
\newblock The replica-symmetric prediction for compressed sensing with gaussian
  matrices is exact.
\newblock In {\em 2016 IEEE International Symposium on Information Theory
  (ISIT)}, July 2016.

\bibitem{barbier2017phase}
J.~Barbier, F.~Krzakala, N.~Macris, L.~Miolane, and L.~Zdeborov{\'a}.
\newblock Optimal errors and phase transitions in high-dimensional generalized
  linear models.
\newblock {\em Proceedings of the National Academy of Sciences},
  116(12):5451--5460, 2019.

\bibitem{2016arXiv161103888L}
M.~Lelarge and L.~Miolane.
\newblock Fundamental limits of symmetric low-rank matrix estimation.
\newblock {\em Probability Theory and Related Fields}, 173(3-4):859--929, 2018.

\bibitem{2017arXiv170200473M}
L.~{Miolane}.
\newblock {Fundamental limits of low-rank matrix estimation: The non-symmetric
  case}.
\newblock {\em ArXiv e-prints}, February 2017.

\bibitem{XXT}
J.~Barbier, M.~Dia, N.~Macris, F.~Krzakala, T.~Lesieur, and L.~Zdeborov\'{a}.
\newblock Mutual information for symmetric rank-one matrix estimation: A proof
  of the replica formula.
\newblock In {\em Advances in Neural Information Processing Systems (NIPS) 29},
  pages 424--432. 2016.

\bibitem{BarbierM17a}
J.~Barbier and N.~Macris.
\newblock The adaptive interpolation method: a simple scheme to prove replica
  formulas in bayesian inference.
\newblock {\em Probability Theory and Related Fields}, Oct 2018.

\bibitem{BarbierMacris2019}
J.~Barbier and N.~Macris.
\newblock The adaptive interpolation method for proving replica formulas.
  applications to the curie{\textendash}weiss and wigner spike models.
\newblock {\em Journal of Physics A: Mathematical and Theoretical},
  52(29):294002, jun 2019.

\bibitem{2017arXiv170108010L}
T.~{Lesieur}, L.~{Miolane}, M.~{Lelarge}, F.~{Krzakala}, and
  L.~{Zdeborov{\'a}}.
\newblock {Statistical and computational phase transitions in spiked tensor
  estimation}.
\newblock In {\em IEEE International Symposium on Information Theory (ISIT),
  2017}.

\bibitem{2017arXiv170910368B}
J.~{Barbier}, N.~{Macris}, and L.~{Miolane}.
\newblock {The Layered Structure of Tensor Estimation and its Mutual
  Information}.
\newblock In {\em 55th Annual Allerton Conference on Communication, Control,
  and Computing (Allerton)}, September 2017.

\bibitem{mourrat2019hamilton}
J.-C. Mourrat.
\newblock Hamilton-jacobi equations for finite-rank matrix inference.
\newblock {\em arXiv preprint arXiv:1904.05294}, 2019.

\bibitem{david2017high}
D.~Gamarnik and I.~Zadik.
\newblock High dimensional regression with binary coefficients. estimating
  squared error and a phase transtition.
\newblock In {\em Conference on Learning Theory}, pages 948--953, 2017.

\bibitem{reeves2019all}
G.~Reeves, J.~Xu, and I.~Zadik.
\newblock The all-or-nothing phenomenon in sparse linear regression.
\newblock In {\em Proceedings of the Thirty-Second Conference on Learning
  Theory}, volume~99 of {\em Proceedings of Machine Learning Research}, pages
  2652--2663. PMLR, 25--28 Jun 2019.

\bibitem{Kabashima_2003}
Y.~Kabashima.
\newblock A {CDMA} multiuser detection algorithm on the basis of belief
  propagation.
\newblock {\em Journal of Physics A: Mathematical and General},
  36(43):11111--11121, oct 2003.

\bibitem{bayati2011dynamics}
M.~Bayati and A.~Montanari.
\newblock The dynamics of message passing on dense graphs, with applications to
  compressed sensing.
\newblock {\em IEEE Trans. on Information Theory}, 2011.

\bibitem{Donoho10112009}
D.~L. Donoho, A.~Maleki, and A.~Montanari.
\newblock Message-passing algorithms for compressed sensing.
\newblock {\em Proceedings of the National Academy of Sciences},
  106(45):18914--18919, 2009.

\bibitem{krz12}
F.~Krzakala, M.~M{\'e}zard, F.~Sausset, Y.~Sun, and L.~Zdeborov{\'a}.
\newblock Probabilistic reconstruction in compressed sensing: algorithms, phase
  diagrams, and threshold achieving matrices.
\newblock {\em J. Stat. Mech. Theory Exp.}, (8), 2012.

\bibitem{MontChap11}
A.~Montanari.
\newblock Graphical models concepts in compressed sensing.
\newblock In Y.~C. Eldar and G.~Kutyniok, editors, {\em Compressed Sensing},
  pages 394--438. Cambridge University Press, 2012.

\bibitem{Rangan11}
S.~Rangan.
\newblock Generalized approximate message passing for estimation with random
  linear mixing.
\newblock In {\em Proc. IEEE Int. Symp. Inf. Theory}, pages 2168--2172, 2011.

\bibitem{deshpande2015finding}
Y.~Deshpande and A.~Montanari.
\newblock Finding hidden cliques of size $\sqrt{N/e}$ in nearly linear time.
\newblock {\em Foundations of Computational Mathematics}, 15(4):1069--1128,
  2015.

\bibitem{deshpande2014information}
Y.~Deshpande and A.~Montanari.
\newblock Information-theoretically optimal sparse pca.
\newblock In {\em 2014 IEEE International Symposium on Information Theory},
  pages 2197--2201. IEEE, 2014.

\bibitem{montanari2017estimation}
A.~Montanari and R.~Venkataramanan.
\newblock Estimation of low-rank matrices via approximate message passing.
\newblock {\em arXiv preprint arXiv:1711.01682}, 2017.

\bibitem{vila2015hyperspectral}
J.~Vila, P.~Schniter, and J.~Meola.
\newblock Hyperspectral unmixing via turbo bilinear approximate message
  passing.
\newblock {\em IEEE Transactions on Computational Imaging}, 1(3):143--158,
  2015.

\bibitem{fletcher2018iterative}
A.~K. Fletcher and S.~Rangan.
\newblock Iterative reconstruction of rank-one matrices in noise.
\newblock {\em Information and Inference: A Journal of the IMA}, 7(3):531--562,
  2018.

\bibitem{parker2014bilinear}
J.~T. Parker, P.~Schniter, and V.~Cevher.
\newblock Bilinear generalized approximate message passing. part i: Derivation.
\newblock {\em IEEE Transactions on Signal Processing}, 62(22):5839--5853,
  2014.

\bibitem{montanari2015non}
A.~Montanari and E.~Richard.
\newblock Non-negative principal component analysis: Message passing algorithms
  and sharp asymptotics.
\newblock {\em IEEE Transactions on Information Theory}, 62(3):1458--1484,
  2015.

\bibitem{baik2005phase}
J.~Baik, G.~B. Arous, and S.~P{\'e}ch{\'e}.
\newblock Phase transition of the largest eigenvalue for nonnull complex sample
  covariance matrices.
\newblock {\em Annals of Probability}, page 1643, 2005.

\bibitem{peche2006largest}
S.~P{\'e}ch{\'e}.
\newblock The largest eigenvalue of small rank perturbations of hermitian
  random matrices.
\newblock {\em Probability Theory and Related Fields}, 134(1):127--173, 2006.

\bibitem{feral2007largest}
D.~F{\'e}ral and S.~P{\'e}ch{\'e}.
\newblock The largest eigenvalue of rank one deformation of large wigner
  matrices.
\newblock {\em Communications in mathematical physics}, 272(1):185--228, 2007.

\bibitem{amini2009}
A.~A. Amini and M.~J. Wainwright.
\newblock High-dimensional analysis of semidefinite relaxations for sparse
  principal components.
\newblock {\em Ann. Statist.}, 37(5B):2877--2921, 10 2009.

\bibitem{pmlr-v75-brennan18a}
M.~Brennan, G.~Bresler, and W.~Huleihel.
\newblock Reducibility and computational lower bounds for problems with planted
  sparse structure.
\newblock In {\em Proceedings of the 31st Conference On Learning Theory},
  volume~75 of {\em Proceedings of Machine Learning Research}, pages 48--166.
  PMLR, 06--09 Jul 2018.

\bibitem{gamarnik2019overlap}
D.~Gamarnik, A.~Jagannath, and S.~Sen.
\newblock The overlap gap property in principal submatrix recovery.
\newblock {\em arXiv preprint arXiv:1908.09959}, 2019.

\bibitem{Cai2015}
T.~Cai, Z.~Ma, and Y.~Wu.
\newblock Optimal estimation and rank detection for sparse spiked covariance
  matrices.
\newblock {\em Probability Theory and Related Fields}, 161(3):781--815, Apr
  2015.

\bibitem{krauthgamer2015}
R.~Krauthgamer, B.~Nadler, and D.~Vilenchik.
\newblock Do semidefinite relaxations solve sparse pca up to the information
  limit?
\newblock {\em Ann. Statist.}, 43(3):1300--1322, 06 2015.

\bibitem{JMLR:v17:15-160}
Y.~Deshpande and A.~Montanari.
\newblock Sparse pca via covariance thresholding.
\newblock {\em Journal of Machine Learning Research}, 17(141):1--41, 2016.

\bibitem{wang2016}
T.~Wang, Q.~Berthet, and R.~J. Samworth.
\newblock Statistical and computational trade-offs in estimation of sparse
  principal components.
\newblock {\em Ann. Statist.}, 44(5):1896--1930, 10 2016.

\bibitem{pmlr-v30-Berthet13}
Q.~Berthet and P.~Rigollet.
\newblock Complexity theoretic lower bounds for sparse principal component
  detection.
\newblock In {\em Proceedings of the 26th Annual Conference on Learning
  Theory}, volume~30 of {\em Proceedings of Machine Learning Research}, pages
  1046--1066, Princeton, NJ, USA, 12--14 Jun 2013. PMLR.

\bibitem{Ma:2015:SLB:2969239.2969419}
T.~Ma and A.~Wigderson.
\newblock Sum-of-squares lower bounds for sparse pca.
\newblock In {\em Proceedings of the 28th International Conference on Neural
  Information Processing Systems - Volume 1}, NIPS'15, pages 1612--1620,
  Cambridge, MA, USA, 2015. MIT Press.

\bibitem{mezard2009information}
M.~Mezard and A.~Montanari.
\newblock {\em Information, physics and computation}.
\newblock Oxford University Press, 2009.

\bibitem{2017arXiv170100858L}
T.~Lesieur, F.~Krzakala, and L.~Zdeborov{\'{a}}.
\newblock Constrained low-rank matrix estimation: phase transitions,
  approximate message passing and applications.
\newblock {\em Journal of Statistical Mechanics: Theory and Experiment},
  2017(7):073403, jul 2017.

\bibitem{bolthausen2014iterative}
E.~Bolthausen.
\newblock An iterative construction of solutions of the tap equations for the
  sherrington--kirkpatrick model.
\newblock {\em Communications in Mathematical Physics}, 325(1):333--366, 2014.

\bibitem{korada2009exact}
S.~B. Korada and N.~Macris.
\newblock Exact solution of the gauge symmetric p-spin glass model on a
  complete graph.
\newblock {\em Journal of Statistical Physics}, 136(2):205--230, 2009.

\bibitem{krzakala2016mutual}
F.~Krzakala, J.~Xu, and L.~Zdeborov{\'a}.
\newblock Mutual information in rank-one matrix estimation.
\newblock In {\em 2016 IEEE Information Theory Workshop (ITW)}, pages 71--75.
  IEEE, 2016.

\bibitem{el2018estimation}
A.~El~Alaoui and F.~Krzakala.
\newblock Estimation in the spiked wigner model: a short proof of the replica
  formula.
\newblock In {\em 2018 IEEE International Symposium on Information Theory
  (ISIT)}, pages 1874--1878. IEEE, 2018.

\bibitem{barbier2019mutual}
J.~Barbier, C.~Luneau, and N.~Macris.
\newblock Mutual information for low-rank even-order symmetric tensor
  factorization.
\newblock {\em arXiv preprint arXiv:1904.04565}, 2019.

\bibitem{GuoShamaiVerdu_IMMSE}
D.~Guo, S.~Shamai, and S.~Verdu.
\newblock Mutual information and minimum mean-square error in gaussian
  channels.
\newblock {\em IEEE Trans. on Information Theory}, 51(4):1261--1282, April
  2005.

\bibitem{guo2011estimation}
D.~Guo, Y.~Wu, S.~S. Shitz, and S.~Verd{\'u}.
\newblock Estimation in gaussian noise: Properties of the minimum mean-square
  error.
\newblock {\em IEEE Transactions on Information Theory}, 57(4):2371--2385,
  2011.

\bibitem{perry2018optimality}
A.~Perry, A.~S. Wein, A.~S. Bandeira, A.~Moitra, et~al.
\newblock Optimality and sub-optimality of pca i: Spiked random matrix models.
\newblock {\em The Annals of Statistics}, 46(5):2416--2451, 2018.

\bibitem{alaoui2017finite}
A.~E. Alaoui, F.~Krzakala, and M.~I. Jordan.
\newblock Finite size corrections and likelihood ratio fluctuations in the
  spiked wigner model.
\newblock {\em arXiv preprint arXiv:1710.02903}, 2017.

\bibitem{alaoui2018detection}
A.~E. Alaoui and M.~I. Jordan.
\newblock Detection limits in the high-dimensional spiked rectangular model.
\newblock In {\em Conference On Learning Theory, {COLT} 2018, Stockholm,
  Sweden, 6-9 July 2018.}, pages 410--438, 2018.

\bibitem{RushVenkataramanan}
C.~Rush and R.~Venkataramanan.
\newblock Finite sample analysis of approximate message passing algorithms.
\newblock {\em {IEEE} Trans. Information Theory}, 64(11):7264--7286, 2018.

\bibitem{montanari2015finding}
A.~Montanari.
\newblock Finding one community in a sparse graph.
\newblock {\em Journal of Statistical Physics}, 161(2):273--299, 2015.

\bibitem{milgrom2002envelope}
P.~Milgrom and I.~Segal.
\newblock Envelope theorems for arbitrary choice sets.
\newblock {\em Econometrica}, 70(2):583--601, 2002.

\bibitem{Montanari-Javanmard}
A.~Javanmard and A.~Montanari.
\newblock State evolution for general approximate message passing algorithms,
  with applications to spatial coupling.
\newblock {\em J. Infor. \& Inference}, 2:115, 2013.

\bibitem{BLMConc}
S.~Boucheron, G.~Lugosi, and P.~Massart.
\newblock {\em Concentration inequalities: A nonasymptotic theory of
  independence}.
\newblock Oxford University Press, 2013.

\end{thebibliography}
 \appendix

 \section{General results on the mutual information}\label{app:matrix}

In this appendix we give a more general form of theorem \ref{thm:ws} in section~\ref{sec:matrixWinfotheorresults}. Our analysis by the adaptive interpolation method works for any regime where the sequences $\lambda_n$ and $\rho_n$ verify: 
\begin{align}
C\le \lambda_n\rho_n = O(n^{\gamma})\, \quad \text{for some constants} \quad \gamma \in [0,1/2) \quad \text{and} \quad C>0\,.
\label{app-scalingRegime}
\end{align}
Of course this contains the regime \eqref{mainregime} as a special case. Our general result is a statement on the smallness of 
\begin{align*}
\Delta I_n \equiv \frac{1}{\rho_n|\ln\rho_n|}\Big|\frac{1}{n}I(\bX;\bW) -\inf_{q\in [0,\rho_n]} i^{\rm pot}_n(q,\lambda_n,\rho_n)\Big|\,.
\end{align*}
The analysis of section \ref{sec:adapInterp_XX} leads to the following general theorem.

%

\begin{thm}[Sparse spiked Wigner model]\label{thm:wsgeneral}
Let the sequences $\lambda_n$ and $\rho_n$ verify \eqref{app-scalingRegime} and let $\alpha>0$. 
There exists a constant $C>0$ independent of $n$, such that the mutual information for the Wigner spike model verifies
\begin{align*}
\Delta I_n\le \frac{C}{|\ln\rho_n|}\max\Big\{\frac{1}{n^\alpha}, \, \frac{\lambda_n}{n\rho_n}, \, \Big(\frac{\lambda_n^4}{n^{1-4\alpha}\rho_n^2}\big(1+\lambda_n\rho_n^2\big)\Big)^{1/3}\Big\}\,.
\end{align*}	
In particular, choosing $\lambda_n= \Theta(|\ln\rho_n|/\rho_n)$ (which is the appropriate scaling to observe a phase transition), 
\begin{align*}
\Delta I_n \le C\max\Big\{\frac{1}{n^{\alpha}|\ln\rho_n|}, \, \frac{1}{n\rho_n^2}, \, \Big(\frac{|\ln \rho_n|}{n^{1-4\alpha}\rho_n^6}\Big)^{1/3}\Big\}\,.
\end{align*}
If, in addition, we set $\rho_n=\Omega(n^{-\beta})$ for $\beta\ge 0$ (which is the regime in \eqref{mainregime}), then we have
\begin{align*}
\Delta I_n \le C\max\Big\{\frac{1}{n^{\alpha}\ln n}, \, \frac{1}{n^{1-2\beta}}, \, \Big(\frac{\ln n}{n^{1-4\alpha-6\beta}}\Big)^{1/3}\Big\}\,.
\end{align*}
This bound vanishes as $n$ grows if $\beta\in[0,1/6)$ and $\alpha\in(0, (1-6\beta)/4]$. The final bound is optimized (up to polylog factors)  by setting $\alpha=(1-6\beta)/7$. In this case (again, when $\lambda_n= \Theta(|\ln\rho_n|/\rho_n)$ and $\rho_n=\Omega(n^{-\beta})$),
\begin{align*}
\Delta I_n \le C\frac{(\ln n)^{1/3}}{n^{(1-6\beta)/7}}\,.
\end{align*}
\end{thm}

 \section{Information theoretic analysis by the adaptive interpolation method}\label{sec:adapInterp_XX}
In this section we provide the essential architecture for the proof of theorem \ref{thm:ws} which relies on the adaptive interpolation method
\cite{BarbierM17a,BarbierMacris2019}. The proof requires concentration properties for ``free energies'' and ``overlaps'' which are deferred to appendices \ref{app:free-energy} and \ref{appendix-overlap}. 
%
When no confusion is possible we use the notation $\EE\|\bA\|^2=\EE[\|\bA\|^2]$.
\subsection{The interpolating model.}

Let $\epsilon \in[s_n, 2s_n]$, for a sequence $s_n$ tending to zero as $s_n = n^{-\alpha}/2 \in(0,1/2)$, for $\alpha >0$ chosen later. Let $q_{n}: [0, 1]\times [s_n,2s_n] \mapsto [0,\rho_n]$ and set 
$$
R_n(t,\epsilon)\equiv \epsilon+\lambda_n\int_0^tds\,q_n(s,\epsilon)\,.
$$
Consider the following interpolating estimation model, where $t\in[0,1]$, with accessible data $(W_{ij}(t))_{i,j}$ and $(\tilde W_i(t,\epsilon))_i$ obtained through
\begin{align*}
\begin{cases}
W_{ij}(t) =W_{ji}(t)\hspace{-5pt}&=\sqrt{(1-t)\frac{\lambda_n}{n}}\, X_iX_j + Z_{ij}\,, \qquad 1\le i< j\le n\,,\\
\tilde \bW(t,\epsilon)  &= \sqrt{R_n(t,\epsilon)}\,\bX + \tilde \bZ\,,
\end{cases} 
\end{align*}
with standard gaussian noise $\tilde \bZ\sim {\cal N}(0,{\rm I}_n)$, and $Z_{ij}=Z_{ji}\sim {\cal N}(0,1)$. 
The posterior associated with this model reads (here $\|-\|$ is the $\ell_2$ norm)
\begin{align*}
&dP_{n, t,\epsilon}(\bx|\bW(t),\tilde{\bW}(t,\epsilon))=\frac{1}{\mathcal{Z}_{n, t,\epsilon}(\bW(t),\tilde{\bW}(t,\epsilon))}\Big(\prod_{i=1}^n dP_{X,n}(x_i)\Big) \nn
& \times\exp\Big\{\sum_{i< j}^n \Big((1-t)\frac{\lambda_n}{n}\frac{x_i^2x_j^2}{2}-\sqrt{(1-t)\frac{\lambda_n}{n}}x_ix_jW_{ij}(t) \Big)+  R_n(t,\epsilon)\frac{\|\bx\|^2}{2} \\
&\qquad\qquad\qquad- \sqrt{R_n(t,\epsilon)} \bx\cdot \tilde \bW(t,\epsilon)\Big\}.
\end{align*}
The normalization factor $\mathcal{Z}_{n,t,\epsilon}(\dots)$ is also called partition function. We also define the mutual information density for the interpolating model
\begin{align}
i_{n}(t,\epsilon)&\equiv\frac{1}{n}I\big(\bX;(\bW(t),\tilde{\bW}(t,\epsilon))\big)\, \label{fnt}.
\end{align}
The $(n,t,\epsilon, R_n)$-dependent Gibbs-bracket (that we simply denote $\langle - \rangle_t$ for the sake of readability) is defined for functions $A(\bx)=A$
\begin{align}
\langle A(\bx) \rangle_{t}= \int dP_{n, t,\epsilon}(\bx| \bW(t),\tilde{\bW}(t,\epsilon))\,A(\bx) \,. \label{t_post}
\end{align}
%
%
\begin{lemma}[Boundary values]\label{lemma:bound1}
The mutual information for the interpolating model verifies 
\begin{align}\label{bound2}
\begin{cases}
i_{n}(0,\epsilon) =\frac1n I(\bX;\bW)+O(\rho_n s_n)\,,\\
i_{n}(1,\epsilon) =I_n(X;\{\lambda_n\int_0^1dt\,q_n(t,\epsilon)\}^{1/2}X+Z)+O(\rho_ns_n)\,. 
\end{cases} 
\end{align}
where $I_n(X;\{\lambda_n\int_0^1dt\,q_n(t,\epsilon)\}^{1/2}X+Z)$ is the mutual information for a scalar gaussian channel with input $X\sim P_{X,n}$ and noise $Z\sim{\cal N}(0,1)$.
\end{lemma}
\begin{proof}
We start with the chain rule for mutual information:
$$
i_{n}(0,\epsilon)= \frac1nI(\bX;{\bW}(0))+\frac1nI(\bX;\tilde{\bW}(0,\epsilon)|{\bW}(0)).
$$
Note that, by the definition of ${\bW}(t)$,  $$I(\bX;{\bW}(0))=I(\bX;\bW) \, .$$ Moreover we claim
$
\frac1nI(\bX;\tilde{\bW}(0,\epsilon)|{\bW}(0))=O(\rho_ns_n)	,
$
which yields the first identity in \eqref{bound2}.
This claim simply follows from the I-MMSE relation (appendix \ref{app:gaussianchannels}) and $R_n(0,\epsilon)=\epsilon$:
\begin{align}
\frac{d}{d\epsilon}\frac1nI(\bX;\tilde{\bW}(0,\epsilon)|{\bW}(0))=\frac1{2n} {\rm MMSE}(\bX|\tilde{\bW}(0,\epsilon),{\bW}(0)) \le \frac{\rho_n}{2}\,. \label{id9}
\end{align}
The last inequality above is true because 
${\rm MMSE}(\bX|\tilde{\bW}(0,\epsilon),{\bW}(0))\le \EE\|\bX-\EE\,\bX\|^2=n{\rm Var}(X_1)\le n\rho_n$, 
as the components of $\bX$ are i.i.d. from $P_{X,n}$. Therefore $\frac1nI(\bX;\tilde{\bW}(0,\epsilon)|{\bW}(0))$ is $\frac{\rho_n}{2}$-Lipschitz in $\epsilon\in[s_n,2s_n]$. Moreover, we have that $I(\bX;\tilde{\bW}(0,0)|{\bW}(0))=0$. This implies the claim. 

The proof of the second identity in \eqref{bound2} again starts from the chain rule for mutual information
$$
i_{n}(1,\epsilon)=\frac1nI(\bX;\tilde\bW(1,\epsilon))+\frac1nI(\bX;{\bW}(1)|\tilde\bW(1,\epsilon))\, .
$$
Note that $I(\bX;{\bW}(1)|\tilde\bW(1,\epsilon))=0$ as ${\bW}(1)$ does not depend on $\bX$. Moreover, 
\begin{align*}
\frac1nI(\bX;\tilde\bW(1,\epsilon))&=I_n(X;\sqrt{R_n(1,\epsilon)}X+Z)\nn
&=\textstyle{I_n(X;\{\lambda_n\int_0^1dt\,q_n(t,\epsilon)\}^{1/2}X+Z)+O(\rho_ns_n)}\,.
\end{align*}
because $I_n(X;\sqrt{\gamma}X+Z)$ is a $\frac{\rho_n}{2}$-Lipschitz function of $\gamma$, 
by an application of the I-MMSE relation (appendix \ref{app:gaussianchannels}) $\frac{d}{d\gamma}I_n(X;\sqrt{\gamma}X+Z)={\rm MMSE}(X|\sqrt{\gamma}X+Z)/2\le {\rm Var}(X)/2\le \rho_n/2$.
\end{proof}
\subsection{Fundamental sum rule.}
\begin{proposition}[Sum rule]\label{prop1}
The mutual information verifies the following sum rule:
\begin{align}
\frac1n I(\bX;\bW) &=   i_n^{\rm pot}\big({\textstyle \int_0^1dt\,q_n(t,\epsilon)};\lambda_n,\rho_n\big) +\frac{\lambda_n}{4}\big({\cal R}_1 - {\cal R}_2-{\cal R}_3\big)+ O(\rho_n s_n)+O\Big(\frac{\lambda_n}{n}\Big)	
\label{MF-sumrule}
\end{align}
with non-negative ``remainders'' that depend on $(n,\epsilon, R_{n})$,
\begin{align}\label{contributions}
	\begin{cases}
{\cal R}_1\equiv\int_0^1 dt \,\big(q_n(t,\epsilon) - \int_0^1 ds \, q_n(s,\epsilon)\big)^2\,,\\ 
{\cal R}_2\equiv \int_0^1	dt\,\E\big\langle\big(Q-\E\langle Q\rangle_t\big)^2\big\rangle_{t}\,,\\
{\cal R}_3\equiv \int_0^1	dt\,\big(q_n(t,\epsilon)-\E\langle Q\rangle_t\big)^2\,,
\end{cases}
\end{align}
where $Q = \frac{1}{n}\bx \cdot\bX$ is called the {\it overlap}. The constants in the $O(\cdots)$ terms are independent of $n,t,\epsilon$.	
\end{proposition}
\begin{proof}
By the fundamental theorem of calculus 
$i_{n}(0,\epsilon)=i_{n}(1,\epsilon)-\int_0^1dt \frac{d}{dt}i_{n}(t,\epsilon)$. 
Note that $i_{n}(0,\epsilon)$ and $i_{n}(1,\epsilon)$ are given by \eqref{bound2}.
The $t$-derivative of the interpolating mutual information is simply computed combining the I-MMSE relation with the chain rule for derivatives
\begin{align}
\frac{d}{dt}i_n(t,\epsilon)&= -\frac{\lambda_n}{2}\frac{1}{n^2}\sum_{i<j}\EE\big[(X_iX_j-\langle x_ix_j\rangle_{t})^2\big]+\frac{\lambda_nq_n(t,\epsilon)}{2}\frac1n\EE\|\bX-\langle \bx\rangle_{t}\|^2\label{9}\\
&=-\frac{\lambda_n}{4}\frac{1}{n^2}\EE\|\bX\otimes \bX-\langle \bx\otimes \bx\rangle_{t}\|_{\rm F}^2+\frac{\lambda_nq_n(t,\epsilon)}{2}\frac1n\EE\|\bX-\langle \bx\rangle_{t}\|^2+ O\Big(\frac{\lambda_n}{n}\Big)\label{34}\,.
\end{align}
The correction term in \eqref{34} comes from completing the diagonal terms 
in the sum $\sum_{i< j}$ in order to construct the matrix-MMSE for $\bX\otimes \bX$, 
namely the first term on the r.h.s. of \eqref{34}. This expression can be simplified by application of the Nishimori identities 
(appendix \ref{app:nishimori} contains a proof of these general identities). Starting with the second term (a vector-MMSE)
\begin{align}
\frac1n\E\|\bX-\langle \bx \rangle_t\|^2 & =\E\big[\|\bX\|^2+\|\langle\bx \rangle_t\|^2- 2\bX\cdot \langle \bx \rangle_t\big]\nn
& =  \frac1n\E\big[\|\bX\|^2 -\bX\cdot \langle \bx \rangle_t\big] =\rho_n -\EE\langle Q\rangle_t\,, \label{mmse}
\end{align}
were we used $\E\|\bX\|^2=n\rho_n$ and the Nishimori identity $\EE\|\langle \bx \rangle_t\|^2 = \EE[\bX\cdot\langle \bx \rangle_t]$.
By similar manipulations we obtain for the matrix-MMSE
\begin{align}\label{manip}
\frac1{n^2}{\rm MMSE}(\bX\otimes \bX|\tilde{\bW}(t,\epsilon),{\bW}(t))=\frac1{n^2}\EE\|\bX\otimes \bX-\langle \bx\otimes \bx\rangle_{t}\|_{\rm F}^2 =\rho_n^2-\EE\langle Q^2\rangle_t\,.
\end{align}
From \eqref{bound2}, \eqref{34}, \eqref{mmse}, \eqref{manip} and the fundamental theorem of calculus we deduce
\begin{align*}
\frac1n I(\bX;\bW) = &   {\textstyle{I_n\big(X;\{\lambda_n\int_0^1dt\,q_n(t,\epsilon)\}^{1/2}X+Z\big)}} 
 \nn & 
 + \frac{\lambda_n}{4}\int_0^1 dt\,\Big\{\rho_n^2-\EE\langle Q^2\rangle_t-2q_n(t,\epsilon)(\rho_n-\EE\langle Q\rangle_t)\Big\}+ O(\rho_n s_n)+O\Big(\frac{\lambda_n}{n}\Big)\,.
\end{align*}
The terms on the r.h.s can be re-arranged so that the potential \eqref{26} appears, and this gives immediately the sum rule \eqref{MF-sumrule}.
\end{proof}

%

Theorem~\ref{thm:ws} follows from the upper and lower bounds proven below, and applied for $s_n=\frac{1}{2}n^{-\alpha}$.

\subsection{Upper bound: linear interpolation path.}
\begin{proposition}[Upper bound]
We have
\begin{align*}
\frac1n I(\bX;\bW) \le \inf_{q\in [0,\rho_n]} i_n^{\rm pot}(q,\lambda_n,\rho_n)+ O(\rho_n s_n)+O\Big(\frac{\lambda_n}{n}\Big)\,.
\end{align*}
\end{proposition}
\begin{proof}
Fix $q_{n}(t,\epsilon)= q_n \in [0, \rho_n]$ a constant independent of $\epsilon, t$. The interpolation path $R_n(t,\epsilon)$  is therefore 
a simple linear function of time.
From \eqref{contributions} ${\cal R}_1$ cancels and since ${\cal R}_2$ and ${\cal R}_3$ are non-negative we get from Proposition
\eqref{prop1}
\begin{align*}
\frac{1}{n}I(\bX;\bW) \leq i_n^{\rm pot}(q,\lambda_n,\rho_n)+ O(\rho_n s_n)+O\Big(\frac{\lambda_n}{n}\Big)\,.
\end{align*}
Note that the error terms $O(\cdots)$ are bounded independently of $q_n$.
Therefore optimizing the r.h.s over the free parameter $q_n\in [0, \rho_n]$ yields the upper bound.
\end{proof}

\subsection{Lower bound: adaptive interpolation path.} 
We start with a definition: the map $\epsilon\mapsto R_n(t,\epsilon)$ is called \emph{regular} if it is a ${\cal C}^1$-diffeomorphism whose jacobian is greater or equal to one for all $t\in [0,1]$. 
\begin{proposition}[Lower bound]
Consider sequences $\lambda_n$ and $\rho_n$ satisfying $c_1\leq \lambda_n \rho_n \leq c_2n^\gamma$ for some constants positive constant $c_1, c_2$ and $\gamma\in [0, 1/2[$.
Then
\begin{align}\label{MF-sumrule_bound}
 \frac1n I(\bX;\bW) \ge\inf_{q\in[0,\rho_n]}i_n^{\rm pot}(q,\lambda_n,\rho_n) + O(\rho_n s_n)+O\Big(\frac{\lambda_n}{n}\Big)+O\Big(\Big(\frac{\lambda_n^4\rho_n}{ns_n^4}\Big)^{1/3}\Big)\,.
\end{align}
\end{proposition}
\begin{proof}
First note that the regime \eqref{mainregime} for the sequences $\lambda_n, \rho_n$ satisfies  the more general condition assumed in this lemma (this is the condition in theorem \ref{thm:wsgeneral} of appendix \ref{app:matrix}).
Assume for the moment that the map $\epsilon\mapsto R_n(t,\epsilon)$ is regular. Then, based on Proposition \ref{L-concentration} and identity \eqref{remarkable} (appendix \ref{appendix-overlap}), we have a bound on the overlap fluctuation. Namely, for some numerical constant $C\ge 0$ independent of $n$
\begin{align}\label{over-concen}
\frac{\lambda_n}{s_n}\int_{s_n}^{2s_n}d\epsilon\, {\cal R}_2 &= \frac{\lambda_n}{s_n}\int_{s_n}^{2s_n}d\epsilon \int_0^1dt\,  \mathbb{E}\big\langle (Q - \mathbb{E}\langle Q\rangle_{n,t,R_n(t,\epsilon)} )^2\big\rangle_{n,t,R_n(t,\epsilon)}\nonumber\\
&\le C\Big(\frac{\lambda_n^4\rho_n}{ns_n^4}\Big)^{1/3} \,.
\end{align}
Using this concentration result, and ${\cal R}_1 \geq 0$, and averaging the sum rule \eqref{MF-sumrule} over $\epsilon\in [s_n, 2s_n]$ (recall the error terms are independent of $\epsilon$) we find
\begin{align}\label{MF-sumrule-simple}
\frac{I(\bX;\bW)}{n} \ge & \frac1{s_n}\int_{s_n}^{2s_n}d\epsilon  i_n^{\rm pot}\big({\textstyle \int_0^1dt\,q_n(t,\epsilon)},\lambda_n,\rho_n\big) -  \frac{\lambda_n}{4}\frac1{s_n}\int_{s_n}^{2s_n}d\epsilon\int_0^1dt\, \big(q_n(t,\epsilon)-\EE\langle Q\rangle_{t}\big)^2 \nonumber\\
& \qquad\qquad+  O(\rho_n s_n)+O\Big(\frac{\lambda_n}{n}\Big)+O\Big(\Big(\frac{\lambda_n^4\rho_n}{ns_n^4}\Big)^{1/3}\Big)\,.
\end{align}
At this stage it is natural to see if we can choose $q_n(t,\epsilon)$ to be the solution of
$q_n(t,\epsilon)=\mathbb{E}\langle Q\rangle_{t}$.
Setting $F_n(t,R_n(t,\epsilon))\equiv \mathbb{E}\langle Q\rangle_{n,t,R_n(t,\epsilon)}$, we recognize 
a {\it first order ordinary differential equation}
\begin{align}\label{odexx}
\frac{d }{dt}R_n(t,\epsilon) = F_n(t,R_n(t,\epsilon))\quad\text{with initial condition}\quad R_n(0,\epsilon)=\epsilon\,.
\end{align}
As $F_n(t,R_n(t,\epsilon))$ is ${\cal C}^1$ with bounded derivative w.r.t. its second argument the Cauchy-Lipschitz theorem implies that \eqref{odexx} admits a unique global solution $R_{n}^*(t,\epsilon)=\epsilon+\int_0^t ds\,q_{n}^*(s,\epsilon)$, where $q_{n}^*:[0,1]\times [s_n,2s_n]\mapsto [0,\rho_n]$.
Note that any solution must satisfy $q_{n}^*(t,\epsilon)\in [0,\rho_n]$ because $\mathbb{E}\langle Q\rangle_{n,t,\epsilon}\in [0,\rho_n]$ as can be seen from a Nishimori identity (appendix \ref{app:nishimori}) and \eqref{mmse}. 

We check that $R_n^*$ is regular. By Liouville's formula the jacobian of the flow $\epsilon\mapsto R_{n}^*(t,\epsilon)$ satisfies 
$$
\frac{d}{d\epsilon}  R_{n}^*(t,\epsilon)=\exp\Big\{\int_0^t ds\, \frac{d}{dR} F_n(s,R)\Big|_{R=R_{n}^*(s,\epsilon)}\Big\}\,.
$$ 
Applying repeatedly the Nishimori identity of Lemma \ref{NishId} (appendix \ref{app:nishimori}) one obtains (this computation does not present any difficulty and can be found in section 6 of \cite{BarbierM17a})
\begin{align}
\frac{d}{dR} F_n(s,R)=\frac{1}{n}\sum_{i,j=1}^n\mathbb{E}\big[(\langle x_ix_j\rangle_{n,s,R}-\langle x_i\rangle_{n,s,R}\langle x_j\rangle_{n,s,R})^2\big]\ge  0\label{jacPos}
\end{align}
so that the flow has a jacobian greater or equal to one. In particular it is locally invertible (surjective). Moreover it is injective because of the unicity of the solution of the differential equation, and therefore it is a $C^1$-diffeomorphism. Thus $\epsilon\mapsto R_{n}^*(t,\epsilon)$ is regular.  With the choice $R_{n}^*$, i.e., by suitably \emph{adapting} the interpolation path, we cancel ${\cal R}_3$. This yields
\begin{align*}
 \frac1n I(\bX;\bW) &\ge \frac1{s_n}\int_{s_n}^{2s_n}d\epsilon \, i_n^{\rm pot}\big({\textstyle \int_0^1dt\,q_n^*(t,\epsilon)},\lambda_n,\rho_n\big) +O(\cdots)
 \nn &
 \ge\inf_{q\in[0,\rho_n]}i_n^{\rm pot}(q,\lambda_n,\rho_n) + O(\cdots)
\end{align*}
where the $O(\cdots)$ is a shorthand notation for the three error terms in \eqref{MF-sumrule-simple}. This the desired result.	
\end{proof}

\section{Concentration of free energy}\label{app:free-energy}

For this appendix it is convenient to use the language of statistical mechanics.

\subsection{Statistical mechanics notations}
 We express the posterior of the interpolating model 
 \begin{align}\label{tpost}
&dP_{n, t,\epsilon}(\bx|\bW(t),\tilde{\bW}(t,\epsilon))=\frac{1}{\mathcal{Z}_{n, t,\epsilon}(\bW(t),\tilde{\bW}(t,\epsilon))}\nn
&\qquad\qquad\qquad\qquad\qquad\qquad\times\Big(\prod_{i=1}^n dP_{X,n}(x_i)\Big) \exp\big\{-{\cal H}_{n, t, \epsilon}(\bx, \bW(t),\tilde{\bW}(t,\epsilon))\big\}
\end{align}
with normalization constant (partition function) $\mathcal{Z}_{n,t,\epsilon}$ and ``hamiltonian''
\begin{align}
&{\cal H}_{n, t, \epsilon}(\bx , \bW(t),\tilde{\bW}(t,\epsilon))={\cal H}_{n, t, \epsilon}(\bx, \bX,\bZ,\tilde \bZ)\label{Ht}\\\
&\ \equiv\sum_{i< j}^n \Big((1-t)\frac{\lambda_n}{n}\frac{x_i^2x_j^2}{2}-\sqrt{(1-t)\frac{\lambda_n}{n}}x_ix_jW_{ij}(t) \Big)+  R_n(t,\epsilon)\frac{\|\bx\|^2}{2} - \sqrt{R_n(t,\epsilon)} \bx\cdot \tilde \bW(t,\epsilon) \nn
&\ =(1-t)\lambda_n\sum_{i< j}^n \Big(\frac{x_i^2x_j^2}{2n}-\frac{x_ix_jX_iX_j}{n} -\frac{x_ix_jZ_{ij}}{\sqrt{n(1-t)\lambda_n}}\Big)\\
& \qquad + R_n(t,\epsilon)\Big(\frac{\|\bx\|^2}{2} - \bx \cdot \bX-\frac{\bx\cdot \tilde{\bZ}}{\sqrt{R_n(t,\epsilon)}}\Big).\nonumber
\end{align}
It will also be convenient to work with ``free energies'' rather than mutual informations. The free energy $F_{n}(t,\epsilon)$ and (its expectation $f_{n}(t,\epsilon)$) for the interpolating model is simply minus the (expected) log-partition function:
\begin{align}
F_{n,t,\epsilon}(\bW(t),\tilde{\bW}(t,\epsilon))&\equiv -\frac1n \ln \mathcal{Z}_{n, t, \epsilon}(\bW(t),\tilde{\bW}(t,\epsilon))\,,\label{F_nonav}\\
f_{n}(t,\epsilon)&\equiv \EE\,F_{n, t, \epsilon}(\bW(t),\tilde{\bW}(t,\epsilon))\label{f_av}	\,.
\end{align}
The expectation $\EE$ carries over the data. The averaged free energy is related to the mutual information $i_{n}(t,\epsilon)$ given by \eqref{fnt} through
\begin{align}
i_{n}(t,\epsilon)	=f_{n}(t,\epsilon)+\frac{n-1}{n}\frac{\rho^2\lambda(1-t)}{4}+\frac{\rho R_n(t,\epsilon)}{2}\,.\label{infomut_freeen}
\end{align}

\subsection{Free energy concentration}
In this section we prove a concentration identity for the free energy \eqref{F_nonav} onto its average \eqref{f_av}.
\begin{proposition}[Free energy concentration for the spiked Wigner model]\label{prop:fconc}
We have 
\begin{align*}
\mathbb{E}\Big[\Big(F_{n,t,\epsilon}(\bW(t),&\tilde{\bW}(t,\epsilon))-f_{n}(t,\epsilon)\Big)^2\Big]
\le\frac{2\rho_nS^2}{n}\Big((2s_n+\lambda_n\rho_n)^2+S^4\Big)+\frac32\frac{\lambda_n\rho_n^2}{n}+2\frac{s_n\rho_n}{n}\,.	
\end{align*}
Considering sequences $\lambda_n$ and $\rho_n$ verifying \eqref{app-scalingRegime} and with $s_n=(1/2)n^{-\alpha}\to0_+$ the bound simplifies to $C(S) \lambda_n^2\rho_n^3/n$ with positive constant $C(S)\le \frac52+8S^2+2S^6$.
\end{proposition}
The proof is based on two classical concentration inequalities,
\begin{proposition}[Gaussian Poincar\'e inequality]\label{poincare}
	Let $\bU = (U_1, \dots, U_N)$ be a vector of $N$ independent standard normal random variables. Let $g: \mathbb{R}^N \to \mathbb{R}$ be a continuously differentiable function. Then
 \begin{align*}
	 \Var (g(\bU)) \leq \E \| \nabla g (\bU) \|^2  \,.
 \end{align*}
\end{proposition}
\begin{proposition}[Efron-Stein inequality]\label{efron_stein}
	Let $\,\mathcal{U}\subset \R$, and a function $g: \mathcal{U}^N \to \mathbb{R}$. Let $\,\bU=(U_1, \dots, U_N)$ be a vector of $N$ independent random variables with law $P_U$ that take values in $\,\mathcal{U}$. Let $\,\bU^{(i)}$ a vector which differs from $\bU$ only by its $i$-th component, which is replaced by $U_i'$ drawn from $P_U$ independently of $\,\bU$. Then
 \begin{align*}
	 \Var(g(\bU)) \leq \frac{1}{2} \sum_{i=1}^N \EE_{\bU}\EE_{U_i'}\big[(g(\bU)-g(\bU^{(i)}))^2\big] \,.
 \end{align*}
\end{proposition}

We start by proving the concentration w.r.t. the gaussian variables. It is convenient to make explicit the dependence of the partition function of the interpolating model in the independent quenched variables instead of the data: 
${\cal Z}_{n,t,\epsilon}(\bX, \bZ, \tilde \bZ)=\mathcal{Z}_{n,t,\epsilon}(\bW(t),\tilde{\bW}(t,\epsilon))$.
\begin{lemma}[Concentration w.r.t. the gaussian variables] We have
\begin{align*}
\mathbb{E}\Big[\Big(\frac{1}{n}\ln {\cal Z}_{n,t,\epsilon}(\bX, \bZ, \tilde \bZ)-\frac{1}{n}\mathbb{E}_{\bZ,\tilde \bZ}\ln {\cal Z}_{n,t,\epsilon}(\bX, \bZ, \tilde \bZ)\Big)^2\Big]\le \frac32\frac{\lambda_n\rho_n^2}{n}+2\frac{s_n\rho_n}{n}\,.	
\end{align*}
\end{lemma}
\begin{proof}
Fix all variables except $\bZ, \tilde \bZ$. Let $g(\bZ, \tilde \bZ)\equiv-\frac{1}{n}\ln {\cal Z}_{n,t,\epsilon}(\bX, \bZ, \tilde \bZ)$ be the free energy seen as a function of the gaussian variables only. The free energy gradient reads $\EE\|\nabla g\|^2=\EE\|\nabla_\bZ g\|^2+\EE\|\nabla_{\tilde \bZ} g\|^2$. Let us denote ${\cal H}(t)\equiv {\cal H}_{n,t,\epsilon}$ the interpolating Hamiltonian \eqref{Ht}.
\begin{align*}
\EE\|\nabla_\bZ g\|^2= \frac1{n^2}\EE\|\langle \nabla_\bZ{\cal H}(t)\rangle_t\|^2	&=\frac{(1-t)\lambda_n}{n^3}\sum_{i<j}\EE[\langle x_ix_j \rangle_t^2]\le\frac{(1-t)\lambda_n}{n^3}\sum_{i<j}\EE\langle (x_ix_j)^2 \rangle_{t}\nn
&\overset{\rm N}{=}\frac{(1-t)\lambda_n}{n^3}\sum_{i<j}\EE[(X_iX_j)^2]\le \frac{\lambda_n\rho_n^2}{2n}
\end{align*}
where we used a Nishimori identity for the last equality.
Similarly, and using $\lambda_n\rho_n\ge1$ and $s_n< 1/2$,
\begin{align*}
\EE\|\nabla_{\tilde \bZ} g\|^2=\frac{R(\epsilon)}{n^2}\EE\|\langle \bx \rangle_t\|^2\le\frac{R(\epsilon)}{n^2}\EE\langle \|\bx\|^2 \rangle_t\overset{\rm N}{=}\frac{R(\epsilon)}{n^2}\EE\|\bX\|^2\le \frac{(2s_n + \rho_n\lambda_n)\rho_n}{n}\,.
\end{align*}
Therefore Proposition~\ref{poincare} directly implies the stated result.
\end{proof}

We now consider the fluctuations due to the signal realization:
\begin{lemma}[Concentration w.r.t. the spike] We have
\begin{align*}
\mathbb{E}\Big[\Big(-\frac{1}{n}\mathbb{E}_{\bZ,\tilde \bZ}\ln {\cal Z}_{n,t,\epsilon}(\bX, \bZ, \tilde \bZ)-f_{n}(t,\epsilon)\Big)^2\Big]\le \frac{2\rho_nS^2}{n}\Big((2s_n+\lambda_n\rho_n)^2+S^4\Big)\,.	
\end{align*}

\end{lemma}
\begin{proof}
Let $g(\bX)\equiv-\frac{1}{n}\mathbb{E}_{\bZ,\tilde \bZ}\ln {\cal Z}_{n,t,\epsilon}( \bX, \bZ, \tilde \bZ)$. Define $\bX^{(i)}$ as a vector with same entries as $\bX$ except the $i$-th one that is replaced by $X_i'$ drawn independently from $P_{X,n}$. Let us estimate $(g(\bX)-g(\bX^{(i)}))^2$ by interpolation. Let ${\cal H}(t,s\bX+(1-s)\bX^{(i)})$ be the interpolating Hamiltonian \eqref{Ht} with $\bX$ replaced by $s\bX+(1-s)\bX^{(i)}$. Then
\begin{align*}
\EE\big[(g(\bX)-&g(\bX^{(i)}))^2\big] = \EE\Big[\Big(\int_0^1 ds \frac{dg}{ds}(s\bX+(1-s)\bX^{(i)})\Big)^2\Big] \nn
&=\frac1{n^2}\EE\Big[\Big(\int_0^1 ds \Big\langle\frac{d{\cal H}}{ds}(t,s\bX+(1-s)\bX^{(i)})\Big\rangle_t\Big)^2\Big]\nn
&= \frac{1}{n^2}\EE\Big[\Big( (X_i-X_i')\Big\langle R_\epsilon(t) x_i +  \frac{1-t}{n}x_i\sum_{j(\neq i)}X_jx_j\Big\rangle_t\Big)^2\Big]\nn
&\le \frac{2}{n^2}\EE\Big[(X_i-X_i')^2\Big(\langle x_i\rangle_t^2(2s_n+\lambda_n\rho_n)^2+\frac1{n^2}\sum_{j,k(\neq i)}X_jX_k\langle x_ix_j\rangle_t\langle x_ix_k\rangle_t\Big)\Big]\nn
&\le \frac{2}{n^2}\EE\big[(X_i-X_i')^2\big]\Big(S^2(2s_n+\lambda_n\rho_n)^2+S^6\Big)\nn
&\le\frac{4\rho_nS^2}{n^2}\Big((2s_n+\lambda_n\rho_n)^2+S^4\Big)\,.
\end{align*}
We used $(a+b)^2\le 2(a^2+b^2)$ for the second inequality and $\EE[(X_i-X_i')^2]=2{\rm Var}(X_i)\le 2\rho_n$. Therefore Proposition~\ref{efron_stein} implies the claim.
\end{proof}

 \section{Overlap concentration: proof of inequality \eqref{over-concen}}\label{appendix-overlap}
%
%
The derivations below will apply for any $t\in[0,1]$ so we drop all un-necessary notations and indices. Only the dependence of the free energies in $R(\epsilon)\equiv R_n(t,\epsilon)$ matters, so we denote $F(R(\epsilon))\equiv F_{n,t,\epsilon}(\bW(t),\tilde{\bW}(t,\epsilon))$ and $f(R(\epsilon))\equiv f_{n}(t,\epsilon)$.

Let $\mathcal{L}$ be the $R(\epsilon)$-derivative of the Hamiltonian \eqref{Ht} divided by $n$:
\begin{align}
\mathcal{L}(\bx,\bX,\tilde \bZ) =\mathcal{L} \equiv \frac1n \frac{d\mathcal{H}_{n,t,\epsilon}}{dR(\epsilon)}= \frac{1}{n}\Big(\frac{\|\bx\|^2}{2} - \bx\cdot \bX - \frac{\bx\cdot \tilde \bZ}{2\sqrt{R(\epsilon)}} \Big)\,.\label{def_L}
\end{align}
The overlap fluctuations are upper bounded by those of $\mathcal{L}$, which are easier to control, as
\begin{align}
\mathbb{E}\big\langle (Q - \mathbb{E}\langle Q \rangle_{t})^2\big\rangle_{t} \le 4\,\mathbb{E}\big\langle (\mathcal{L} - \mathbb{E}\langle \mathcal{L}\rangle_{t})^2\big\rangle_{t}\,.\label{remarkable}
\end{align}
The bracket is again the expectation w.r.t. the posterior of the interpolating model \eqref{t_post}.
A detailed derivation of this inequality can be found in appendix \ref{proof:remarkable_id} and involves only elementary algebra using the Nishimori identity
and integrations by parts w.r.t.\ the gaussian noise $\tilde \bZ$.

We have the following identities: for any given realisation of the quenched disorder
\begin{align}
 \frac{dF}{dR(\epsilon)}  &= \langle \mathcal{L} \rangle_{t} \,,\label{first-derivative}\\
 \frac{1}{n}\frac{d^2F}{dR(\epsilon)^2}  &= -\big\langle (\mathcal{L}  - \langle \mathcal{L} \rangle_t)^2\big\rangle_{t}+
 \frac{1}{4 n^2R(\epsilon)^{3/2}} \langle \bx\rangle_{t} \cdot \tilde \bZ\,.\label{second-derivative}
\end{align}
The gaussian integration by part formula \eqref{GaussIPP} with hamiltonian \eqref{Ht} yields
\begin{align}
\frac{\mathbb{E} \big\langle \tilde \bZ\cdot  \bx \big\rangle_t}{\sqrt{R(\epsilon)}} 
	 =   \mathbb{E}\big\langle \|\bx\|^2 \big\rangle_t - \EE\|\langle \bx \rangle_t\|^2 \overset{\rm N}{=}\mathbb{E}\big\langle \|\bx\|^2 \big\rangle_t - \EE \big\langle\bX\cdot \bx \big\rangle_t= \mathbb{E}\big\langle \|\bx\|^2 \big\rangle_t - n\,\EE \langle Q\rangle_t\,.\label{NishiTildeZ}
\end{align}
Therefore averaging \eqref{first-derivative} and \eqref{second-derivative} we find 
\begin{align}
 \frac{df}{d R(\epsilon)} &= \mathbb{E}\langle \mathcal{L} \rangle_{t} 
 \overset{\rm N}{=}-\frac{1}{2}\mathbb{E}\langle Q\rangle_{t}\,,\label{first-derivative-average}\\
 \frac{1}{n}\frac{d^2f}{dR(\epsilon)^2} &= -\mathbb{E}\big\langle (\mathcal{L} - \langle \mathcal{L} \rangle_{t})^2\big\rangle_{t}
 +\frac{1}{4n^2R(\epsilon)} \mathbb{E}\big\langle \|\bx- \langle \bx \rangle_t\|^2\big\rangle_t\,.\label{second-derivative-average}
\end{align} 
We always work under the assumption that the map $\epsilon\in [s_n, 2s_n] \mapsto R(\epsilon)\in [R(s_n), R(2s_n)]$ is regular, and do not repeat this assumption in the statements below. The concentration inequality \eqref{over-concen} is  a direct consequence of the following result (combined with Fubini's theorem):
\begin{proposition}[Total fluctuations of $\mathcal{L}$]\label{L-concentration} Let the sequences $\lambda_n$ and $\rho_n$ verify \eqref{app-scalingRegime}. Then $$\int_{s_n}^{2s_n} d\epsilon\,\mathbb{E}\big\langle (\mathcal{L} - \mathbb{E}\langle \mathcal{L}\rangle_{t})^2\big\rangle_{t} \le C\Big(\frac{\lambda_n\rho_n}{ns_n}\big(1+\lambda_n\rho_n^2\big)\Big)^{1/3}$$
for a constant $C>0$ that is independent of $n$, as long as the r.h.s. is $\omega(1/n)$.
\end{proposition}

The proof of this proposition is broken in two parts, using the decomposition
\begin{align*}
\mathbb{E}\big\langle (\mathcal{L} - \mathbb{E}\langle \mathcal{L}\rangle_{t})^2\big\rangle_{t}
& = 
\mathbb{E}\big\langle (\mathcal{L} - \langle \mathcal{L}\rangle_t)^2\big\rangle_{t}
+ 
\mathbb{E}\big[(\langle \mathcal{L}\rangle_t - \mathbb{E}\langle \mathcal{L}\rangle_t)^2\big]\,.
\end{align*}
Thus it suffices to prove the two following lemmas. The first lemma expresses concentration w.r.t.\ the posterior distribution (or ``thermal fluctuations'') and is a direct consequence of concavity properties of the average free energy and the Nishimori identity.
\begin{lemma}[Thermal fluctuations of $\cal L$]\label{thermal-fluctuations}
	We have $$\int_{s_n}^{2s_n} d\epsilon\, 
  \mathbb{E} \big\langle (\mathcal{L} - \langle \mathcal{L}\rangle_t)^2 \big\rangle_t  \le \frac{\rho_n}{n}\Big(1+\frac{\ln2}{4}\Big)\,.$$
\end{lemma}
\begin{proof}

We emphasize again that the interpolating free energy \eqref{fnt} is here viewed as a function of $R(\epsilon)$. In the argument that follows we consider derivatives of this function w.r.t. $R(\epsilon)$.
By \eqref{second-derivative-average}
\begin{align}
\mathbb{E}\big\langle (\mathcal{L} - \langle \mathcal{L} \rangle_t)^2\big\rangle_t
& = 
-\frac{1}{n}\frac{d^2f}{dR(\epsilon)^2}
+\frac{1}{4n^2R(\epsilon)} \big(\mathbb{E}\big\langle \|\bx\|^2 \big\rangle_t - \EE\|\langle \bx \rangle_t\|^2\big)
\nonumber \\ &
\leq 
-\frac{1}{n}\frac{d^2f}{dR(\epsilon)^2} +\frac{\rho_n}{4n\epsilon} \,,
\label{directcomputation}
\end{align}
where we used $R(\epsilon)\geq \epsilon$ and $\frac1n\mathbb{E}\langle \|\bx\|^2\rangle_t \overset{\rm N}{=} \mathbb{E}[X_1^2]=\rho_n$. We integrate this inequality over $\epsilon\in [s_n, 2s_n]$. Recall the map 
$\epsilon\mapsto R(\epsilon)$ has a Jacobian $\ge 1$, is $\mathcal{C}^1$ 
and has a well defined $\mathcal{C}^1$ inverse since we have assumed that it is regular. Thus integrating \eqref{directcomputation} and performing a change of variable (to get the second inequality) we obtain
\begin{align*}
\int_{s_n}^{2s_n} d\epsilon\, \mathbb{E}\big\langle (\mathcal{L} - \langle \mathcal{L} \rangle_t)^2\big\rangle_t
& \leq 
- \frac{1}{n}\int_{s_n}^{2s_n} d\epsilon \,\frac{d^2f}{dR(\epsilon)^2} + \frac{\rho_n}{4n}\int_{s_n}^{2s_n} \,\frac{d\epsilon}{\epsilon} 
\nonumber \\ &
\leq 
- \frac{1}{n}\int_{R(s_n)}^{R(2s_n)} dR(\epsilon) \,\frac{d^2f}{dR(\epsilon)^2}
+ \frac{\rho_n}{4n}\int_{s_n}^{2s_n} \,\frac{d\epsilon}{\epsilon} 
\nonumber \\ &
=
\frac{1}{n}\Big(\frac{df}{dR(\epsilon)}(R(s_n)) 
- \frac{df}{dR(\epsilon)}(R(2s_n))\Big)+ \frac{\rho_n}{4n}\ln 2\,.
\end{align*}
We have $|f'(R(\epsilon))| = |\mathbb{E}\langle Q\rangle_t/2|\le \rho_n/2$ so the first term is certainly smaller in absolute value than $\rho_n/n$. This concludes the proof of Lemma \ref{thermal-fluctuations}.
\end{proof}

The second lemma expresses the concentration w.r.t.\ the quenched disorder variables
and is a consequence of the concentration of the free energy onto its average (w.r.t. the quenched variables).
\begin{lemma}[Quenched fluctuations of $\cal L$]\label{disorder-fluctuations}
	Let the sequences $\lambda_n$ and $\rho_n$ verify \eqref{app-scalingRegime}. Then $$\int_{s_n}^{2s_n} d\epsilon\, 
  \mathbb{E}\big[ (\langle \mathcal{L}\rangle_t - \mathbb{E}\langle \mathcal{L}\rangle_t)^2\big] \le C\Big(\frac{\lambda_n\rho_n}{ns_n}\big(1+\lambda_n\rho_n^2\big)\Big)^{1/3}$$
  for a constant $C>0$ that is independent of $n$, as long as the r.h.s. is $\omega(1/n)$. 
\end{lemma}
\begin{proof}
Consider the following functions of $R(\epsilon)$:
\begin{align}\label{new-free}
 & \tilde F(R(\epsilon)) \equiv F(R(\epsilon)) +S\frac{\sqrt{R(\epsilon)}}{n} \sum_{i=1}^n\vert \tilde Z_i\vert\,,
 \nonumber \\ &
 \tilde f(R(\epsilon)) \equiv \mathbb{E} \,\tilde F(R(\epsilon))= f(R(\epsilon)) + S\sqrt{R(\epsilon)} \mathbb{E}\,\vert \tilde Z_1\vert\,.
\end{align}
Because of 
\eqref{second-derivative} we see that the second derivative of $\tilde F(R(\epsilon))$ w.r.t. $R(\epsilon)$ is negative so that it is concave. Note $F(R(\epsilon))$ itself is not necessarily concave in $R(\epsilon)$, although $f(R(\epsilon))$ is. Concavity of $f(R(\epsilon))$ is not obvious from \eqref{second-derivative-average} (obtained from differentiating $\EE\langle\mathcal{L}\rangle_t$ w.r.t. $R(\epsilon)$) but can be seen from \eqref{secondDer_f_pos} (obtained instead by differentiating $-\frac12\EE\langle{Q}\rangle_t$) which reads $\frac{d}{dR(\epsilon)}\EE\langle Q\rangle_t=-2\frac{d^2}{dR(\epsilon)^2}f\ge 0$. 
Evidently $\tilde f(R(\epsilon))$ is concave too.
Concavity then allows to use the following standard lemma:
\begin{lemma}[A bound for concave functions]\label{lemmaConvexity}
Let $G(x)$ and $g(x)$ be concave functions. Let $\delta>0$ and define $C^{-}_\delta(x) \equiv g'(x-\delta) - g'(x) \geq 0$ and $C^{+}_\delta(x) \equiv g'(x) - g'(x+\delta) \geq 0$. Then
\begin{align*}
|G'(x) - g'(x)| \leq \delta^{-1} \sum_{u \in \{x-\delta,\, x,\, x+\delta\}} |G(u)-g(u)| + C^{+}_\delta(x) + C^{-}_\delta(x)\,.
\end{align*}
\end{lemma}
First, from \eqref{new-free} we have 
\begin{align}\label{fdiff}
 \tilde F(R(\epsilon)) - \tilde f(R(\epsilon)) = F(R(\epsilon)) - f(R(\epsilon)) + S\sqrt{R(\epsilon)}  A_n 
\end{align} 
with $A_n \equiv \frac{1}{n}\sum_{i=1}^n \vert \tilde Z_i\vert -\mathbb{E}\,\vert \tilde Z_1\vert$.
Second, from \eqref{first-derivative}, \eqref{first-derivative-average} we obtain for the $R(\epsilon)$-derivatives
\begin{align}\label{derdiff}
 \tilde F'(R(\epsilon)) - \tilde f'(R(\epsilon)) = 
\langle \mathcal{L} \rangle_t-\mathbb{E}\langle \mathcal{L} \rangle_{t} + \frac{SA_n}{2\sqrt{R(\epsilon)}} \,.
\end{align}
From \eqref{fdiff} and \eqref{derdiff} it is then easy to show that Lemma \ref{lemmaConvexity} implies
\begin{align}\label{usable-inequ}
\vert \langle \mathcal{L}\rangle_t - \mathbb{E}\langle \mathcal{L}\rangle_t\vert&\leq 
\delta^{-1} \sum_{u\in \{R(\epsilon) -\delta,\, R(\epsilon),\, R(\epsilon)+\delta\}}
 \big(\vert F(u) - f(u) \vert + S\vert A_n \vert \sqrt{u} \big)\nonumber\\
 &\qquad\qquad\qquad\qquad
  + C_\delta^+(R(\epsilon)) + C_\delta^-(R(\epsilon)) + \frac{S\vert A_n\vert}{2\sqrt \epsilon} 
\end{align}
where $C_\delta^-(R(\epsilon))\equiv \tilde f'(R(\epsilon)-\delta)-\tilde f'(R(\epsilon))\ge 0$ 
and $C_\delta^+(R(\epsilon))\equiv \tilde f'(R(\epsilon))-\tilde f'(R(\epsilon)+\delta)\ge 0$. 
We used $R(\epsilon)\ge \epsilon$ for the term $S\vert A_n\vert/(2\sqrt \epsilon)$. Note that $\delta$ will 
be chosen later on strictly smaller than $s_n$ so that $R(\epsilon) -\delta \geq \epsilon - \delta \geq s_n -\delta$ remains 
positive. Remark that by independence of the noise variables $\mathbb{E}[A_n^2]  = (1-2/\pi)/n\le 1/n$. 
We square the identity \eqref{usable-inequ} and take its expectation. Then using $(\sum_{i=1}^pv_i)^2 \le p\sum_{i=1}^pv_i^2$, and 
that $R(\epsilon)\le 2s_n+\lambda_n\rho_n$, as well as the free energy concentration Proposition \ref{prop:fconc} (under the assumption that $\lambda_n$ and $\rho_n$ verify \eqref{app-scalingRegime}),
\begin{align}\label{intermediate}
 \frac{1}{9}\mathbb{E}\big[(\langle \mathcal{L}\rangle_t - \mathbb{E}\langle \mathcal{L}\rangle_t)^2\big]
 &
 \leq \, 
 \frac{3}{n\delta^2} \Big(C\lambda_n^2\rho_n^3 +S(2s_n+\lambda_n\rho_n+\delta)\Big) 
 \nonumber \\ & \qquad \qquad\qquad\qquad\qquad+ C_\delta^+(R(\epsilon))^2 + C_\delta^-(R(\epsilon))^2
 + \frac{S}{4n\epsilon} \,.
\end{align}
Recall $|C_\delta^\pm(R(\epsilon))|=|\tilde f'(R(\epsilon)\pm\delta)-\tilde f'(R(\epsilon))|$. By \eqref{first-derivative-average}, \eqref{new-free}  and $R(\epsilon)\ge \epsilon$ we have
\begin{align}
|\tilde f'(R(\epsilon))|  \leq \frac12\Big(\rho_n  +\frac{S}{\sqrt{R(\epsilon)}} \Big)\leq \frac12\Big(\rho_n  +\frac{S}{\sqrt \epsilon} \Big)\label{boudfprime}	
\end{align}
Thus, as $\epsilon\ge s_n$,  
$$|C_\delta^\pm(R(\epsilon))|\le \rho_n  +\frac{S}{\sqrt{\epsilon-\delta}}\le \rho_n  +\frac{S}{\sqrt{s_n-\delta}}\,.$$ We reach
\begin{align*}
 &\int_{s_n}^{2s_n} d\epsilon\, \big\{C_\delta^+(R(\epsilon))^2 + C_\delta^-(R(\epsilon))^2\big\}\\
 &\qquad\leq 
 \Big(\rho_n +\frac{S}{\sqrt {s_n-\delta}}\Big)
 \int_{s_n}^{2s_n} d\epsilon\, \big\{C_\delta^+(R(\epsilon)) + C_\delta^-(R(\epsilon))\big\}
 \nonumber \\ 
  &\qquad\leq  \Big(\rho_n +\frac{S}{\sqrt {s_n-\delta}}\Big)
 \int_{R(s_n)}^{R(2s_n)} dR(\epsilon)\, \big\{C_\delta^+(R(\epsilon)) + C_\delta^-(R(\epsilon))\big\}
 \nonumber \\ 
&\qquad= \Big(\rho_n +\frac{S}{\sqrt {s_n-\delta}}\Big)\Big[\Big(\tilde f(R(s_n)+\delta) - \tilde f(R(s_n)-\delta)\Big)\nn
&\qquad\qquad\qquad\qquad\qquad+ \Big(\tilde f(R(2s_n)-\delta) - \tilde f(R(2s_n)+\delta)\Big)\Big]
\end{align*}
where we used that the Jacobian of the $\mathcal{C}^1$-diffeomorphism $\epsilon\mapsto R(\epsilon)$ is $\ge 1$ (by regularity) for 
the second inequality. The mean value theorem and \eqref{boudfprime} 
imply $|\tilde f(R(\epsilon)-\delta) - \tilde f(R(\epsilon)+\delta)|\le \delta(\rho_n  +\frac{S}{\sqrt{s_n-\delta}})$. Therefore
\begin{align*}
 \int_{s_n}^{2s_n} d\epsilon\, \big\{C_\delta^+(R(\epsilon))^2 + C_\delta^-(R(\epsilon))^2\big\}\leq 
 2\delta \Big(\rho_n +\frac{S}{\sqrt{s_n-\delta}}\Big)^2\,.
\end{align*}
Set $\delta  = \delta_n = o(s_n)$. Thus, integrating \eqref{intermediate} over $\epsilon\in [s_n, 2s_n]$ yields
\begin{align*}
 &\int_{s_n}^{2s_n} d\epsilon\, 
 \mathbb{E}\big[(\langle \mathcal{L}\rangle_t - \mathbb{E}\langle \mathcal{L}\rangle_{t})^2\big]\nonumber\\
 &\qquad\qquad\leq \frac{27s_n}{n\delta_n^2}\Big(C\lambda_n^2\rho_n^3 +S(2s_n+\lambda_n\rho_n+\delta_n)\Big) +18\delta_n \Big(\rho_n +\frac{S}{\sqrt {s_n-\delta_n}}\Big)^2 +  \frac{9 S\ln 2}{4n}  \nonumber\\
 &\qquad\qquad\leq\frac{Cs_n\lambda_n\rho_n}{n\delta_n^2}(1+\lambda_n\rho_n^2)+\frac{C\delta_n}{s_n}+\frac{C}{n}
\end{align*}
where the constant $C$ is generic, and may change from place to place. Finally we optimize the bound choosing $\delta_n^3=  s_n^2\lambda_n\rho_n(1+\lambda_n\rho_n^2)/n$. We verify the condition $\delta_n =o(s_n)$: we have $(\delta_n/s_n)^3=O(\lambda_n\rho_n(1+\lambda_n\rho_n^2)/(ns_n))$ which, by \eqref{app-scalingRegime}, indeed tends to $0_+$ for an appropriately chosen sequence $s_n$. So the dominating term $\delta_n/s_n$ gives the result. 
\end{proof}

 \section{Proof of inequality \eqref{remarkable}}\label{proof:remarkable_id}

Let us drop the index in the bracket $\langle -\rangle_t$ and simply denote $R\equiv R_{n}(t,\epsilon)$. We start by proving the identity
\begin{align}
-2\,\mathbb{E}\big\langle Q(\mathcal{L} - \mathbb{E}\langle \mathcal{L}\rangle)\big\rangle
&=\mathbb{E}\big\langle (Q - \mathbb{E}\langle Q \rangle)^2\big\rangle
+ \mathbb{E}\big\langle (Q-   \langle Q \rangle)^2\big\rangle\,.\label{47}
\end{align}
Using the definitions $Q \equiv \frac1n \bx\cdot \bX$ and \eqref{def_L} gives
\begin{align}
2\,\mathbb{E}\big\langle Q ({\cal L} -\mathbb{E}\langle {\cal L} \rangle) \big\rangle
	 =\, &  \mathbb{E} \Big [ \frac{1}{n}\big\langle Q \|\bx\|^2 \big\rangle - 2 \langle Q^2 \rangle - \frac{1}{n\sqrt{R}}\big\langle Q (\tilde \bZ \cdot \bx) \big\rangle \Big ] \nonumber \\
	&  -  \mathbb{E}\langle Q \rangle  \, \mathbb{E} \Big [ \frac{1}{n}\big\langle \|\bx\|^2 \big\rangle - 2 \langle Q \rangle - \frac{1}{n\sqrt{R}} \tilde \bZ \cdot \langle \bx\rangle \Big ]\,. \label{eq:QL:1}
\end{align}
The gaussian integration by part formula \eqref{GaussIPP} with Hamiltonian \eqref{Ht} yields
\begin{align*}
\frac{1}{n\sqrt{R}}\mathbb{E}\big\langle Q (\tilde \bZ \cdot \bx) \big\rangle 
	& = \frac{1}{n}\mathbb{E}\big\langle Q \|\bx\|^2 \big\rangle - \frac{1}{n}\EE\big\langle Q (\bx \cdot \langle \bx \rangle) \big \rangle \overset{\rm N}{=}\frac1{n}\mathbb{E}\big\langle Q \|\bx\|^2 \big\rangle - \EE [\langle Q  \rangle^2]\,.
\end{align*}
Fort the last equality we used the Nishimori identity as follows
$$
\frac1n\EE\big\langle Q (\bx \cdot \langle \bx \rangle) \big \rangle=\frac1{n^2}\EE\big\langle (\bx\cdot \bX) (\bx \cdot \langle \bx \rangle) \big \rangle\overset{\rm N}{=} \frac1{n^2}\EE \big\langle (\bX\cdot \bx) (\bX \cdot \langle \bx \rangle) \big \rangle=\EE [\langle Q \rangle^2]\,.
$$ 
Note that we already proved \eqref{NishiTildeZ}, namely 
$$
\frac{1}{n\sqrt{R}}\mathbb{E} \langle \tilde \bZ\cdot  \bx \rangle 
=   \frac1n \mathbb{E}\big\langle \|\bx\|^2 \big\rangle - \EE \langle Q\rangle \, .
$$
Therefore \eqref{eq:QL:1} finally simplifies to 
\begin{align*}
&2\,\mathbb{E}\big\langle Q ({\cal L} -\mathbb{E}\langle {\cal L} \rangle) \big\rangle = \mathbb{E}[\langle Q\rangle^2] - 2\,\mathbb{E}\langle Q^2\rangle +  \mathbb{E}[\langle Q\rangle]^2 \nn
&\qquad\qquad= -  \big ( \mathbb{E}\langle Q^2\rangle - \mathbb{E}[\langle Q\rangle]^2 \big ) -  \big ( \mathbb{E}\langle Q^2\rangle - \mathbb{E}[\langle Q\rangle^2] \big ).
\end{align*} 
which is identity \eqref{47}.

This identity implies the inequality
\begin{align*}
2\big|\mathbb{E}\big\langle Q(\mathcal{L} - \mathbb{E}\langle \mathcal{L}\rangle)\big\rangle\big|&=2\big|\mathbb{E}\big\langle (Q-\mathbb{E}\langle Q \rangle)(\mathcal{L} - \mathbb{E}\langle \mathcal{L}\rangle)\big\rangle\big|\ge \mathbb{E}\big\langle (Q - \mathbb{E}\langle Q \rangle)^2\big\rangle
\end{align*}
and an application of the Cauchy-Schwarz inequality gives
\begin{align*}
2\big\{\mathbb{E}\big\langle (Q-\mathbb{E}\langle Q \rangle)^2\big\rangle\, \mathbb{E}\big\langle(\mathcal{L} - \mathbb{E}\langle \mathcal{L}\rangle)^2\big\rangle \big\}^{1/2}	\ge\mathbb{E}\big\langle (Q - \mathbb{E}\langle Q \rangle)^2\big\rangle\,.
\end{align*}
This ends the proof of \eqref{remarkable}.\\

 \section{Heurisitic derivation of the information theoretic phase transition}\label{sec:allOrNothing}
In this section we analyze the potential function in order to heuristically locate the information theoretic transition in the special case of the spiked Wigner model with Bernoulli prior $P_X={\rm Ber}(\rho)$. The main hypotheses behind this computation are $i)$ that the SNR $\lambda=\lambda(\rho)$ varies with $\rho$ as $\lambda=4\gamma |\ln \rho|/\rho$ with $\gamma >0$ and independent of $\rho$; that $ii)$ in this SNR regime the potential possesses only two minima $\{q^+,q^-\}$ that approach, as $\rho\to 0_+$, the boundary values $q^-=o(\rho/|\ln \rho|)$ and $q^+\to\rho$. For the Bernoulli prior the potential explicitly reads
\begin{align*}
	&i_n^{\rm pot}(q,\lambda,\rho)
	\nn
	&\equiv\frac{\lambda (q^2+\rho^2)}{4} -(1-\rho)\EE\ln\Big\{1-\rho+\rho e^{-\frac12\lambda q+\sqrt{\lambda q}Z}\Big\}-\rho\,\EE\ln\Big\{1-\rho+\rho e^{\frac12\lambda q+\sqrt{\lambda q}Z}\Big\}\,.
\end{align*}
We used that 
\begin{align}
I(X;\sqrt{\gamma} X+Z)=-\EE \ln \int dP_X(x)e^{-\frac12\gamma x^2+\gamma Xx+\sqrt{\gamma}Zx}+\frac12\EE[X^2]\gamma\,.	\label{linkPsiI}
\end{align}
Let us compute this function around its assumed minima. Starting with $q^-=o(\rho/|\ln \rho|)$ (this means that this quantity goes to $0_+$ faster than $\rho/|\ln \rho|$ as $\rho$ vanishes) we obtain at leading order after a careful Taylor expansion in $\lambda q^- \to 0_+$ (the symbol $\approx$ means equality up to lower order terms as $\rho \to 0_+$)
\begin{align}
	i_n^{\rm pot}(q^-,\lambda,\rho)
	&\approx \frac{\lambda (q^-)^2}4+\frac{\lambda\rho^2}4-\frac{\rho(\lambda q^-)^2}8\approx\frac{\lambda\rho^2}4= \gamma \rho |\ln \rho|\,.	\label{solq-}
\end{align}
For the other minimum $q^+\to\rho$, because $\lambda q^+ \to +\infty$ the $Z$ contribution in the exponentials appearing in the potential can be dropped due to the precense of the square root. We obtain at leading order
\begin{align*}
	i_n^{\rm pot}(q^+,\lambda,\rho)&\approx 2\gamma\rho|\ln \rho|-\ln\{1+\rho^{1+2\gamma}\}-\rho \ln\{1+\rho^{1-2\gamma}\}\,.
\end{align*}
Here there are two cases to consider: $\gamma>1/2$ and $0<\gamma \le 1/2$. We start with $\gamma>1/2$. In this case the potential simplifies to
\begin{align*}
	i_n^{\rm pot}(q^+,\lambda,\rho)&\approx  \rho|\ln \rho|\,. 
\end{align*}
Now for $0<\gamma\le 1/2$ we have
\begin{align*}
	i_n^{\rm pot}(q^+,\lambda,\rho)&\approx  2\gamma\rho|\ln \rho|\,.
\end{align*}
The information theoretic threshold $\lambda_c=\lambda_c(\rho)$ is defined as the first non-analiticy in the mutual information. In the present setting this corresponds to a discontinuity of the first derivative w.r.t. the SNR of the mutual information (and we therefore speak about a``first-order phase transition''). By the I-MMSE formula this threshold manifests itself as a discontinuity in the MMSE. In the high sparsity regime $\rho \to 0_+$ the transition is actually as sharp as it can be with a $0$--$1$ behavior. This translates, at the level of the potential, as the SNR threshold where its minimum is attained at $q^-$ just below and instead at $q^+$ just above. So we equate $\lim_{\rho\to 0_+}i_n^{\rm pot}(q^-,\lambda_c,\rho)=\lim_{\rho\to 0_+}i_n^{\rm pot}(q^+,\lambda_c,\rho)$ and solve for $\lambda_c$. This is only possible, under the constraint $\gamma >0$ independent of $\rho$, in the case $\gamma >1/2$ and gives $\gamma=1$ which is the claimed information theoretic threshold $\lambda_c(\rho)=4|\ln \rho|/\rho$. Repeating this analysis for the Bernoulli-Rademacher prior $P_X=(1-\rho)\delta_0+\frac12\rho(\delta_{-1}+\delta_1)$ leads the same threshold, which suggests that the transition is only related (for discrete priors) to the recovery of the support of the signal. 

Another piece of information gained from this analysis is that around the transition the mutual information divided by $n$ is $\Theta(\rho |\ln \rho|)$. Therefore the proper normalization for the mutual information is $(n \rho|\ln \rho|)^{-1}I(\bX;\bW)$ for it to have a well defined non trivial limit in the regime $\rho\to 0_+$.

Finally for $\gamma \le 1$ the minimum of the potential is attained at $q^-$ and the rescaled mutual information $(n \rho|\ln \rho|)^{-1}I(\bX;\bW)$ equals $\gamma$ as seen from \eqref{solq-}. If instead $\gamma \ge 1$ the minimum is attained at $q^+$ and the mutual information instead saturates to $1$, so we get the asymptotic singular function $\gamma \mathbb{I}(\gamma \le 1)+\mathbb{I}(\gamma \geq 1)$.

\section{Heurisitic derivation of the AMP algorithmic transition}\label{app:AMPtrans}

In this section we derive the AMP algorithmic transition for the spiked Wigner model in the Bernoulli case $P_{X,n}=\rho_n\delta_1 + (1-\rho_n)\delta_0$. The approach can be applied to the Bernoulli-Rademacher case as well (and probably more generically), and leads to the same scaling for the AMP threshold. The derivation starts from the state evolution recursion for the overlap of AMP \eqref{eq:state_evolution2}, or equivalently, $$\tau^n_{t+1}=\mathbb{E}\Big\{ X_0^n\, \mathbb{E}\big \{X_0^n \mid  \sqrt{\lambda_n} \tau^n_{t} X_0^n + \sqrt{\tau^n_t} Z\big\}\Big\}, \qquad \tau^n_{0}=0,$$ which, in the Bernoulli case, reads as (recall $Z\sim{\cal N}(0,1)$),
\begin{align}
\label{eq:tau_rep}
\tau^n_{t+1}=\mathbb{E} \left\{\frac{\rho_n^2	}{\rho_n+(1-\rho_n)\exp\{-\frac12\lambda_n \tau^n_{t}-\sqrt{\lambda_n\tau^n_{t}} Z\}}\right\}.
\end{align}
Therefore by plugging $\tau^n_{0}=0$ in the recursion we get $\tau^n_{1}=\rho_n^2$, and then 
\begin{align*}
\tau^n_{2}=\mathbb{E} \left\{\frac{\rho_n^2	}{\rho_n+(1-\rho_n)\exp\{-\frac12\lambda_n \rho_n^2-\sqrt{\lambda_n\rho_n^2} Z\}} \right\}.
\end{align*}
Now depending on $\lambda_n \rho_n^2 \gg 1$ or $\lambda_n \rho_n^2 \ll 1$ the next step of the recursion has two very different behaviors. When $\lambda_n \rho_n^2 \gg 1$, it becomes
\begin{align*}
\lambda_n \rho_n^2 \gg 1 : \qquad  \tau^n_{2} \approx \rho_n \qquad \text{and thus} \qquad \lambda_n \tau^n_{2} \approx \lambda_n \rho_n  \gg 1.
\end{align*}
Therefore, the recursion will remain stuck in this ``reconstruction state'' and converges towards $\tau^n_{\infty} \approx \rho_n$ which yields the minimal value of the MSE: $$\frac{{\rm MSE}_{\rm AMP}^{\infty}}{(n\rho_n)^2} =1 - \Big(\frac{\tau^n_{\infty}}{\rho_n}\Big)^2\approx 0.$$

When $\lambda_n \rho_n^2 \ll 1$,
\begin{align*}
\lambda_n \rho_n^2 \ll 1 :  \qquad \tau^n_{2} \approx \rho_n^2 = \tau^n_{1} \,.
\end{align*}
In this case, the recursion converges towards the ``no reconstruction state'' $\tau^n_{\infty} \approx \rho_n^2$, which corresponds to the MSE of a random guess (according to the prior) for the spike signal-matrix, i.e., the MSE corresponding to take as estimator $\bX'\otimes \bX'$ where $\bX'\sim P_{X,n}$ is independent from the ground-truth $\bX$: $$\frac{{\rm MSE}_{\rm AMP}^{\infty}}{(n\rho_n)^2}=1 - \Big(\frac{\tau^n_{\infty}}{\rho_n}\Big)^2\approx 1 \,.$$
This reasoning shows that the behavior of the state evolution must change for a scaling $\lambda_n \rho_n^2 =O(1)$. This argument cannot catch the constant $\lambda_n \rho_n^2 \approx 1/e$, which was numerically approximated in \cite{2017arXiv170100858L}.

 \section{The Nishimori identity}\label{app:nishimori}

\begin{lemma}[Nishimori identity]\label{NishId}
Let $(\bX,\bY)$ be a couple of random variables with joint distribution $P(\bX, \bY)$ and conditional distribution 
$P(\bX| \bY)$. Let $k \geq 1$ and let $\bx^{(1)}, \dots, \bx^{(k)}$ be i.i.d.\ samples from the conditional distribution. We use the bracket $\langle - \rangle$ for the expectation w.r.t. the product measure $P(\bx^{(1)}| \bY)P(\bx^{(2)}| \bY)\ldots P(\bx^{(k)}| \bY)$ and $\mathbb{E}$ for the expectation w.r.t. the joint distribution. Then, for all continuous bounded function $g$ we have 
\begin{align*}
\mathbb{E} \big\langle g(\bY,\bx^{(1)}, \dots, \bx^{(k)}) \big\rangle
=
\mathbb{E} \big\langle g(\bY, \bX, \bx^{(2)}, \dots, \bx^{(k)}) \big\rangle\,. 
\end{align*}	
\end{lemma}
\begin{proof}
This is a simple consequence of Bayes formula.
It is equivalent to sample the couple $(\bX,\bY)$ according to its joint distribution or to sample first $\bY$ according to its marginal distribution and then to sample $\bX$ conditionally on $\bY$ from the conditional distribution. Thus the two $(k+1)$-tuples $(\bY,\bx^{(1)}, \dots,\bx^{(k)})$ and  $(\bY, \bX, \bx^{(2)},\dots,\bx^{(k)})$ have the same law.	
\end{proof}

 \section{I-MMSE relation}\label{app:gaussianchannels}
In this appendix we prove the I-MMSE relation of \cite{GuoShamaiVerdu_IMMSE,guo2011estimation} for the convenience of the reader.

\begin{lemma}[I-MMSE formula]\label{app:I-MMSE}
 Consider a signal $\bX\in \mathbb{R}^n$ with $\bX\sim P_X$ that has finite support, and gaussian corrupted data $\bY\sim {\cal N}(\sqrt{R}\,\bX,{\rm I}_n)$ and possibly additional generic data $\bW\sim P_{W|X}(\cdot\,|\bX)$ with $H(\bW)$ bounded. The \emph{I-MMSE formula} linking the mutual information and the MMSE then reads
\begin{align}\label{I-MMSE-relation}
	\frac{d}{dR}I\big(\bX;(\bY,\bW)\big)=\frac{d}{dR}I(\bX;\bY|\bW) = \frac12 {\rm MMSE}(\bX|\bY,\bW)= \frac12 \EE\| \bX -\langle \bx\rangle \|^2\,,
\end{align}
where the Gibbs-bracket $\langle - \rangle$ is the expectation acting on $\bx\sim P(\cdot\,|\bY,\bW)$. 
\end{lemma}
\begin{proof}
First note that by the chain rule for mutual information $I(\bX;(\bY,\bW))=I(\bX;\bY|\bW)+I(\bX;\bW)$, so the derivatives in \eqref{I-MMSE-relation} are equal.
We will now look at $\frac{d}{dR}I(\bX;(\bY,\bW))$. Since, conditionally on $\bX$, $\bY$ and $\bW$ are independent,  we have 
\begin{align*}
I\big(\bX;(\bY,\bW)\big) = H(\bY,\bW)-H(\bY,\bW|\bX) = H(\bY,\bW) - H(\bY|\bX) - H(\bW|\bX)\, .
\end{align*}
With gaussian noise contribution $H(\bY|\bX)=\frac n2\ln(2\pi e)$. Therefore only $H(\bY,\bW)$ depends on $R$. 
Let us then compute, using the change of variable $\bY=\sqrt{R}\,\bX+\bZ$,	
\begin{align}
	&\frac{d}{dR}I\big(\bX;(\bY,\bW)\big)=\frac{d}{dR}H(\bY,\bW)&\nn
	&\ =-\frac{d}{dR} \int dP_X(\bX)d\bY d\bW P_{W|X}(\bW|\bX) \frac{e^{-\frac12\|\bY-\sqrt{R}\bX\|^2}}{(2\pi)^{n/2}}\nn
	&\qquad\qquad\times\ln \int dP_X(\bx)P_{W|X}(\bW|\bx)\frac{e^{-\frac12\|\bY-\sqrt{R}\bx\|^2}}{(2\pi)^{n/2}}\nonumber\\
	&\ =- \int dP_X(\bX)d\bZ d\bW P_{W|X}(\bW|\bX)\frac{e^{-\frac12\|\bZ\|^2}}{(2\pi)^{n/2}}\nn
	&\qquad\qquad\times\frac{d}{dR}\ln \!\int dP_X(\bx)P_{W|X}(\bW|\bx)\frac{e^{-\frac12\|\bZ-\sqrt{R}(\bx-\bX)\|^2}}{(2\pi)^{n/2}}\nonumber\\
	&\ =\frac1{2\sqrt{R}}\EE_{\bX,\bZ,\bW|\bX}\big\langle (\bZ+\sqrt{R}(\bX-\bx))\cdot (\bX-\bx)\big \rangle
	\label{IPPtoapply}
\end{align}
where $\bZ\sim{\cal N}(0,{\rm I}_n)$ and the bracket notation is the expectation w.r.t. the posterior proportional to 
$$
dP_X(\bx)dP_{W|X}(\bW|\bx)d\bZ \exp\Big\{-\frac12\|\bZ-\sqrt{R}(\bx-\bX)\|^2\Big\}\, .
$$
In \eqref{IPPtoapply} the interchange of derivative and integrals is permitted by a standard application of Lebesgue's dominated convergence theorem in the case where the support of $P_X$ is bounded. 
%
Now we use the following gaussian integration by part formula: for any bounded function ${\bm  g} :\mathbb{R}^n\mapsto \mathbb{R}^n$ of a standard gaussian random vector $\bZ\sim{\cal N}(0,{\rm I}_n)$ we obviously have
\begin{align}
\EE [\bZ\cdot {\bm  g}(\bZ)]=\EE [ \nabla_{\bZ} \cdot {\bm  g}(\bZ)]\,.
\label{SteinLemma}
\end{align}
This formula applied to a Gibbs-bracket associated to a general Gibbs distribution with hamiltonian ${\cal H}(\bx, \bZ)$ (depending on the Gaussian noise and possibly other variables) yields
\begin{align}
	\EE [\bZ\cdot \langle {\bm  h}(\bx)\rangle]&=	\EE\, \nabla_\bZ \cdot \frac{\int dP(\bx)e^{-{\cal H}(\bx,\bZ)} {\bm  h}(\bx)}{\int dP(\bx') e^{-{\cal H}(\bx',\bZ)}}\nn
	&= -\EE\,\frac{\int dP_X(\bx) e^{-{\cal H}(\bx,\bZ)} {\bm  h}(\bx)\cdot \nabla_\bZ {\cal H}(\bx,\bZ)}{\int dP_X(\bx') e^{-{\cal H}(\bx',\bZ)}}\nn
	&\hspace{2cm}+ \EE\Big[\frac{\int dP_X(\bx) e^{-{\cal H}(\bx,\bZ)} {\bm  h}(\bx)}{\int dP_X(\bx') e^{-{\cal H}(\bx',\bZ)}} \cdot \frac{\int dP_X(\bx) e^{-{\cal H}(\bx,\bZ)}\nabla_\bZ {\cal H}(\bx,\bZ) }{\int dP_X(\bx') e^{-{\cal H}(\bx',\bZ)}}\Big]\nn
	&=-\EE \big\langle {\bm  h}(\bx)\cdot \nabla_\bZ  {\cal H}(\bx,\bZ)\big\rangle+\EE \big[\big\langle {\bm  h}(\bx)\big\rangle \cdot\big\langle\nabla_\bZ  {\cal H}(\bx,\bZ)\big\rangle\big]\,.\label{GaussIPP}
	\end{align}
Applied to \eqref{IPPtoapply}, where the ``hamiltonian'' is $\mathcal{H}(\bx, \bZ)=-\ln P_{W|X}(\bW|\bx)+\frac12\|\bZ-\sqrt{R}(\bx-\bX)\|^2$, this identity gives
\begin{align*}
\frac{d}{dR}I\big(\bX;(\bY,\bW)\big)
&=\frac12\EE\big[\big\langle  \|\bX-\bx\|^2\big \rangle + \frac{1}{\sqrt{R}} \nabla_\bZ\cdot \langle \bX-\bx\rangle\big]\nonumber\\
&=\frac12\EE\big[\big\langle  \|\bX-\bx\|^2\big \rangle - \frac{1}{\sqrt{R}}\big\langle (\bX-\bx)\cdot(\bZ+\sqrt{R}(\bX-\bx))\big\rangle\nonumber \\
&\qquad\qquad+ \frac{1}{\sqrt{R}}\big\langle (\bX-\bx)\big\rangle \cdot\big\langle \bZ+\sqrt{R}(\bX-\bx)\big\rangle\big]\nonumber\\
&=\frac12\EE \|\bX-\langle \bx\rangle \|^2\,.
\end{align*}
\end{proof}

The MMSE cannot increase when the SNR increases. This translates into the concavity of 
the mutual information of gaussian channels as a function of the SNR. 

\begin{lemma}[Concavity of the mutual information in the SNR] \label{lemma:Iconcave}Consider the same setting as Lemma~\ref{app:I-MMSE}. Then the mutual informations $I(\bX;(\bY,\bW))$ and $I(\bX;\bY|\bW)$ are concave in the SNR of the gaussian channel:
\begin{align*}
	\frac{d^2}{dR^2}I\big(\bX;(\bY,\bW)\big)&=\frac{d^2}{dR^2}I(\bX;\bY|\bW) \nn
	&= \frac12 \frac{d}{dR}{\rm MMSE}(\bX|\bY,\bW)= -\frac{1}{2n}\sum_{i,j=1}^n\mathbb{E}\big[(\langle x_ix_j\rangle-\langle x_i\rangle\langle x_j\rangle)^2\big]\le 0
\end{align*}
where the Gibbs-bracket $\langle - \rangle$ is the expectation acting on $\bx\sim P(\cdot\,|\bY,\bW)$.
\end{lemma}
\begin{proof}
Set $Q\equiv \bx\cdot \bX/n$ where $\bx\sim P(\cdot\,|\bY,\bW)$. From a Nishimori identity 
${\rm MMSE}(\bX|\bY,\bW) = \EE_{P_X}[X^2]-\EE\langle Q\rangle$. Thus by the I-MMSE formula we have, by a calculation similar to \eqref{GaussIPP},
\begin{align}
-2\frac{d^2}{dR^2}I\big(\bX;(\bY,\bW)\big)=\frac{d\,\EE\langle Q\rangle}{dR}=n\EE [\langle Q\rangle\langle {\cal L}\rangle-\langle Q {\cal L}\rangle]\label{tocombine}
\end{align}
where we have set $${\cal L}\equiv \frac{1}{n}\Big(\frac{1}{2}\|\bx\|^2 - \bx\cdot \bX - \frac{1}{2\sqrt{R}}\bx\cdot \bZ\Big)\,.$$ Now we look at each term on the right hand side of this equality.
The calculation of appendix \ref{proof:remarkable_id} shows that
$$
-\EE\langle Q{\cal L}\rangle= \EE\langle Q^2\rangle-\frac12\EE[\langle Q\rangle^2]\, 
$$
so it remains to compute
\begin{align*}
\EE[\langle Q\rangle \langle {\cal L}\rangle]=\mathbb{E}\Big[\langle Q \rangle\frac{\big\langle \|\bx\|^2 \big\rangle}{2n} -  \langle Q \rangle^2 - \langle Q \rangle\frac{\bZ \cdot \langle \bx\rangle}{2n\sqrt{R}}  \Big]	\,.
\end{align*}
By formulas \eqref{SteinLemma} and \eqref{GaussIPP} in which the Hamiltonian is \eqref{Ht} we have
\begin{align*}
-\frac{1}{2n\sqrt{R}}\EE\big[ \bZ \cdot \langle \bx\rangle\langle Q \rangle \big]	&= -\frac{1}{2n\sqrt{R}}\EE\big[\langle Q\rangle \nabla_{\bZ}\cdot \langle \bx\rangle + \langle \bx\rangle\cdot \nabla \langle Q\rangle\big]\nn
&=-\frac{1}{2n}\EE\big[\langle Q\rangle \big(\big\langle \|\bx\|^2\big\rangle-\|\langle\bx\rangle\|^2\big) + \langle \bx\rangle\cdot \big(\langle Q\bx\rangle-\langle Q\rangle \langle \bx\rangle\big)\big]\nn
&\overset{\rm N}{=}-\frac{1}{2n}\EE\big[\langle Q\rangle \big\langle \|\bx\|^2\big\rangle\big]+\frac{1}{n}\EE\big[\langle Q\rangle \|\langle \bx \rangle \|^2\big]-\frac12\EE[\langle Q\rangle^2]\,.
\end{align*}
In the last equality we used the following consequence of the Nishimori identity. Let $\bx,\bx^{(2)}$ be two replicas, i.e., conditionally (on the data) independent samples from the posterior \eqref{tpost}. Then 
$$
\frac1n\EE\big[\langle \bx\rangle\cdot \langle Q \bx\rangle\big]=\frac1{n^2}\EE\big\langle (\bx^{(2)}\cdot \bx) (\bx\cdot \bX)\big\rangle\overset{\rm N}{=}\frac1{n^2}\EE\big\langle (\bx^{(2)}\cdot \bX) (\bX\cdot \bx)\big\rangle=\EE[\langle Q\rangle^2]\,.
$$
Thus we obtain 
\begin{align*}
\EE[\langle Q\rangle \langle {\cal L}\rangle -\langle Q {\cal L}\rangle] &=	\EE\langle Q^2\rangle-2\EE[\langle Q\rangle^2]+\frac{1}{n}\EE\big[\langle Q\rangle \|\langle \bx \rangle \|^2\big]\nn
&=\frac{1}{n^2}\EE\big\langle(\bx\cdot\bX)^2-2(\bx\cdot\bX)(\bx^{(1)}\cdot\bX)+(\bx\cdot\bX)(\bx^{(2)}\cdot\bx^{(3)})\big\rangle\nn
&\overset{\rm N}{=}\frac{1}{n^2}\EE\big\langle(\bx\cdot\bx^{(0)})^2-2(\bx\cdot\bx^{(0)})(\bx^{(1)}\cdot\bx^{(0)})+(\bx\cdot\bx^{(0)})(\bx^{(2)}\cdot\bx^{(3)})\big\rangle
\end{align*}
where $\bx^{(0)},\bx,\bx^{(1)},\bx^{(2)}, \bx^{(3)}$ are replicas and the last equality again follows from a Nishimori identity. Multiplying this identity by $n$ and rewriting the inner products component-wise we get
\begin{align}
\frac{d\,\EE\langle Q\rangle}{dR} &=\frac{1}{n}\sum_{i,j=1}^n\EE\big\langle x_ix_i^{(0)} x_jx_j^{(0)} -2x_ix_i^{(0)}x_j^{(1)}x_j^{(0)}+x_ix_i^{(0)}x_j^{(2)}x_j^{(3)}\big\rangle\nn
&=\frac{1}{n}\sum_{i,j=1}^n\EE\big[\langle x_ix_j\rangle^2 -2\langle x_i\rangle \langle x_j\rangle \langle x_ix_j\rangle+ \langle x_i\rangle^2\langle x_j\rangle^2\big] \label{secondDer_f_pos}
\end{align}
Using \eqref{tocombine} this ends the proof of the lemma. Note that we have also shown the positivity claimed in \eqref{jacPos} of section \ref{sec:adapInterp_XX}.
\end{proof}

 \section{Proof of corollary~\ref{Cor1:MMSEwigner}}

The proof of corollary~\ref{Cor1:MMSEwigner} follows from a combination of theorem~\ref{thm:ws} and the I-MMSE relation (see \cite{GuoShamaiVerdu_IMMSE,guo2011estimation}, and also appendix \ref{app:gaussianchannels}). Denote 
\begin{align}
{\rm M}_n(s)\equiv \frac{1}{(n\rho_n)^2}{\rm MMSE}((X_iX_j)_{i<j}\vert \bW)|_{\lambda_n=s}\quad \text{and}\quad I_n(s)\equiv \frac{1}{n\rho_n^2}I(\bX; \bW)|_{\lambda_n=s}.	
\end{align}
The I-MMSE relation in its integral formulation implies 
\begin{align}
\frac{I_n(s+\epsilon)- I_n(s)}{\epsilon}= \frac{1}{2\epsilon}\int_{s}^{s+\epsilon}{\rm M}_n(\lambda)d\lambda\,.
\end{align}
Because $M_n(s)$ is a non-increasing function (``information can't hurt'', which is equivalent to the concavity of mutual information in the signal-to-noise ratio, see \cite{GuoShamaiVerdu_IMMSE,guo2011estimation} or lemma~\ref{lemma:Iconcave}) the above identity implies
\begin{align}
	 \frac{{\rm M}_n(s+\epsilon)}{2}\le \frac{I_n(s+\epsilon)- I_n(s)}{\epsilon}\le \frac{{\rm M}_n(s)}{2}.
\end{align}
Set $$i_n(s)\equiv \frac{1}{\rho_n^{2}}\inf_{q\in [0,\rho_n]} i^{\rm pot}_n(q,s,\rho_n)\quad \text{so that}\quad \frac{1}{2}{m}_n(s,\rho_n)\equiv  \frac{d}{ds}i_{n}(s).$$ Because $s\mapsto {m}_n(s,\rho_n)$ is also a non-increasing function (see, e.g., \cite{barbier2017phase}) we obtain similarly 
\begin{align}
	 \frac{ {m}_n(s+\epsilon,\rho_n)}{2}\le \frac{i_n(s+\epsilon)- i_n(s)}{\epsilon}\le \frac{{m}_n(s,\rho_n)}{2}.
\end{align}
Set $c_n \equiv C (\ln n)^{1/3}n^{-(1-6\beta)/7}\vert \ln\rho_n \vert/\rho_n$ which is the right-hand side of \eqref{mainbound} multiplied 
by $(\rho_n \vert \ln\rho_n\vert)/\rho_n^{2}$. Theorem~\ref{thm:ws} then implies
\begin{align}
\frac{{\rm M}_n(s+\epsilon)}{2}&\le \frac{i_n(s+\epsilon)-i_n(s)+2c_n}{\epsilon}\le \frac{{m}_n(s,\rho_n)}{2}+\frac{2c_n}{\epsilon},\\
\frac{{m}_n(s+\epsilon,\rho_n)}{2}-\frac{2c_n}{\epsilon}&\le \frac{i_n(s+\epsilon)-i_n(s)-2c_n}{\epsilon}\le \frac{{\rm M}_n(s)}{2}.
\end{align}
Replacing $\rho_n = \Omega(n^{-\beta})$ with $\beta\in [0, 1/13)$ yields the claimed inequality:
\begin{align}
	{m}_n(s+\epsilon,\rho_n)-\frac{C^\prime}{\epsilon} \frac{(\ln n)^{4/3}}{n^{(1-13\beta)/7}} \le {\rm M}_n(s) \le {m}_n(s-\epsilon,\rho_n)+\frac{C^\prime}{\epsilon} \frac{(\ln n)^{4/3}}{n^{(1-13\beta)/7}}.
\end{align}

\section{AMP algorithmic phase transition}\label{app:matrixWalgo}

In this appendix, we prove theorem~\ref{AMP-theorem}.  To do this, we begin by introducing a general `symmetric' AMP algorithm  in section~\ref{appsec:finitesample} and show it is quite similar to the AMP algorithm in \eqref{eq:AMP}.  For this symmetric AMP algorithm, we provide finite sample guarantees like those given in \cite{RushVenkataramanan} for  various `non-symmetric' AMP algorithms. However, we have an added challenge in that terms like the Lipschitz constant of the denoiser $f_t$ in \eqref{eq:denoiser} and the state evolution values in \eqref{eq:state_evolution2} depend on $n$ and therefore cannot be treated as universal constants in the rate of concentration, as they were in \cite{RushVenkataramanan}.  The main concentration result for the symmetric AMP is given in theorem~\ref{thm:sym} in section~\ref{appsec:finitesample}. Then, in section~\ref{appsec:our_AMP_result}, we use theorem~\ref{thm:sym} to prove result \eqref{eq:finite_sample}, from which we prove theorem~\ref{AMP-theorem}.

\subsection{Symmetric AMP finite sample guarantees}  \label{appsec:finitesample}

We begin by analyzing a `symmetric'  AMP algorithm, described now, that is similar to the AMP algorithm in \eqref{eq:AMP}. Assume the matrix $\bZ \sim \textsf{GOE}(n)$ is an $n \times n$ matrix form the gaussian orthogonal ensemble, i.e.\ $\bZ$ is a symmetric matrix with $\{Z_{ij}\}_{1 \leq i \leq j \leq n}$ i.i.d.\ $\mathcal{N}(0,1/n)$, and $\{Z_{ii}\}_{1 \leq i  \leq n}$ i.i.d.\ $\mathcal{N}(0,2/n)$.  Start with an initial condition $\bh^0 \in \mathbb{R}^n$, independent of $\bZ$, and calculate for $t \geq 0$,
\begin{equation}
\bh^{t+1} = \bZ g_t(\bh^t, \bX^n) - \textsf{c}_t g_{t-1}(\bh^{t-1}, \bX^n).
\label{eq:sym_AMP}
\end{equation}
In the above, $g_t: \mathbb{R}^2 \to \mathbb{R}$ is Lipschitz and separable (i.e.\ it acts component-wise when applied to vectors) and may depend on $n$ through its Lipschitz constant, denoted $L_g^n$. The function $g_t$ takes
as its second argument  a random vector $\bX^n \in \mathbb{R}^n$ with entries that are i.i.d.\ $p_{X,n}$, a sub-gaussian distribution, $\textsf{c}_t =  \frac{1}{n} \sum_{i=1}^n g'_t(h^{t}_i, X^n_i),$ the derivative taken with respect to the first argument, and all terms with negative indices take the value $0$ (so that, for example, $\bh^{1} = \bZ g_0(\bh^0, \bX^n)$). The key result, stated in theorem~\ref{thm:sym} below, is that for each $t \geq 1$, the empirical distribution of the components of $\bh^t$ is approximately equal in distribution to a gaussian $\mathcal{N}(0,\sigma^n_t)$ where the variances $\{\sigma^n_t\}_{t \geq 0}$ are defined via the state evolution: initialize with $\sigma^n_1 = \|g_0(\bh^0, \bX^n)\|^2/n$, calculate for $t\geq 1$,
\begin{equation}
\sigma^n_{t+1} = \mathbb{E}\left[\left(g_{t}(\sqrt{\sigma^n_{t}}Z, X_0^n) \right)^2\right],
\label{eq:symm_SE}
\end{equation}
where the expectation is with respect to standard gaussian $Z$ independent of $X_0^n \sim p_{X,n}$.

Before stating theorem~\ref{thm:sym} below, we give the assumptions on the model and the functions used to define the AMP.  In what follows, $C, c > 0$ are generic positive constants whose values are not exactly specified but do not depend on $n$.

 \textbf{Random Vectors:} The random vector $ \bX^n \in \mathbb{R}^n$ used in the denoising functions, is assumed to have entries that are i.i.d.\ according to a sub-gaussian distribution $p_{X,n}$, in particular, $p_{X,n}$ is ${\rm Ber}(\rho_n)$ or Bernoulli-Rademacher, where sub-gaussian  random variables are defined in lemma~\ref{lem:subgauss}.

\textbf{The function $g_t$}: The denoising function $g_t: \mathbb{R}^2 \rightarrow \mathbb{R}$ in \eqref{eq:sym_AMP} is defined as $g_t(\bh^t, \bX^n) = f_t(\bh^t + \sqrt{\lambda_n} \sigma^n_{t} \bX^n)$ where $f_t$ is the conditional expectation denoiser in \eqref{eq:denoiser}. With this definition, $g_t$ is
separable and Lipschitz continuous for each $t \geq 0$, with Lipschitz constant denoted $L^n_g >0$, that depends on $n$. By the Lipschitz property, $g_t$ are weakly differentiable and the weak derivatives, denoted by $g'_t$, are also differentiable. 

\begin{thm} \label{thm:sym}
Consider the AMP algorithm in \eqref{eq:sym_AMP} for $\bh^0 \in \mathbb{R}^m$ independent of $\bA$ under the assumptions above.   Then, for any (order-$2$) pseudo-Lipschitz function $\phi: \mathbb{R}^2 \rightarrow \mathbb{R}$,  $\ep \in (0,1)$, and  $t \geq 1$. 
 \begin{equation}
 \mathbb{P} \Big(\Big| \frac{1}{n} \sum_{i=1}^n \phi(h^{t}_i, X^n_i) -  \mathbb{E}\Big\{\phi\big(\sqrt{\sigma^n_t} Z, X^n_0\big) \Big\} \Big| \geq \epsilon \Big) \leq C C_t \exp\Big\{ \frac{-c c_t n \epsilon^2}{L_{\phi}^2  \widetilde{\gamma}_n^{t}}\Big\},
 \label{eq:sym_bound}
 \end{equation}
where the expectation is with respect to standard gaussian $Z$ independent of $X_0^n \sim p_{X,n}$, the state evolution values $\sigma^n_t$ are defined in \eqref{eq:symm_SE}, the constants $C_t, c_t$ are defined in theorem~\ref{AMP-theorem}, and 
\begin{align}
\widetilde{\gamma}_n^{t+1} &:= \lambda_n^{2t}  (\nu^n +  \sigma^n_{1})  (\nu^n +  \sigma^n_{1} +  \sigma^n_{2})  \cdots  (\nu^n + \sum_{i=1}^{t+1} \sigma^n_{i}) \max\{1, \hat{\textsf{c}}_1\}  \max\{1, \hat{\textsf{c}}_2\} \cdots \max\{1, \hat{\textsf{c}}_{t}\},
\label{eq:tilde_gamma}
\end{align}
where $\hat{\textsf{c}}_t = \mathbb{E}[ g'_t(\sqrt{\sigma^n_t} Z, X)]$ and $\nu$ is the variance factor of sub-gaussian $X^n \sim p_{X,n}$, which equals $12 \rho_n$ for $p_{X, n} \sim {\rm Ber}(\rho_n)$ (see lemma~\ref{lem:subgauss}).
\end{thm}

The proof of theorem~\ref{thm:sym} is given in section~\ref{appsec:thmproof}.  The proof relies heavily on the proof of the finite sample guarantees for various `non-symmetric' AMP algorithms given in \cite[theorem 1]{RushVenkataramanan} and we reference this result throughout. We will use theorem~\ref{thm:sym}  to prove theorem~\ref{AMP-theorem}, but before doing so, we make a few remarks about extensions of the result and the major differences between theorem~\ref{thm:sym} and the finite sample guarantees in \cite{RushVenkataramanan}.

\textbf{Remark 1: Spectral initialization.}
 We assume that the AMP iteration in \eqref{eq:sym_AMP} was initialized with $\bh^0 \in \mathbb{R}^n$ independent of $\bZ$. As mentioned previously in section~\ref{sec:matrixWalgo}, for estimation with a Bernoulli-Rademacher signal prior, one needs to instead use a spectral initialization that will not be independent of the matrix $\bZ$. Theoretically, as introduced in \cite{montanari2017estimation}, one deals with this dependency by analyzing the AMP iteration using a matrix $\widetilde{\bZ}$ that is an \emph{approximate} representation of the conditional distribution of $\bZ$ given the initialization and then showing that the two algorithms are close each other  with high probability.  We do not give the details of this rather technical argument here, and instead analyze the simpler case using an independent initialization, though the generalization is likely straightforward.

\textbf{Remark 2: Rate of the concentration.}
The rate of concentration depends on $\lambda_n$, $\rho_n$, and the state evolution values, $\sigma_t^n$, through $\widetilde{\gamma}_n^{t}$ defined in \eqref{eq:tilde_gamma}.  In particular, the term $\lambda_n^{2(t-1)}$ in $\widetilde{\gamma}_n^{t}$, appears through the dependency of the rate on the Lipschitz constant of  $g_t$, where $g_t(\bh^t, \bX^n) = f_t(\bh^t + \sqrt{\lambda_n} \sigma^n_{t} \bX^n)$ and $f_t$ is the conditional expectation denoiser in \eqref{eq:denoiser}. With this definition, $L_g^n = \sqrt{\lambda_n}$. 
The dependence on these values was not stated explicitly in the concentration bound of \cite[theorem 1]{RushVenkataramanan} as the authors assume that the Lipschitz constant, sparsity, and state evolution terms do not change with $n$ and, thus, can be absorbed into the universal constants. 

The presence of these terms in our rate comes from the inductive portion of the proof where one must show that the values $\|g_t(\bh^t, \bX^n)\|^2/n$ concentrate to known constants. Essentially, this step will add a term $(L_g^n)^2 (\nu^n + \sigma^n_{t})$ in the rate at each step of the induction.  To see this, we point the reader to three facts. First, notice that the approximate distribution of $h^t_i$ is gaussian with variance $\sigma^n_{t}$. Second, it is easy to see that a function $[f(x)]^2$ has the same pseudo-Lipschitz constant as $f(\cdot)$ if $|f(\cdot)|$ is bounded (as $n$ grows), which is the case for $g_t$ in our setting. (More generally, the Lipschitz constant of $[f(x)]^2$ will be no more than $L_f^2$.) Finally, we highlight that pseudo-Lipschitz functions taking gaussian and sub-gaussian input concentrate as in \cite[Lemma B.4]{RushVenkataramanan} with $L^2(\nu^n +  \sigma^n_{t})$ in the denominator of the rate where $L$ is the associated pseudo-Lipschitz constant, $\nu^n$ is the sub-gaussian variance factor, and $\sigma^n_{t}$ is the gaussian variance. Indeed, we restate \cite[Lemma B.4]{RushVenkataramanan} here for clarity.
\begin{lemma}{\cite[lemma B.4]{RushVenkataramanan}}
\label{lem:PLsubgaussconc}
Let $Z \in \mathbb{R}^n$ be an i.i.d.\ standard gaussian vector and $G \in  \mathbb{R}^n$ a random vector with entries  $G_1, \ldots, G_n$  i.i.d.\   $\sim p_{G}$, where $p_G$ is sub-gaussian with variance factor $\nu^n$.   Then, for any pseudo-Lipschitz function $f: \mathbb{R}^{2} \to \mathbb{R}$ with constant $L_f^n$, non-negative values $\sigma^n$, and  $0< \ep \leq 1$, 
\begin{align*}
& \mathbb{P}\Big(\Big\lvert \frac{1}{n}\sum_{i=1}^n f(\sqrt{\sigma^n} Z_{i}, G_i)  - \mathbb{E}[f(\sqrt{\sigma^n} Z, G)] \Big \lvert \geq \ep \Big)  \hspace{-2pt} \leq \hspace{-2pt} 2 \exp \Big\{ \frac{ - \kappa n \ep^2}{(L_f^n)^2[ \nu^n + 4 (\nu^n)^2 + \sigma^n +  4  (\sigma^n)^{2}] }  \Big\}.
\end{align*}
\end{lemma}
Since $0 \leq \nu^n,  \sigma^n  \leq 1$ we drop the squared terms $(\nu^n)^2, (\sigma^n)^{2}$ from the rate since $\nu^n, \sigma^n$ dominate.

\textbf{Remark 3: Denoisers}
The proof of \cite[theorem 1]{RushVenkataramanan}  assumes that the weak derivative of the denoiser, $g'_t$, has bounded derivative everywhere it exists. Here, $g_t(\bh^t, \bX^n) = f_t(\bh^t + \sqrt{\lambda_n} \sigma^n_{t} \bX^n)$ where $f_t$ is the conditional expectation denoiser in \eqref{eq:denoiser} and  $f'_t$ is given in lemma~\ref{lem:PL}. In particular, $f_t'(x) = \sqrt{\lambda_n} f_t(x) (1-f_t(x))$, which is not bounded (in $n$) since $\lambda_n$ grows with $n$.  However, we can show that $f_t'(x)$ is also Lipschitz, with constant $L_f^2 = \lambda_n$,  and we use this fact directly in the proof to get around the boundedness assumption originally used in \cite[theorem 1]{RushVenkataramanan}.

\subsection{Proving theorem~\ref{AMP-theorem}}  \label{appsec:our_AMP_result}

Before we get to the proof of theorem~\ref{AMP-theorem}, we discuss how we apply the result of theorem~\ref{thm:sym} to our problem. This will lead to the concentration result in \eqref{eq:finite_sample},
which concerns convergence within pseudo-Lipschitz loss functions of the empirical distribution of $x^t_i$, the iterate of the AMP algorithm in \eqref{eq:AMP}, to its approximating distribution with mean and variance determined by the state evolution. Recall the following definition of a pseudo-Lipschitz function. 
 \begin{defi}  \label{def:PL}
 For any $n,m \in \mathbb{N}_{>0}$, a function $\phi : \mathbb{R}^n\to \mathbb{R}^m$ is \emph{pseudo-Lipschitz of order $2$} if there exists a constant $L>0$ such that  $\|\phi(\bx)-\phi(\by)\|\leq L\left(1+ \|\bx\|+ \|\by\|\right) \|\bx-\by\|$ for $\bx, \by \in \mathbb{R}^n$.
 \end{defi}
Now we prove \eqref{eq:finite_sample}. Recall that in our model~\eqref{WSM},
\begin{equation}\label{WSM_app}
\frac{1}{\sqrt{n}}\bW = \frac{\sqrt{\lambda_n}}{n} \bX\otimes \bX +  \bZ\,,
\end{equation}
where $\lambda_n >0$ controls the strength of the signal and the noise is i.i.d.\ gaussian $Z_{ij}\sim{\cal N}(0,1/n)$ for $i<j$ and symmetric, $Z_{ij}=Z_{ji}$. The AMP algorithm for recovering $\bX$ from the data $\bW$ is given in \eqref{eq:AMP}.

Notice that the AMP algorithm in \eqref{eq:AMP} is similar to \eqref{eq:sym_AMP}, the only difference being that the matrix $\bA$ in \eqref{eq:AMP} is our data matrix, as opposed to it being $\textsf{GOE}(n)$ as in \eqref{eq:sym_AMP}. If we plug the value of $\bW$ from \eqref{WSM_app} into \eqref{eq:AMP}, we find the following iteration: $\bx^{1} = \frac{\sqrt{\lambda_n}}{n}  \bX \langle \bX, f_0(\bx^{0}) \rangle  +  \bZ  f_0(\bx^{0})$, and for $t \geq 1$,
\begin{equation}
\bx^{t+1} = \frac{\sqrt{\lambda_n}}{n}  \bX \langle \bX, f_t(\bx^{t}) \rangle  +  \bZ  f_t(\bx^{t}) - \mathsf{b}_t  f_{t-1}(\bx^{t-1}).
\label{eq:AMP3}
\end{equation}

Now  we define a related iteration to \eqref{eq:AMP3} as follows. Initialize with $\bh^0 = \bx^0$ with denoiser $g_{0}(\bh^{0}, \bX) := f_0(\bx^0)$ and $\bh^{1} = \bZ g_{0}(\bh^{0}, \bX)$.  
Then calculate for $t \geq 1$,
\begin{equation}
\bh^{t+1} =   \bZ  g_t(\bh^{t}, \bX^n) - \mathsf{c}_t  g_{t-1}(\bh^{t-1}, \bX^n), \qquad \mathsf{c}_t  = \frac{1}{n} \sum_{i=1}^n g'_t(h^{t}_i, X^n_i) ,
\label{eq:AMP_correct}
\end{equation}
where $g_t(h, X) = f_t(h + \mu_{t}^n X^n)$ for $f_t(\cdot)$ the conditional expectation denoiser used in \eqref{eq:AMP3} and $ \mu_{t+1}^n$ calculated from the state evolution in \eqref{eq:state_evolution} for our original iteration (i.e., the algorithm in \eqref{eq:AMP}).  In the above, $\bX^n$ is the signal in \eqref{WSM_app} and we drop the $n$ superscript in what follows.  Then, the iteration in \eqref{eq:AMP_correct} takes the exact form of the symmetric AMP in \eqref{eq:sym_AMP} with state evolution given by \eqref{eq:symm_SE}.  In particular, the state evolution associated with \eqref{eq:AMP_correct} is $\tau^n_1 = \|g_0(\bh^0, \bX)\|^2/n$ and for $t \geq 1,$
\begin{equation}
\tau^n_{t+1} = \mathbb{E}\left[\left(g_{t}(\sqrt{\sigma^n_{t}}Z, X) \right)^2\right]= \mathbb{E}\left[\left(f_{t}(\sqrt{\sigma^n_{t}}Z + \mu^n_t X) \right)^2\right].
\label{eq:symm_SE_new}
\end{equation}
The above state evolution is exactly the state evolution for the AMP algorithm in \eqref{eq:AMP3} defined  in \eqref{eq:state_evolution2}. 
For this reason, we used the $\tau$ notation.

As the AMP algorithm in \eqref{eq:AMP_correct} takes the exact form of the symmetric AMP in \eqref{eq:sym_AMP}, we can apply theorem~\ref{thm:sym}.  The proof idea is to use theorem~\ref{thm:sym} to give performance guarantees to the algorithm in \eqref{eq:AMP_correct} and then to argue that the algorithm in \eqref{eq:AMP3} is asymptotically equivalent to the algorithm in \eqref{eq:AMP_correct} so the performance guarantees hold for \eqref{eq:AMP3} as well.

We apply theorem~\ref{thm:sym} to \eqref{eq:AMP_correct} using  the pseudo-Lipschitz function  $\phi(h^{t}_i, X^n_i) =  \psi(h^{t}_i+ \mu_{t}^n X^n_i, X^n_i)$, where $\psi$ is the order $2$ pseudo-Lipschitz function in \eqref{eq:finite_sample},  to find  that for $t \geq 1$,
 \begin{equation}
 \begin{split}
& \mathbb{P} \Big(\Big| \frac{1}{n} \sum_{i=1}^n \psi(h^{t}_i+ \mu_{t}^n X_i, X_i) -  \mathbb{E}\Big\{\psi\big(\sqrt{\tau^n_t} Z+ \mu_{t}^n X, X\big) \Big\} \Big| \geq \epsilon \Big) \leq C C_{t} \exp\Big\{\frac{-c c_{t} n \epsilon^2}{L_{\psi}^2   \widetilde{\gamma}_n^{t}}\Big\}. 
  \label{eq:sym_result}
  \end{split}
 \end{equation}
%
%
We have used that $L_{\phi} = 2L_{\psi} (1+\mu_{t}^n)^2$, which is shown in lemma~\ref{lem:PL_second},
and that  $L_{\phi} =2 L_{\psi} (1+\mu_{t}^n)^2 \leq \kappa L_{\psi}$, which follows from the fact that $ \mu_{t}^n \leq \kappa'$ in the regime of interest, as discussed, for example, in \eqref{eq:mu_bound} in section~\ref{appsec:lemmaproof}.


To show how \eqref{eq:finite_sample} follows from~\eqref{eq:sym_result},  we use the following lemma. 
%
 \begin{lemma} \label{lem:aux}
 Define $\textsf{bound}_t := C C_{t} \exp\Big\{\frac{-c c_{t} n \epsilon^2}{L_{\psi}^2\widetilde{\gamma}_n^{t}}\Big\},$ for $  \widetilde{\gamma}_n^{t}$ in \eqref{eq:tilde_gamma}.  Let $\bh^{t}$ be defined by the algorithm in \eqref{eq:AMP_correct} and $ \bx^{t}$ by \eqref{eq:AMP3} Then for $t \geq 1$, the following are true
 \begin{align}
 & \mathbb{P} \Big( \frac{ 1 }{\sqrt{n}}  \| \bh^{t} + \mu_{t}^n \bX  \| \geq \kappa_h \Big) \leq C C_{t} e^{\frac{-c c_{t} n}{\widetilde{\gamma}_n^{t}}} , \quad   \mathbb{P} \Big(   \frac{1}{n}\sum_{i=1}^n h^{t}_i + \mu_{t}^n X_i\geq \frac{\mu_t^n}{2} \Big) \leq C C_{t} e^{\frac{-c c_{t} n }{\rho_n^{-2} \widetilde{\gamma}_n^{t}}}, \label{eq:lemmaboundh} \\
 &        \mathbb{P} \Big( \frac{1}{n}   \Big\| \bx^{t} - \bh^{t} - \mu_{t}^n \bX  \Big\|^2 \geq \frac{\kappa \epsilon^2}{ L^2_{\psi}} \Big)\leq \textsf{bound}_t,\label{eq:lemmabound3} \\
&    \mathbb{P} \Big( \frac{1}{\sqrt{n}}\| \bx^{t}  \| \geq \kappa_x \Big) \leq C C_{t} e^{\frac{-c c_{t} n}{\widetilde{\gamma}_n^{t}}}, \qquad \mathbb{P} \Big( \frac{1}{n}  \sum_{i=1}^n x^{t}_i \geq \frac{\mu_t^n}{2} \Big) \leq C C_{t} e^{\frac{-c c_{t} n}{ \rho_n^{-2} \widetilde{\gamma}_n^{t}}},\label{eq:lemmaboundx}  \\
&       \mathbb{P} \Big(\Big| \frac{1}{n} \sum_{i=1}^n \psi(X_i, x^{t}_i)  - \psi(X_i, h^{t}_i+ \mu_{t}^n X_i) \Big| \geq \epsilon \Big)\leq \textsf{bound}_t,  \label{eq:lemmabound2}\\
 & \mathbb{P} \Big(\Big| \frac{1}{n} \sum_{i=1}^n \psi(X_i, x^{t}_i) -  \mathbb{E}\Big\{\psi\big(X_0, \mu^n_{t} X_0 + \sqrt{\tau^n_t} Z\big) \Big\} \Big| \geq \epsilon \Big)  \leq \textsf{bound}_t. \label{eq:lemmabound1}
   \end{align}
   In \eqref{eq:lemmaboundh} and \eqref{eq:lemmaboundx} both $\kappa_h$ and $\kappa_x$ are universal constants.
 \end{lemma}
 %
The proof of lemma~\ref{lem:aux} is rather long and technical, so we include it in full detail at the end of the appendix in section~\ref{appsec:lemmaproof} and give a high level sketch here.  

The basic idea behind the proof of lemma~\ref{lem:aux} is that the results in \eqref{eq:lemmaboundh} follow from the fact that $h^{t}_i + \mu_{t}^n X_i \approx \sqrt{\tau_t^n} Z + \mu_{t}^n X$ for $X \sim p_{X,n}$ independent of $Z$ standard gaussian by theorem~\ref{thm:sym}.  Thus, $ \frac{1}{n}  \sum_{i=1}^n h^{t}_i  + \mu_{t}^n X_i$ concentrates on $\mu_{t}^n  \rho_n$.  Similarly, $\frac{1}{\sqrt{n}} \| \bh^{t} + \mu_{t}^n \bX  \|$ will concentrate to $\tau_t^n + ( \mu_{t}^n)^2 \rho_n$. Then we use concentration to imply boundedness with high probability.
The result \eqref{eq:lemmaboundx} follows from the same ideas since it can be shown that $x^{t}_i \approx \sqrt{\tau_t^n} Z + \mu_{t}^n X$. 

Next,  results  \eqref{eq:lemmabound2} and \eqref{eq:lemmabound1} follow immediately from \eqref{eq:lemmaboundh}--\eqref{eq:lemmabound3}.  To see this, first notice that \eqref{eq:lemmabound1} follows directly from the bound in \eqref{eq:sym_result} and \eqref{eq:lemmabound2} using lemma~\ref{sums}.
Next, \eqref{eq:lemmabound2} follows from results \eqref{eq:lemmaboundh} -- \eqref{eq:lemmabound3}.  This can be seen by using the following upper bound due to Cauchy-Schwarz,
 \begin{equation}
 \begin{split}
 \label{eq:CS_split1}
&\Big| \frac{1}{n} \sum_{i=1}^n \psi(X_i, x^{t}_i)  - \psi(X_i, h^{t}_i+ \mu_{t}^n X_i) \Big| \leq \frac{1}{n} \sum_{i=1}^n \Big|\psi(X_i, x^{t}_i)  - \psi(X_i, h^{t}_i+ \mu_{t}^n X_i) \Big|\\
 & \leq \frac{L_{\psi}}{n} \sum_{i=1}^n  \Big(1 + \|(X_i, x^{t}_i)\| + \|(X_i, h^{t}_i+ \mu_{t}^n X_i)\| \Big)  \Big| x^{t}_i - h^{t}_i-  \mu_{t}^n X_i \Big|\\
 & \leq \frac{\kappa L_{\psi}}{\sqrt{n}}  \left\| \bx^{t} - \bh^{t}- \mu_{t}^n \bX  \right\| \sqrt{  \Big(1 + \frac{2}{n} \|\bX\|^2 + \frac{1}{n}\|\bx^{t}\|^2 + \frac{1}{n}\|\bh^{t}+ \mu_{t}^n \bX\|^2 \Big) },
   %
\end{split}
 \end{equation}
 and the boundedness of the term $\|\bX\|^2/n$.  
Thus, using $\kappa_B = \sqrt{1  + 4 + \kappa_x^2 + \kappa_h^2} > 0$, a universal constant, by the above bound it follows that
 \begin{equation}
 \begin{split}
 \label{eq:tobebounded}
 \mathbb{P}& \Big(\Big| \frac{1}{n} \sum_{i=1}^n \psi(X_i, x^{t}_i)  - \psi(X_i, h^{t}_i+ \mu_{t}^n X_i) \Big| \geq \epsilon \Big)\\%
 & \leq \mathbb{P} \Big( \frac{1}{\sqrt{n}}   \Big\| \bx^{t} - \bh^{t} - \mu_{t}^n \bX  \Big\|\sqrt{  \Big(1 + \frac{2}{n} \|\bX\|^2 + \frac{1}{n}\|\bx^{t}\|^2 + \frac{1}{n}\|\bh^{t}+ \mu_{t}^n \bX\|^2 \Big) }  \geq \frac{\kappa \epsilon}{ L_{\psi}} \Big) \\
 & \leq \mathbb{P} \Big( \frac{1}{\sqrt{n}}   \Big\| \bx^{t} - \bh^{t} - \mu_{t}^n \bX  \Big\| \geq \frac{\kappa \epsilon}{\kappa_B L_{\psi}} \Big) + \mathbb{P} \Big(\frac{2}{n} \|\bX\|^2  \geq 2(1+ \rho_n) \Big)  \\
 & \hspace{3cm}+  \mathbb{P} \Big(  \frac{1}{n}\|\bx^{t}\|^2   \geq \kappa^2_x \Big) +  \mathbb{P} \Big(  \frac{1}{n}\|\bh^{t}+ \mu_{t}^n \bX\|^2  \geq \kappa^2_h \Big).
 \end{split}
 \end{equation}
Note, we have used $\rho_n \leq 1$ so $\frac{1}{\kappa_B^2}\left(1 + 2(1 + \rho_n) +\kappa_x^2 + \kappa_h^2\right) \leq 1$.
Considering the result in \eqref{eq:tobebounded}, we notice that result \eqref{eq:lemmabound2} follows directly from \eqref{eq:lemmaboundh}--\eqref{eq:lemmabound3}, since by Chernoff's bound (lemma~\ref{lem:chernoff}),
\begin{equation*}
 \mathbb{P} \Big(\frac{2}{n} \|\bX\|^2  \geq 2(1+ \rho_n)  \Big)  \leq  \mathbb{P} \Big(\Big \lvert \frac{1}{n} \|\bX\|^2 - \rho_n  \Big \lvert \geq 1 \Big)  \leq 2\exp\left\{  \frac{-n}{3 \rho_n}\right\}.
\end{equation*}
Thus, \eqref{eq:lemmabound2} (hence, \eqref{eq:lemmabound1},) follows easily from \eqref{eq:lemmaboundh}--\eqref{eq:lemmabound3} and the main technical piece of proving lemma~\ref{lem:aux} is then proving results \eqref{eq:lemmaboundh}--\eqref{eq:lemmabound3} rigorously.  This is done in section~\ref{appsec:lemmaproof}.


Now that we show that \eqref{eq:finite_sample}  follows from lemma~\ref{lem:aux} result \eqref{eq:lemmabound1}, and then we finally prove theorem~\ref{AMP-theorem}. Notice that \eqref{eq:finite_sample} is recovered by applying \eqref{eq:lemmabound1} with pseudo-Lipschitz function $ \widetilde{\psi}(X_i, x^{t}_i) =  \psi(X_i, f_t(x^{t}_i))$,
as the only difference between \eqref{eq:finite_sample} and  \eqref{eq:lemmabound1} is that  $x^t_i$ in \eqref{eq:finite_sample} is replaced with $f_t(x^t_i)$ in \eqref{eq:lemmabound1}. With this choice of pseudo-Lipschitz function, an $L_f^2$ term is added in the denominator of the rate of concentration, since  
$L_{\widetilde{\psi}} = 3  L_{\psi}  \max\{1, L_f\}$, which is shown in lemma~\ref{lem:PL_second}.


We note that from \eqref{eq:gamma_def} and \eqref{eq:tilde_gamma} it is easy to see that 
$
 \widetilde{\gamma}_n^{t}  L_f^2 = \widetilde{\gamma}_n^{t}  \lambda_n 
 = \gamma_n^{t},
$
 noting that $\hat{\textsf{c}}_t = \mathbb{E}[ g'_t(\sqrt{\tau^n_t} Z, X)] =  \mathbb{E}[ f'_t(\sqrt{\tau^n_t} Z + \mu^n_t X)] = \hat{\textsf{b}}_t$.  Moreover, the bound on the RHS of \eqref{eq:finite_sample} equals $\textsf{bound}_t$ defined in theorem~\ref{AMP-theorem} when $L_{\psi}$ is a universal constant.

 
 Now we prove theorem~\ref{AMP-theorem} using \eqref{eq:finite_sample}.
First, notice that theorem~\ref{AMP-theorem} result \eqref{eq:finite_sample_vector1} follows directly from \eqref{eq:finite_sample} using pseudo-Lipschitz function $\psi(X_i, f_t(x^{t}_i))  = (X_i -  f_{t}(x^{t}_i))^2$. This function is pseudo-Lipschitz with constant $L_{\psi} $ by lemma~\ref{lem:PL}.  To see how this proves result \eqref{eq:finite_sample_vector1} in more details, notice that
 \begin{equation*}
 \frac{1}{n} \sum_{i=1}^n \psi(X_i, f_{t}(x^{t}_i))  = \frac{1}{n} \sum_{i=1}^n (X_i - f_{t}(x^{t}_i))^2 = \frac{1}{n}\|\bX -  f_{t}(\bx^{t}) \|^2,
  \end{equation*}
and
 \begin{equation*}
 \begin{split}
& \mathbb{E}\Big\{\psi\big(X_0^n,  f_{t}(\mu^n_{t} X_0^n + \sqrt{\tau^n_t} Z)\big) \Big\} =  \mathbb{E}\Big\{\big(X_0^n - f_t\big(\mu^n_{t}  X_0^n - \sqrt{\tau^n_t} Z\big)\big)^2 \Big\} \\
 &=  \mathbb{E}\Big\{\big(X_0^n\big)^2 \Big\} + \mathbb{E}\Big\{\big[ f_t\big(\mu^n_{t}  X_0^n - \sqrt{\tau^n_t} Z)\big]^2 \Big\} - 2 \mathbb{E}\Big\{ X_0^n f_t\big(\mu^n_{t}  X_0^n - \sqrt{\tau^n_t} Z\big) \Big\} = \rho_n -  \tau^n_{t+1},
 \end{split}
  \end{equation*}
where the final uses that $ \mathbb{E}\{(X_0^n)^2\} = \rho_n$ when $P_{X,n}$ is ${\rm Ber}(\rho_n)$ or Bernoulli-Rademacher, the Law of Total Expectation to give $ \mathbb{E}\{X_0^n f_t(\sqrt{\lambda_n} \tau^n_{t} X_0^n + \sqrt{\tau^n_t} Z)\} = \mathbb{E}\{[f_t(\sqrt{\lambda_n} \tau^n_{t}X_0^n + \sqrt{\tau^n_t} Z)]^2\}$ in the case of the conditional expectation denoiser as in \eqref{eq:denoiser}, and the state evolution definition in \eqref{eq:state_evolution2}.


Now we prove theorem~\ref{AMP-theorem} result \eqref{eq:finite_sample_matrix1}. 
Now considering the concentration result in \eqref{eq:finite_sample_matrix1}, notice that
 \begin{equation*}
 \frac{1}{n^2}\|\bX \bX^T - f_t(\bx^t) [f_t(\bx^t)]^T\|_F^2  =  \frac{1}{n^2}\|\bX\|^4 + \frac{1}{n^2}\|f_t(\bx^t)\|^4 - \frac{2}{n^2} \langle \bX, f_t(\bx^t) \rangle^2.
  \end{equation*}
Then we will prove the following three results: for $\textsf{bound}_t$ defined in the  theorem~\ref{AMP-theorem} statement,
 \begin{align}
\mathbb{P} \Big(\Big|  \frac{1}{n^2}\|\bX\|^4 -  \rho_n^2 \Big| \geq \epsilon \Big) &\leq 2e^{ -2 n \e^2},  \label{eq:finite_sample_matrix01} \\
\mathbb{P} \Big(\Big|  \frac{1}{n^2}\|f_t(\bx^t)\|^4 - (\tau^n_{t})^2 \Big| \geq \epsilon \Big) &\leq \textsf{bound}_t,  \label{eq:finite_sample_matrix02} \\
\mathbb{P} \Big(\Big|  \frac{1}{n^2} \langle \bX, f_t(\bx^t) \rangle^2 -  ( \tau^n_{t})^2 \Big| \geq \epsilon \Big) &\leq \textsf{bound}_t.\label{eq:finite_sample_matrix03}
 \end{align}
Then the final concentration result in \eqref{eq:finite_sample_matrix1} follows from lemma~\ref{sums} as follows:
 \begin{align*}
&\mathbb{P} \Big(\Big| \frac{1}{n^2}\|\bX \bX^T - f_t(\bx^t) [f_t(\bx^t)]^T\|_F^2 - ( \rho_n^2 -  (  \tau^n_{t+1})^2)  \Big| \geq \epsilon \Big) \\
&= \mathbb{P} \Big(\Big|  \Big(\frac{1}{n^2}\|\bX\|^4 -  \rho_n^2 \Big) + \Big( \frac{1}{n^2}\|f_t(\bx^t)\|^4 - (\tau^n_{t+1})^2   \Big)-  \Big(\frac{2}{n^2} \langle \bX, f_t(\bx^t) \rangle^2  - 2 (   \tau^n_{t+1})^2\Big) \Big| \geq \epsilon \Big) \\
&\leq \mathbb{P} \Big(\Big|  \frac{1}{n^2}\|\bX\|^4 -  \rho_n^2 \Big| \geq \frac{\epsilon}{3} \Big) +  \mathbb{P} \Big(\Big|  \frac{1}{n^2}\|f_t(\bx^t)\|^4 - (\tau^n_{t+1})^2  \Big| \geq  \frac{\epsilon}{3} \Big)\\
&\qquad\qquad +  \mathbb{P} \Big(\Big|  \frac{2}{n^2} \langle \bX,f_t(\bx^t) \rangle^2  - 2 (  \tau^n_{t+1})^2 \Big| \geq  \frac{\epsilon}{3} \Big).
 \end{align*}
As a final step, notice that the bounds in  \eqref{eq:finite_sample_matrix01} - \eqref{eq:finite_sample_matrix03} applied to the above give the result in \eqref{eq:finite_sample_matrix1}.


Now we prove  \eqref{eq:finite_sample_matrix01} - \eqref{eq:finite_sample_matrix03}. First we prove \eqref{eq:finite_sample_matrix01} using Heoffding's Inequality, lemma~\ref{lem:hoeff_lem},
 \begin{align*}
\mathbb{P} \Big(\Big|  \frac{1}{n}\|\bX\|^2 -  \rho_n \Big| \geq \epsilon \Big)  = \mathbb{P} \Big(\Big|  \frac{1}{n} \sum_{i=1}^n(X_i^2 -  \mathbb{E}\{X_i^2\}) \Big| \geq \epsilon \Big)  &\leq 2e^{ -2 n \ep^2}.
 \end{align*}
Then the result in  \eqref{eq:finite_sample_matrix01} then follows from the above by lemma~\ref{powers}.

 
 Next, for \eqref{eq:finite_sample_matrix02} we apply \eqref{eq:finite_sample} using the function $ \psi(X_i, f_t(x^t_i)) = [f_t(x^t_i)]^2$, which is pseudo-Lipschitz with constant $L_{\psi} = 2$ by lemma~\ref{lem:PL}), to find
 \begin{equation*}
 \mathbb{P} \Big(\Big|  \frac{1}{n}\|f_t(\bx^t)\|^2 - \tau^n_{t+1} \Big| \geq \epsilon \Big) = \mathbb{P} \Big(\Big| \frac{1}{n} \sum_{i=1}^n [f_t(x^t_i)]^2 -\tau^n_{t+1} \Big| \geq \epsilon \Big) \leq \textsf{bound}_t,
 \end{equation*}
 where we have used the definition of the state evolution in \eqref{eq:state_evolution2} to give
 \begin{equation*}
 \begin{split}
  \mathbb{E}\Big\{\psi\big( X_0^n, f_t\big(\mu^n_{t} X_0^n + \sqrt{\tau^n_{t}} Z\big)\big) \Big\} &=   \mathbb{E}\Big\{\big[f_{t}\big(\mu^n_{t} X_0^n + \sqrt{\tau^n_{t}} Z\big)\big]^2 \Big\} = \tau^n_{t+1}.
  \end{split}
 \end{equation*}
   Then the result in  \eqref{eq:finite_sample_matrix02} follows from the above by lemma~\ref{powers} and the fact that $\tau^n_{t} \leq \rho_n$.
 

Finally we prove result \eqref{eq:finite_sample_matrix03} by applying \eqref{eq:finite_sample} using the function $ \psi(X_i, f_t(x^t_i)) = X_i f_t(x^t_i)$, which is pseudo-Lipschitz with constant $L_{\psi} = 2$ by lemma~\ref{lem:PL}, to find
 \begin{equation*}
 \mathbb{P} \Big(\Big|  \frac{1}{n} \langle \bX, f_t(\bx^t) \rangle -    \tau^n_{t+1} \Big| \geq \epsilon \Big) = \mathbb{P} \Big(\Big| \frac{1}{n} \sum_{i=1}^n  X_i f_t(x^t_i) -    \tau^n_{t+1} \Big| \geq \epsilon \Big) \leq \textsf{bound}_t,
 \end{equation*}
where 
 \begin{equation*}
 \begin{split}
  \mathbb{E}\Big\{\psi\big( X_0^n, f_t\big(\mu^n_{t} X_0^n + \sqrt{\tau^n_{t}} Z\big)\big) \Big\} =  \mathbb{E}\Big\{ X_0^n f_{t}\big(\mu^n_{t} X_0^n + \sqrt{\tau^n_{t}} Z\big) \Big\} &=   \tau^n_{t+1},
  \end{split}
 \end{equation*}
where the final step uses the Law of Total Expectation to give $ \mathbb{E}\{X_0^n f_{t}(\mu^n_{t} X_0^n + \sqrt{\tau^n_{t}} Z)\} = \mathbb{E}\{[f_{t}(\mu_{t}X_0^n + \sqrt{\tau^n_{t}} Z)]^2\}$ in the case of the conditional expectation denoiser as in \eqref{eq:denoiser} and the state evolution definition in \eqref{eq:state_evolution2}.
   Then the result in  \eqref{eq:finite_sample_matrix03} follows  from the above  by lemma~\ref{powers} $\tau^n_{t} \leq \rho_n$.

\subsection{Proof of theorem~\ref{thm:sym}} \label{appsec:thmproof}

\begin{proof}
The proof of theorem~\ref{thm:sym} proceeds in two steps.  In the first step, one studies the conditional distribution of $\bZ$ given the  output of the algorithm up until iteration $t$, treating $\bZ$ as random and the output as deterministic.  In the non-symmetric AMP studied in \cite[Theorem 1]{RushVenkataramanan}, the relevant measurement matrix has i.i.d.\ gaussian entries and this conditional distribution was  originally studied in \cite{bayati2011dynamics}.  The result for the case of i.i.d.\ gaussian $\bZ$ is concisely stated in \cite[ Lemma 4.2]{RushVenkataramanan}. For the symmetric AMP of \eqref{eq:sym_AMP} that we are interested in, the matrix $\bZ$ is $\textsf{GOE}(n)$ and so this conditioning argument needs to take into account the symmetry of the matrix entries (and consequently the added dependencies).  This has been studied in other works that give asymptotic characterizations of the performance of symmetric AMP, for example in \cite[Lemma 3]{Montanari-Javanmard}, and these results apply directly to our case since this distributional characterization is already non-asymptotic and does not change in our setting.  This then allows us to characterize the conditional distribution of the iterates $\bh^{t+1}$, conditional on the previous output of the algorithm. We give this result in Lemma~\ref{lem:hb_cond} below, but before stating the lemma, we introduce some useful notation.

\newcommand{\norm}[1]{\lVert#1\rVert}
\newcommand{\mscrs}{\mathscr{S}}
\newcommand{\bH}{\mathbf{H}}
\newcommand{\bDel}{\boldsymbol{\Delta}}
\newcommand{\bigM}{\mathbf{M}}

First, denote $\mathbf{m}^0 := g_0(\bh^0, \bX^n), ..., \mathbf{m}^t := g_t(\bh^t, \bX^n)$ where the terms $g_t(\bh^t, \bX^n)$ are those used in the symmetric AMP in \eqref{eq:sym_AMP}.  Then we define $\mathscr{S}_{0}$ to be the sigma-algebra generated by $\{g_0(\bh^0, \bX^n), \bX^n\}$ and
$\mathscr{S}_{t}$ for $ t\geq 1$ to be the sigma-algebra generated by
\begin{equation*}
\bh^1, ..., \bh^{t}, \mathbf{m}^0, ...,  \mathbf{m}^t, \text{ and }  \bX^n.
\end{equation*}
Using \cite[Lemma 3]{Montanari-Javanmard} to characterize the distribution of $\bZ$ conditioned on the sigma algebra $\mathscr{S}_{t}$, we are able to specify the conditional distributions of $\bh^{t+1}$ given $\mathscr{S}_{t}$, by observing that conditioning on $\mscrs_{t}$ for $t \geq 1$ is equivalent to conditioning on the linear constraint\footnote{While conditioning on the linear constraints, we emphasize that only $\bA$ is treated as random.}
\begin{equation*}
\bZ \bigM_{t} = \bY_{t},
\end{equation*}
where $\bigM_t \in \mathbb{R}^{n \times t}$ and $\bH_t \in \mathbb{R}^{n \times t}$ are the matrices 
\begin{equation*}
\bigM_{t} = [\mathbf{m}^0  \mid  ...  \mid  \mathbf{m}^{t-1}] \quad \text { and } \quad \bH_{t} = [\bh^1  \mid  ...  \mid  \bh^{t}],
\end{equation*}
and $\bY_t \in \mathbb{R}^{n \times t}$ is the matrix $\bY_{1} = \bH_1$ and $\bY_{t} = \bH_t + [\mathbf{0} |\bigM_{t-1}] \boldsymbol{\textsf{C}_t}$ for $t \geq 2$,  where 
$ \boldsymbol{\textsf{C}_t} =  \text{diag}(\textsf{c}_0, ..., \textsf{c}_{t-1})$.
Note that $[\bc_1 \mid \bc_2 \mid ... \mid \bc_k]$ denotes a matrix with columns $\bc_1, ..., \bc_k$.

We use the notation $\mathbf{m}^{t+1}_{\|}$ to denote the projection of $\mathbf{m}^{t+1}$ onto the column space of $\bigM_{t+1}$. Let
 \begin{equation}
 \boldsymbol{\alpha}^{t+1} := (\alpha^{t+1}_0, \alpha^{t+1}_1, \ldots, \alpha^{t+1}_{t})^\intercal \in \mathbb{R}^{t+1},
 \label{eq:vec_alph_gam_conc}
\end{equation}
 be the coefficient vectors of these projections, i.e.,
$\mathbf{m}^{t+1}_{\|} := \sum_{i=0}^{t} \alpha^t_i \mathbf{m}^i$,
meaning $ \boldsymbol{\alpha}^t = (\bigM_{t+1}^\intercal \bigM_{t+1})^{-1}\bigM_{t+1}^\intercal \mathbf{m}^{t+1}.$
 The projections of $\mathbf{m}^{t+1}$onto the orthogonal complement of $\bigM_{t+1}$, is denoted by
$\mathbf{m}^{t+1}_{\perp} := \mathbf{m}^{t+1} - \mathbf{m}^{t+1}_{\|}.$
 Lemma \ref{lem:main_lem}  shows that for large $n$, the entries of $\boldsymbol{\alpha}$  concentrate around constants. In what follows we show that,  
 for $t \geq 0$, the vector $ \boldsymbol{\alpha}^{t+1} \in \mathbb{R}^{t+1}$ in \eqref{eq:vec_alph_gam_conc} concentrates to the vector
 \begin{equation}
\hat{\boldsymbol{\alpha}}^{t+1} := \left [ 0 , \ldots , 0 , \frac{\sigma^n_{t+2}}{\sigma^n_{t+1}}\right]^\intercal \in \mathbb{R}^{t+1},
 \label{eq:hatalph_hatgam_def}
\end{equation}
for the state evolution values given in \eqref{eq:symm_SE}.  Similarly,  Lemma \ref{lem:main_lem} will show that for large $n$, the norm $\|\mathbf{m}^{t-1}_{\perp} \|^2/n$ concentrates to a constant $\sigma_{t}^{\perp}$,  defined as $\sigma_1^{\perp} = \sigma_1^n$, and for $t \geq 2,$ 
 \begin{equation}
\sigma_t^{\perp} := \sigma_t^n \Big(1 - \frac{\sigma_t^n }{\sigma_{t-1}^n}\Big).
 \label{eq:sigperp_defs}
\end{equation}

With the above notation, we find the following result for the symmetric AMP in \eqref{eq:sym_AMP}.

\begin{lemma}[Conditional Distribution Lemma]
For the vectors $\bh^{t+1}$ defined in \eqref{eq:sym_AMP}, the following hold for $t \geq 1$, provided $n >t$ and $\bigM_{t}^\intercal \bigM_t$ has full column rank.
 \begin{equation}
 \begin{split}
 \bh^{1}  \lvert_{\mscrs_{0}}   \stackrel{d}{=}  \sqrt{\sigma^n_1}  \, \bU_0 + \bDel_{0}, \qquad \text{ and } \qquad  \bh^{t+1} \lvert_{\mscrs_{t}} &  \stackrel{d}{=} \hat{\alpha}^{t+1}_t  \, \bh^{t} + \sqrt{\sigma^{\perp}_{t+1}} \, \bU_t + \bDel_{t}, \label{eq:Ha_dist}
 \end{split}
\end{equation}
where $\bU_0, \bU_t \in \mathbb{R}^n$ are random vectors with elements that are marginally standard gaussian random variables that are independent of the corresponding conditioning sigma-algebras. The terms $\hat{\alpha}^t_{i}$ for $i \in \{0,1,..., t\}$ are defined in \eqref{eq:hatalph_hatgam_def} and the terms $\sigma_{t}^{\perp}$ in \eqref{eq:sigperp_defs}.  The deviation terms are $\bDel_{t} = \mathbf{0}$ and for $t >0$,
\begin{align}
&\bDel_{t} =   \sum_{r=1}^{t} (\alpha^{t+1}_r - \hat{\alpha}^{t+1}_r)\bh^r   +   \Big[ \Big( \frac{\| \mathbf{m}^t_{\perp} \|}{\sqrt{n}} - \sqrt{\sigma_{t+1}^{\perp}}\Big)\mathsf{I}  -  \frac{\| \mathbf{m}^t_{\perp} \|}{\sqrt{n}}\mathsf{P}^{\parallel}_{\bigM_{t}} \Big]\bU_t  \nonumber  \\
&+     \bigM_{t-1} (\bigM_{t-1}^\intercal\bigM_{t-1})^{-1} \Big[\bH_{t-1}^\intercal   \mathbf{m}^t_{\perp} -  \bigM_{t-1}^*\Big(\textsf{c}_{t} \mathbf{m}^{t-1} - \sum_{i=1}^{t-1} \textsf{c}_{i} \alpha^t_{i} \mathbf{m}^{i-1} \Big)\Big].\label{eq:Dt1t}  
\end{align} 
\label{lem:hb_cond}
\end{lemma}

The second step of the proof is inductive on the iteration $t$, showing that if the result in \eqref{eq:sym_bound} holds up to iteration $t-1$  then it will hold at iteration $t$ as well. This is done by showing that the standardized $\ell_2$ norms of the terms in the AMP algorithm in \eqref{eq:sym_AMP}, like $\frac{1}{n}\|\bh^t\|^2$ and $\frac{1}{n}\|\mathbf{m}^t\|^2$, concentrate on deterministic values predicted by the state evolution. This is done by relating these iteration $t$ values to the iteration $t-1$ values through the iteration in \eqref{eq:sym_bound} and appealing to the conditional distributions from Lemma~\ref{lem:hb_cond} and your inductive hypothesis.  While the proof for the symmetric AMP is largely similar to that for the non-symmetric version, there are additional challenges due to the dependencies created by the symmetry of the $\textsf{GOE}(n)$ matrix.  The details of the inductive proof are quite technical and long and are therefore not included here but we state the result for the symmetric AMP in \eqref{eq:sym_AMP}.

For $t \geq 0$, let $\kappa_{-1} = K_{-1} = 1$, and 
\begin{equation}
\begin{split}
K_t = C (t+1)^5K_{t-1},  &\quad  \kappa_t  = \frac{\kappa_{t-1}}{  c(t+1)^{11}},
\end{split}
\label{eq:Kkappa_def}
\end{equation}
where $C, c > 0$ are universal constants (not depending on $t$, $n$, or $\e$).  To keep the notation compact,  we  use $K, \kappa, \kappa'$ to denote generic positive universal constants whose values may change through  the  lemma statement.

%

The result of theorem~\ref{thm:sym} follows from lemma~\ref{lem:main_lem} result \eqref{eq:Hb1} below.

\begin{lemma}
The following statements hold for $1 \leq t < T^*$ and $\epsilon\in (0,1)$. Define
\begin{equation}
\gamma_n^{t+1} := (\nu^n +  \sigma^n_{1})  (\nu^n +  \sigma^n_{1} +  \sigma^n_{2})  \cdots  (\nu^n + \sum_{i=1}^{t+1} \sigma^n_{i}) \times \max\{1, \hat{\textsf{c}}_1\}  \max\{1, \hat{\textsf{c}}_2\} \cdots  \max\{1, \hat{\textsf{c}}_t\},
\end{equation}
where $\nu$ is the variance factor of sub-gaussian $\bX^n$ which equals $\kappa \rho_n$ for $p_{X, n}$ Bernoulli.
\begin{enumerate}
\item 
\begin{equation}
\mathbb{P}\Big(\frac{1}{n}\norm{\bDel_{{t}}}^2 \geq \epsilon\Big) \leq  K (t+1)^2 K_{t-1} \exp\Big\{-\frac{\kappa \kappa_{t-1} n \epsilon}{ (t+1)^4  L_{g}^{4t}  \gamma_n^{t}  \max\{1, \hat{\textsf{c}}_t\} }\Big\}. \label{eq:Ba} 
\end{equation}

\item 
Denote $\mathbb{E}_{\phi} :=  \mathbb{E}\, \phi_h(\sqrt{\sigma^n_{1}} \tilde{Z}_{1}, \ldots ,\sqrt{\sigma^n_{t+1}} \tilde{Z}_{t+1}, X^n)$.  Then for pseudo-Lipschitz function $\phi: \mathbb{R}^{t+1} \rightarrow \mathbb{R}$ with constant $L_{\phi}$ we have that
\begin{equation}
\mathbb{P}\Big(\Big \lvert \frac{1}{n}\sum_{i=1}^n \phi(h^1_i, \ldots,  h^{t+1}_i, X^n_{i}) - \mathbb{E}_{\phi} \Big \lvert \geq \ep\Big) \leq K (t+1)^3 K_{t-1} \exp\Big\{-\frac{ \kappa \kappa_{t-1} n \epsilon^2}{(t+1)^7 L_{\phi}^{2} L_{g}^{4t}  \gamma_n^{t+1} }\Big\}.
  \label{eq:Hb1}
\end{equation}
The random variables $\tilde{Z}_{0}, \ldots, \tilde{Z}_t$ are jointly gaussian with zero mean and covariance given by $\mathbb{E}[\tilde{Z}_r \tilde{Z}_t] = \sqrt{{\sigma^n_t}/{\sigma^n_r}}$ for $r < t$, and are independent of $\bX^n \sim p_{X, n}$.

Denote $\mathbb{E}_{\phi_{Lip}} :=  \mathbb{E}\, \phi_{Lip}(\sqrt{\sigma^n_{1}} \tilde{Z}_{1}, \ldots ,\sqrt{\sigma^n_{t+1}} \tilde{Z}_{t+1}, X^n)$.  Then for Lipschitz function $\phi_{Lip}: \mathbb{R}^{t+1} \rightarrow \mathbb{R}$ with constant $L_{\phi_{Lip}}$ we have that
\begin{align}
&\mathbb{P}\Big(\Big \lvert \frac{1}{n}\sum_{i=1}^n m^0_i \phi_{Lip}(h^1_i, \ldots,  h^{t+1}_i, X^n_{i}) - \widetilde{\sigma} \mathbb{E}_{\phi_{Lip}}\Big \lvert \geq \ep\Big) \nonumber\\
&\qquad \leq  K  (t+1)^3 K_{t-1} \exp\Big\{-\frac{ \kappa \kappa_{t-1} n \epsilon^2}{ (t+1)^7 L_{\phi}^{2} L_{g}^{4t}  \gamma_n^{t+1} }\Big\}.
 \label{eq:Bb2}
\end{align}
The random variables $\tilde{Z}_{0}, \ldots, \tilde{Z}_t$ are as above.




\item Let $L_g > 0$ be the pseudo-Lipschitz constant for the denoiser functions $\{g_t\}_{ t\geq 0}$ and let $X_n \overset{\mathbf{. .}}{=} c$ be shorthand for $$\mathbb{P}(\abs{X_n -c} \geq \epsilon) \leq K (t+1)^3 K_{t-1} \exp\Big\{-\frac{ \kappa \kappa_{t-1} n \epsilon^2}{(t+1)^7 L_{g}^{4t+2} \gamma_n^{t+1} }\Big\}.$$ 
For all $0 \leq r \leq t$, 
\begin{align*}
&\mathbb{P}\Big(\Big \lvert \frac{1}{n}(\bh^{r+1})^* \bh^{t+1}- \sigma^n_{t+1} \Big \lvert \geq \epsilon \Big) \leq K (t+1)^3 K_{t-1} \exp\Big\{-\frac{ \kappa \kappa_{t-1} n \epsilon^2}{(t+1)^7L_{g}^{4t} \gamma_n^{t+1} }\Big\}. \\ 
& \frac{1}{n}(\mathbf{m}^{0})^* \mathbf{m}^{t+1} \overset{\mathbf{. .}}{=} \widetilde{\sigma}^n \mathbb{E}[g_{t+1}(\sqrt{\sigma^n_{t+1}} \tilde{Z}_{t+1}, X^n)],  \quad
\frac{1}{n}(\mathbf{m}^{r+1})^*\mathbf{m}^{t+1}  \overset{\mathbf{. .}}{=} \sigma^n_{t+2}. \\ 
&\mathbb{P}\Big(\Big \lvert \textsf{c}_{t+1} - \hat{\textsf{c}}_{t+1} \Big \lvert \geq \ep\Big) \leq K (t+1)^3 K_{t-1} \exp\Big\{-\frac{ \kappa \kappa_{t-1} n \epsilon^2}{(t+1)^7 L_{g}^{4(t+1)} \gamma_n^{t+1} }\Big\}. \\
& \frac{1}{n}(\bh^{t+1})^*\mathbf{m}^{r+1}  \overset{\mathbf{. .}}{=} \hat{\textsf{c}}_{r+1}  \sigma^n_{t+2}, \qquad
 \frac{1}{n}(\bh^{r+1})^*\mathbf{m}^{t+1} \overset{\mathbf{. .}}{=}  \hat{\textsf{c}}_{t+1} \sigma^n_{t+2}. 
\end{align*}


\item 
\begin{equation}
 \mathbb{P}\left( \frac{1}{n} \bigM_{t+1}^* \bigM_{t+1} \text{  is singular}\right) \leq (t+1) K_{t-1} \exp\Big\{-\frac{\kappa_{t-1} \kappa n}{{(t+1)^7 L_{g}^{4t+2}\gamma_n^{t+1} }}\Big\}. \label{eq:Msing}
\end{equation}
For $\hat{\alpha}^{t+1}$ defined in \eqref{eq:hatalph_hatgam_def}, when the inverse of $\frac{1}{n} \bigM_{t+1}^* \bigM_{t+1} $ exists, for $1\leq i, j \leq t+1$,
\begin{equation}
\begin{split}
& \mathbb{P}\Big( \Big \lvert \Big[(\frac{1}{n} \bigM_{t+1}^* \bigM_{t+1})^{-1} - (\mathbf{C}^{t+1})^{-1} \Big]_{i,j} \Big \lvert \geq \epsilon \Big) \leq K  K_{t-1} \exp\Big\{-\frac{\kappa \kappa_{t-1}n\epsilon^2}{L_{g}^{4t+2} \gamma_n^{t+1} } \Big\},  \\
&  \mathbb{P}\Big( \lvert\alpha^{t+1}_{i-1} - \hat{\alpha}^{t+1}_{i-1} \lvert \geq \epsilon\Big) \leq K (t+1)^4 K_{t-1}   \exp\Big\{-\frac{\kappa \kappa_{t-1}n\epsilon^2}{ (t+1)^9 L_{g}^{4(t+1)} \gamma_n^{t+1} } \Big\}.
\end{split}
\label{eq:Hg}   
\end{equation}
In the above, the matrix $\mathbf{C}^{t+1} \in \mathbb{R}^{(t+1) \times (t+1)}$ has elements 
$[\mathbf{C}^{t+1}]_{i,j} =  \sigma^n_{\max\{i,j\}}$ for $1 \leq i,j \leq t+1$.
\item With $\sigma_{t}^{\perp}$ defined in \eqref{eq:sigperp_defs},
\begin{equation}
\mathbb{P}\Big( \Big \lvert\frac{1}{n}\norm{\mathbf{m}^{t+1}_{\perp}}^2 - \sigma_{t+2}^{\perp} \Big \lvert \geq \epsilon\Big) \leq K (t+1)^5 K_{t-1}   \exp\Big\{-\frac{\kappa \kappa_{t-1}n\epsilon^2}{ (t+1)^{11} L_{g}^{4(t+1)} \gamma_n^{t+1} } \Big\}. \label{eq:Hh}
\end{equation}
\end{enumerate}
\label{lem:main_lem}
\end{lemma}

\end{proof}

\subsection{Proof of lemma~\ref{lem:aux}} \label{appsec:lemmaproof}

\begin{proof}

To begin with, we prove result \eqref{eq:lemmaboundh} then we prove the other results, \eqref{eq:lemmaboundx}--\eqref{eq:lemmabound1}, inductively.


\paragraph{Result \eqref{eq:lemmaboundh}.}  We first show that \eqref{eq:lemmaboundh} follows immediately from theorem~\ref{thm:sym}.  Before we do so we establish upper and lower bounds on $\tau_n^t$ defined in  \eqref{eq:state_evolution2}.  Notice that for the Bernoulli case, 
\begin{align*}
\tau^n_{t+1}=\mathbb{E}\Big[\frac{\rho_n^2}{\rho_n+(1-\rho_n)\exp\{-\frac12\lambda_n \tau^n_{t}-\sqrt{\lambda_n\tau^n_{t}} Z\}} \Big],
\end{align*}
where $Z \sim \mathcal{N}(0,1)$, as shown in appendix \ref{app:AMPtrans} result \eqref{eq:tau_rep}.  Therefore, trivially
$
\tau^n_{t+1}\leq \rho_n.
$
We also wish to establish a lower bound.  First, by Jensen's Inequality applied to the convex function $f(x) = 1/x$ on $x \in (0, \infty)$, we have that
\begin{align*}
\tau^n_{t+1} \geq \frac{\rho_n^2}{\rho_n+(1-\rho_n) \mathbb{E}\Big[\exp\{-\frac12\lambda_n \tau^n_{t}-\sqrt{\lambda_n\tau^n_{t}} Z\}\Big]} \overset{(a)}{=} \rho_n^2,
\end{align*}
%
where step $(a)$ uses that 
$\mathbb{E}[\exp\{-\frac12\lambda_n \tau^n_{t}-\sqrt{\lambda_n\tau^n_{t}} Z\}] = 1$ since $\mathbb{E}[\exp\{-t Z\}] = \exp\{\frac12 t^2\}$.
%
Thus, $\rho_n^2 \leq \tau^n_{t+1} \leq \rho_n$ and, in the regime of interest where $\lambda_n = \kappa \rho_n^{-2}$, using that $\mu^n_t = \sqrt{\lambda_n} \tau^n_t$ by \eqref{eq:state_evolution2}, we find
$(\mu^n_t)^2 = \lambda_n (\tau^n_t)^2 = \kappa \rho_n^{-2} (\tau^n_t)^2$ and therefore 
\begin{equation}
\label{eq:mu_bound}
\kappa' \rho_n^{2} \leq  (\mu^n_t)^2 =  \kappa \rho_n^{-2} (\tau^n_t)^2 \leq \kappa.
\end{equation}

Now we demonstrate \eqref{eq:lemmaboundh}.  Using theorem~\ref{thm:sym} with pseudo-Lipschitz function $\phi(h^{t}_i, X^n_i) = (\mu_t^n)^{-1} h^{t}_i +  X^n_i$, having constant $L_{\phi} = \sqrt{2}\max\{1, (\mu_t^n)^{-1}\}$ as is shown in lemma~\ref{lem:PL_second}, 
 \begin{equation}
 \label{eq:hbound1}
 \mathbb{P} \Big(\Big| \frac{1}{n} \sum_{i=1}^n  \frac{h^{t}_i }{\mu_t^n }+ X^n_i - \rho_n \Big| \geq \epsilon \Big) \leq C C_t \exp\Big\{ \frac{-c c_t n \epsilon^2}{  \max\{1, (\mu_t^n)^{-2}\}  \widetilde{\gamma}_n^{t}}\Big\} \leq C C_t \exp\Big\{ \frac{-c c_t n \epsilon^2}{  \rho_n^{-2} \widetilde{\gamma}_n^{t}}\Big\},
 \end{equation}
where the final inequality follows from the bound $ \max\{1, (\mu_t^n)^{-2}\} \leq \kappa' \rho_n^{-2}$ justified above in \eqref{eq:mu_bound}. We have also used that $\mathbb{E}[\phi(\sqrt{\tau_t^n} Z, X_0^n)] = \mathbb{E}[ (\sqrt{\tau_t^n} /\mu_t^n)Z + X_0^n] = \mathbb{E}[X_0^n] = \rho_n$.
The result in \eqref{eq:lemmaboundh} follows from \eqref{eq:hbound1} since, when $\rho_n \leq 1/4$,
 \begin{equation*}
 \begin{split}
  &  \mathbb{P} \Big( \frac{1}{n}  \sum_{i=1}^n  h^{t}_i  + \mu_t^nX_i \geq \frac{\mu_t^n}{2} \Big)  =   \mathbb{P} \Big( \frac{1}{n}  \sum_{i=1}^n  \frac{h^{t}_i }{\mu_t^n }  + X_i \geq \frac{1}{2} \Big) \\
    &\qquad \leq  \mathbb{P} \Big( \Big\lvert \frac{1}{n}  \sum_{i=1}^n   \frac{h^{t}_i }{\mu_t^n }  + X_i -  \rho_n \big \lvert \geq \frac{1}{2} - \rho_n \Big) \leq  \mathbb{P} \Big( \Big\lvert \frac{1}{n}  \sum_{i=1}^n \frac{h^{t}_i }{\mu_t^n }   +  X_i -  \rho_n \big \lvert \geq \frac{1}{4} \Big).
   \end{split}
  \end{equation*}

Similarly, for the first result in \eqref{eq:lemmaboundh}, we use the pseudo-Lipschitz function $\phi(h^{t}_i, X^n_i) = (h^{t}_i + \mu_t^n X^n_i)^2$, having constant $L_{\phi} = 2\max\{1, (\mu_t^n)^2\}$,  as is shown in lemma~\ref{lem:PL_second}. Then by theorem~\ref{thm:sym},
 \begin{equation}
  \label{eq:hbound2}
 \mathbb{P} \Big(\Big| \frac{1}{n} \| \bh^{t} + \mu_t^n \bX^n\|^2 - (\mu_t^n)^2 \rho_n - \tau_t^n \Big| \geq \epsilon \Big) \leq C C_t \exp\Big\{ \frac{-c c_t n \epsilon^2}{\max\{1, (\mu_t^n)^4\}  \widetilde{\gamma}_n^{t}}\Big\} \leq \textsf{bound}_t,
 \end{equation}
where the final inequality follows since $(\mu^n_t)^4 \leq \kappa$ as discussed above in \eqref{eq:mu_bound}. We have also used that $\mathbb{E}[\phi(\sqrt{\tau_t^n}  Z, X_0^n)] = \mathbb{E}[ (\sqrt{\tau_t^n} Z +  \mu_t^n X_0^n)^2] = \tau_t^n + ( \mu_t^n)^2 \rho_n$.
Then, since $(\mu_t^n)^2 \rho_n \leq 1$, choosing $ \kappa_h^2  = 3 >  1 + (\mu_t^n)^2 \rho_n +\tau_t^n$, we recover the first result in \eqref{eq:lemmaboundh} with
 \begin{equation*}
   \mathbb{P} \Big( \frac{1}{\sqrt{n}}   \Big\| \bh^{t} + \mu_{t}^n \bX  \Big\| \geq \kappa_h \Big)  \leq   \mathbb{P} \Big(  \Big\lvert \frac{1}{n}   \Big\| \bh^{t} + \mu_{t}^n \bX \Big\|^2 - (\mu_t^n)^2 \rho_n - \tau_t^n \Big \lvert  \geq 1\Big).
  \end{equation*}

\paragraph{Other results \eqref{eq:lemmaboundx}--\eqref{eq:lemmabound1}.}
The proof is inductive on the iteration $t$. We first show the initialization case $t=1$. Consider \eqref{eq:lemmabound3}, then using the definitions of $\bx^{1} = \frac{\sqrt{\lambda_n}}{n}  \bX \langle \bX, f_0(\bx^{0}) \rangle  +  \bZ  f_0(\bx^{0})$  from \eqref{eq:AMP3} and $\bh^{1} =  \bZ g_{0}(\bh^{0}, \bX)$ from \eqref{eq:AMP_correct} along with the fact that $f_{0}(\bx^{0}) = g_{0}(\bh^{0}, \bX)$, 
\begin{equation*}
\bx^{1} - \bh^{1} - \mu_{1}^n \bX = \bX \Big(\frac{\sqrt{\lambda_n}}{n}  \langle \bX, f_0(\bx^{0}) \rangle  - \mu_{1}^n\Big) = 0,
\end{equation*}
where the final inequality follows since $ \mu^n_{1} = \sqrt{\lambda_n} \langle f_0(\bx^0), \bX \rangle/n$  by \eqref{eq:state_evolution_start}. Next for result \eqref{eq:lemmaboundx}, first notice that by the Triangle Inequality, $ \| \bx^{1}  \| \leq \| \bx^{1} - \bh^{1} - \mu_{1}^n \bX \| + \|  \bh^{1} + \mu_{1}^n \bX  \|$.  Then let $ \kappa_x = 2 \kappa_h + 2\kappa$ and therefore, by lemma~\ref{sums},
\begin{equation*}
\begin{split}
    \mathbb{P} \Big( \frac{1}{\sqrt{n}}   \| \bx^{1}  \| \geq \kappa_x \Big)  &=  \mathbb{P} \Big( \frac{1}{\sqrt{n}} \| \bx^{1} - \bh^{1} - \mu_{1}^n \bX  \| +   \frac{1}{\sqrt{n}} \|  \bh^{1} + \mu_{1}^n \bX  \|  \geq 2 \kappa_h + 2\kappa \Big)  \\
&\leq  \mathbb{P} \Big( \frac{1}{\sqrt{n}}  \| \bx^{1} - \bh^{1} - \mu_{1}^n \bX  \|  \geq  \kappa \Big)  +  \mathbb{P} \Big(   \frac{1}{\sqrt{n}} \|  \bh^{1} + \mu_{1}^n \bX \|  \geq  \kappa_h  \Big).
\end{split}
\end{equation*}
Then the upper bound follows by \eqref{eq:lemmaboundh} and \eqref{eq:lemmabound3}.  

Next, notice that with the bound on $   \mathbb{P} ( \frac{1}{\sqrt{n}}   \| \bx^{1}  \| \geq \kappa_x )$ established above, one can prove the $t=1$ case for \eqref{eq:lemmabound2} and \eqref{eq:lemmabound1},  as justified in the work in \eqref{eq:CS_split1} -- \eqref{eq:tobebounded} along with results \eqref{eq:lemmaboundh} and \eqref{eq:lemmabound3}.
Thus, for the second result we would like to show in \eqref{eq:lemmaboundx}, namely the bound on $\mathbb{P}( \frac{1}{n}  \sum_{i=1}^n x^{t}_i \geq \frac{\mu_t^n}{2} )$, we note that we can employ the result \eqref{eq:lemmabound1} with pseudo-Lipschitz function $ \psi(a, b) = b/\mu^n_{1}$ with constant $L_{\psi} = (\mu^n_{1})^{-1}$ to give
\begin{equation*}
\mathbb{P} \Big(\Big| \frac{1}{n} \sum_{i=1}^n \frac{x^{1}_i}{\mu^n_{1}} -  \rho_n \Big| \geq \epsilon \Big)  \leq C C_{1} \exp\Big\{\frac{-c c_{1} n \epsilon^2}{(\mu^n_{1})^{-2}\widetilde{\gamma}_n^{1}}\Big\}  \leq C C_{1} \exp\Big\{\frac{-c c_{1} n \epsilon^2}{\rho_n^{-2}\widetilde{\gamma}_n^{1}}\Big\},
\end{equation*}
where the final inequality follows from the bound $(\mu_t^n)^{-2} \leq \kappa' \rho_n^{-2}$ justified above in \eqref{eq:mu_bound}.  Then the desired result in \eqref{eq:lemmaboundx} follows from the above since, when $\rho_n \leq 1/4$,
 \begin{equation*}
 \begin{split}
    \mathbb{P} \Big( \frac{1}{n}  \sum_{i=1}^n x^{1}_i  \geq \frac{\mu_1^n}{2} \Big) & =   \mathbb{P} \Big( \frac{1}{n}  \sum_{i=1}^n  \frac{x^{1}_i}{\mu_1^n }  \geq \frac{1}{2} \Big) \\
    &\leq  \mathbb{P} \Big( \Big\lvert \frac{1}{n}  \sum_{i=1}^n   \frac{x^{1}_i}{\mu_1^n }  -  \rho_n \big \lvert \geq \frac{1}{2} - \rho_n \Big) \leq  \mathbb{P} \Big( \Big\lvert \frac{1}{n}  \sum_{i=1}^n \frac{x^{1}_i}{\mu_1^n } -  \rho_n \big \lvert \geq \frac{1}{4} \Big).
   \end{split}
  \end{equation*}
%

%
%
%
Now assume that all results \eqref{eq:lemmaboundx}--\eqref{eq:lemmabound1} hold up until iteration $t-1$ and we prove the results for iteration $t$.  As justified in the work in \eqref{eq:CS_split1} -- \eqref{eq:tobebounded}, the results \eqref{eq:lemmabound2} and  \eqref{eq:lemmabound1} follow immediately from \eqref{eq:lemmaboundh} -- \eqref{eq:lemmabound3} so we only aim to prove \eqref{eq:lemmabound3} and  \eqref{eq:lemmaboundx} here.  We begin by proving \eqref{eq:lemmabound3}  which we will then use to prove \eqref{eq:lemmaboundx}.

\paragraph{Result \eqref{eq:lemmabound3}.} Next we consider result \eqref{eq:lemmabound3}. Using the definitions of $\bx^{t+1}$ and $\bh^{t+1}$ from \eqref{eq:AMP3} and \eqref{eq:AMP_correct} along with Cauchy-Schwarz inequality, we have that
 \begin{equation}
 \begin{split}
  \label{eq:bound2}
& \frac{1}{n} \sum_{i=1}^n  \Big| x^{t}_i - h^{t}_i- \mu_{t}^n X_i  \Big|^2\\
 &\leq   \frac{3}{n}\Big\lvert  \frac{\sqrt{\lambda_n}}{n}  \langle \bX, f_{t-1}(\bx^{t-1}) \rangle  - \mu_{t}^n\Big \lvert^2 \sum_{i=1}^n  X_i^2 +   \frac{3}{n} \sum_{i=1}^n  \Big|   [\bZ  f_{t-1}(\bx^{t-1})]_i -   [\bZ g_{t-1}(\bh^{t-1}, \bX)]_i \Big|^2 \\%
 & \quad  +  \frac{3}{n} \sum_{i=1}^n  \Big|  \mathsf{b}_{t-1}  f_{t-2}(x^{t-2}_i) - \mathsf{c}_{t-1} g_{t-2}(h^{t-2}_i, X_i)  \Big|^2 \\
  &=3  \Big \lvert \frac{\sqrt{\lambda_n}}{n}  \langle \bX, f_{t-1}(\bx^{t-1}) \rangle  - \mu_{t}^n \Big\lvert^2  \frac{1}{n}  \sum_{i=1}^n  X_i^2  +   \frac{3}{n}  \Big \|  \bZ \big( f_{t-1}(\bx^{t-1}) -   g_{t-1}(\bh^{t-1}, \bX) \big) \Big\|^2 \\%
 &\quad  +  \frac{3}{n}   \Big\|  \mathsf{b}_{t-1}  f_{t-2}(\bx^{t-2}) - \mathsf{c}_{t-1} g_{t-2}(\bh^{t-2}, \bX)  \Big \|^2.
\end{split}
 \end{equation}
Now we use the upper bounds in  \eqref{eq:bound2} along with lemma~\ref{sums} to give the following upper bound on the probability on the LHS of \eqref{eq:lemmabound3}:
 \begin{equation}
 \begin{split}
\mathbb{P} \Big( \frac{1}{n}   \Big\| \bx^{t} - \bh^{t} - \mu_{t}^n \bX  \Big\|^2 \geq \frac{\kappa \epsilon^2}{ L_{\psi}^2} \Big)  &\leq \mathbb{P} \Big( \Big \lvert \frac{\sqrt{\lambda_n}}{n}  \langle \bX, f_{t-1}(\bx^{t-1}) \rangle  - \mu_{t}^n \Big\lvert^2 \frac{1}{n}  \sum_{i=1}^n   X_i^2  \geq \frac{\kappa \epsilon}{ L_{\psi}^2 } \Big) \\
 &+ \mathbb{P} \Big(\frac{1}{n}  \Big \|  \bZ \big( f_{t-1}(\bx^{t-1}) -   g_{t-1}(\bh^{t-1}, \bX) \big) \Big\|^2  \geq \frac{\kappa \epsilon}{L_{\psi}^2} \Big) \\
 &+ \mathbb{P} \Big( \frac{1}{n}   \Big\|  \mathsf{b}_{t-1}  f_{t-2}(\bx^{t-2}) - \mathsf{c}_{t-1} g_{t-2}(\bh^{t-2}, \bX)  \Big \|^2 \geq \frac{\kappa \epsilon}{L_{\psi}^2} \Big).
 \label{eq:three_terms}
 \end{split}
 \end{equation}
We label the three terms in the above $T_1,T_2,T_3$ and provide an upper bound for each.

 
First consider term $T_1$ of \eqref{eq:three_terms}, and recall that $ \mu^n_{t-1} = \sqrt{\lambda_n}  \tau^n_{t-1}$.  Thus, we have the upper bound
\begin{equation}
  \begin{split}
&T_1 \leq  \mathbb{P} \Big(  \Big|   \frac{1}{n} \langle \bX, f_{t-1}(\bx^{t-1}) \rangle  -  \tau^n_{t-1}  \Big| \geq \frac{\kappa \sqrt{\epsilon}}{ \sqrt{\rho_n \lambda_n} L_{\psi}}  \Big)  +  \mathbb{P} \Big(   \frac{ 1}{n} \sum_{i=1}^nX_i^2  \geq2 \rho_n \Big). 
\label{eq:T1_bound}
 \end{split}
 \end{equation}
 %
Notice that we can upper bound the second term in \eqref{eq:T1_bound} with using $2 e^{-{ n \rho_n}/{2 }}$ Chernoff's bounds (lemma~\ref{lem:chernoff}).
We can upper bound the first term  in \eqref{eq:T1_bound} using the induction hypothesis for result \eqref{eq:lemmabound3} for the pseudo-Lipschitz function $\widetilde{\psi}(a, b) = a f_{t-1}(b)$ with constant $L_{\widetilde{\psi}} = L_f$. Thus,
\begin{equation*}
  \begin{split}
\mathbb{P} \Big(  \Big|   \frac{1}{n} \langle \bX, f_{t-1}(\bx^{t-1}) \rangle  -  \tau^n_{t-1}  \Big| \geq \frac{\kappa \sqrt{\epsilon}}{ \sqrt{\rho_n \lambda_n} L_{\psi}}  \Big)  \leq C C_{t-1} \exp\Big\{\frac{-c c_{t-1} n \epsilon^2}{L_{\psi}^2 L_f^2 \lambda_n \rho_n \widetilde{\gamma}_n^{t-1}}\Big\}.
 \end{split}
 \end{equation*}
Finally we notice that the desired result follows since $ \lambda_n^2 \rho_n \widetilde{\gamma}_n^{t-1} \leq  \widetilde{\gamma}_n^{t}$ using the definition of $ \widetilde{\gamma}_n^{t}$ in \eqref{eq:tilde_gamma}. Indeed, it follows using $L_f = \sqrt{\lambda_n}$, proved in lemma~\ref{lem:PL_cond}, that
\begin{equation*}
  \begin{split}
C C_{t-1} \exp\Big\{\frac{-c c_{t-1} n \epsilon^2}{L_{\psi}^2 L_f^2   \lambda_n \rho_n \widetilde{\gamma}_n^{t-1}}\Big\} \leq C C_{t-1} \exp\Big\{\frac{-c c_{t-1} n \epsilon^2}{L_{\psi}^2   \lambda_n^2 \rho_n \widetilde{\gamma}_n^{t-1}}\Big\} \leq \textsf{bound}_t.
 \end{split}
 \end{equation*}
 %


Now consider term $T_2$ of \eqref{eq:three_terms}. We define an event
\begin{equation}
\label{eq:Fdef}
\mathcal{F}_{t-1} := \left\{\max_i \{x^{t-1}_i\} \leq  \frac{\mu^n_{t-1}}{2} \, \cap \, \max_i\{h^{t-1}_i +  \mu_{t-1}^n X_i\} \leq  \frac{\mu^n_{t-1}}{2}\right\},
\end{equation}
and when considering term $T_2$ of \eqref{eq:three_terms} we define $\Pi$ to be the event of interest so that $T_2 = \mathbb{P}(\Pi)$. Clearly, then
\begin{equation}
\label{eq:F_bound}
T_2 = \mathbb{P}(\Pi) = \mathbb{P}(\Pi \, \cap \, \mathcal{F}_{t-1}) + \mathbb{P}(\Pi \, \cap \, \mathcal{F}_{t-1}^c) \leq  \mathbb{P}( \mathcal{F}_{t-1}) \mathbb{P}(\Pi \, \lvert \, \mathcal{F}_{t-1}) + \mathbb{P}(\mathcal{F}_{t-1}^c).
\end{equation}
So in what follows we bound $\mathbb{P}(\mathcal{F}_{t-1}^c)$, the probability of the complement of the event in $\mathcal{F}_{t-1}$ defined in \eqref{eq:Fdef}, and 
\begin{equation}
\label{eq:conditioned}
 \mathbb{P}( \mathcal{F}_{t-1}) \mathbb{P}(\Pi \, \lvert \, \mathcal{F}_{t-1}) =  \mathbb{P} \Big(\frac{1}{n}  \Big \|  \bZ \big( f_{t-1}(\bx^{t-1}) -   g_{t-1}(\bh^{t-1}, \bX) \big) \Big\|^2  \geq \frac{\kappa \epsilon}{L_{\psi}^2}  \, \big \lvert  \, \mathcal{F}_{t-1}\Big)  \mathbb{P}(\mathcal{F}_{t-1}). 
\end{equation}
The idea is that, conditional on $\mathcal{F}_{t-1}$, the function $f_{t-1}$ has a Lipschitz constant $\sqrt{\lambda_n} \rho_n$ (instead of $\sqrt{\lambda_n}$, its Lipschitz constant over the real line) as proved in lemma~\ref{lem:PL_cond}.


First we bound $\mathbb{P}(\mathcal{F}_{t-1}^c)$. First, notice that
\begin{equation*}
\begin{split}
&\mathbb{P}\Big(\max_i \{x^{t-1}_i\} \leq  \frac{\mu^n_{t-1}}{2} \, \cap \, \max_i\{h^{t-1}_i +  \mu_{t-1}^n X_i\} \leq  \frac{\mu^n_{t-1}}{2}\Big) \\
&\leq \mathbb{P}\Big(\max_i\{x^{t-1}_i\} \leq  \frac{\mu^n_{t-1}}{2}\Big) + \mathbb{P}\Big( \max_i\{h^{t-1}_i +  \mu_{t-1}^n X_i\} \leq  \frac{\mu^n_{t-1}}{2}\Big)\\
&\overset{(a)}{\leq}\mathbb{P}\Big(\frac{1}{n} \sum_{i=1}^n x^{t-1}_i \leq  \frac{\mu^n_{t-1}}{2}\Big) + \mathbb{P}\Big( \frac{1}{n} \sum_{i=1}^n h^{t-1}_i +  \mu_{t-1}^n X_i \leq  \frac{\mu^n_{t-1}}{2}\Big) \overset{(b)}{\leq} \textsf{bound}_t,
\end{split}
\end{equation*}
where the  step $(a)$ follows since if $\max_i(x_i) \leq B$ then $\bar{x} \leq B$ and step $(b)$ follows from results \eqref{eq:lemmaboundx} and \eqref{eq:lemmaboundh} at iteration $t-1$ (i.e.\ the inductive hypothesis for \eqref{eq:lemmaboundx}) and the fact that $\rho_n^{-2} \widetilde{\gamma}_n^{t-1} \leq \lambda_n \widetilde{\gamma}_n^{t-1} \leq \widetilde{\gamma}_n^{t}$ in the regime of interest where $\lambda_n = \kappa \rho_n^{-2}$.


Now we upper bound the probability in \eqref{eq:conditioned}. First notice that, conditioned on event $\mathcal{F}_{t-1}$,
\begin{equation*}
\begin{split}
&\frac{1}{\sqrt{n}}  \Big \|  \bZ \big( f_{t-1}(\bx^{t-1}) -   g_{t-1}(\bh^{t-1}, \bX) \big) \Big\| \leq \frac{1}{\sqrt{n}}   \|  \bZ \|_{op} \Big\| f_{t-1}(\bx^{t-1}) -   g_{t-1}(\bh^{t-1}, \bX) \Big\| \\
%
%
&\overset{(a)}{\leq} \frac{1}{\sqrt{n}}   \|  \bZ \|_{op} \Big\| f_{t-1}(\bx^{t-1}) -  f_{t-1}(\bh^{t-1} + \mu_{t-1}^n \bX) \Big\| \overset{(b)}{\leq} \|  \bZ \|_{op} \frac{\sqrt{\lambda_n} \rho_n}{\sqrt{n}}  \Big\| \bx^{t-1} -   \bh^{t-1} - \mu_{t-1}^n \bX \Big\|,
\end{split}
\end{equation*}
where step $(a)$ uses that $g_{t-1}(\bh^{t-1}, \bX) = f_{t-1}(\bh^{t-1} + \mu_{t-1}^n \bX)$ and step $(b)$ uses the Lipschitz property of $f_{t-1}$, conditioned on event $\mathcal{F}_{t-1}$.  
Therefore, 
\begin{equation*}
  \begin{split}
& \mathbb{P} \Big(\frac{1}{\sqrt{n}}  \Big \|  \bZ \big( f_{t-1}(\bx^{t-1}) -   g_{t-1}(\bh^{t-1}, \bX) \big) \Big\|  \geq \frac{\kappa \sqrt{\epsilon}}{L_{\psi}}  \, \big \lvert  \, \mathcal{F}_{t-1}\Big) \mathbb{P}(\mathcal{F}_{t-1})\\
& \leq   \mathbb{P} \Big( \|  \bZ \|_{op}   \frac{1}{\sqrt{n}} \sqrt{\lambda_n} \rho_n \Big\| \bx^{t-1} -   \bh^{t-1} - \mu_{t-1}^n \bX \Big\| \geq \frac{\kappa \epsilon}{L_{\psi} }  \, \big \lvert  \, \mathcal{F}_{t-1}\Big) \mathbb{P}(\mathcal{F}_{t-1}) \\
& \leq   \mathbb{P} \Big( \|  \bZ \|_{op}   \frac{1}{\sqrt{n}} \sqrt{\lambda_n} \rho_n \Big\| \bx^{t-1} -   \bh^{t-1} - \mu_{t-1}^n \bX \Big\| \geq \frac{\kappa \epsilon}{L_{\psi} } \Big) \\
&\leq  \mathbb{P} \Big(   \frac{1}{\sqrt{n}} \Big\| \bx^{t-1} -   \bh^{t-1} - \mu_{t-1}^n \bX \Big\| \geq \frac{\kappa \ep}{L_{\psi} \sqrt{\lambda_n} \rho_n} \Big) + \mathbb{P} \Big(   \|  \bZ \|_{op}   \geq \kappa  \Big) \\
&\leq C C_{t-1} \exp\Big\{\frac{-c c_{t-1} n \epsilon^2}{L_{\psi}^2  \widetilde{\gamma}_n^{t}}\Big\} +C \exp\{- c n\},
 \end{split}
 \end{equation*}
 where the final inequality follows from the inductive hypothesis for \eqref{eq:lemmabound3} and standard results about tail bounds for operator norms of GOE matrices.  In particular, we have used the  inductive hypothesis to find
\begin{equation*}
  \begin{split}
 \mathbb{P} \Big(   \frac{1}{n} \Big\| \bx^{t-1} -   \bh^{t-1} + \mu_{t-1}^n \bX \Big\|^2 \geq \frac{\kappa \ep^2}{L^2_{\psi} \lambda_n \rho_n^2} \Big) &\leq  C C_{t-1} \exp\Big\{\frac{-c c_{t-1} n \epsilon^2}{L_{\psi}^2  \lambda_n \rho_n^2 \widetilde{\gamma}_n^{t-1}}\Big\} \\
 &\leq C C_{t-1} \exp\Big\{\frac{-c c_{t-1} n \epsilon^2}{L_{\psi}^2     \widetilde{\gamma}_n^{t}}\Big\},
 \end{split}
 \end{equation*}
where the final inequality follows since $\lambda_n \rho_n^2 \widetilde{\gamma}_n^{t-1} \leq   \widetilde{\gamma}_n^{t}$.

 
 Finally, consider term $T_3$ of \eqref{eq:three_terms}.  To bound this term, we use a strategy as we did for term $T_2$ in \eqref{eq:Fdef}-\eqref{eq:F_bound}: conditioning on an event that makes sure the input to the denoiser is small enough that the Lipschitz constant can be assumed to be $\sqrt{\lambda_n} \rho_n$ instead of $\sqrt{\lambda_n}$.  However, we do not go through this argument in detail since it is analogous to that for term $T_2$.
 
 We first give an upper bound using  the definition of $g_t$ and the Lipschitz property of $f_t$ with $L_f = \sqrt{\lambda_n} \rho_n$ as follows:
\begin{equation}
\begin{split}
\label{eq:termT3_1}
 & \Big\|  \mathsf{b}_{t-1}  f_{t-2}(\bx^{t-2}) - \mathsf{c}_{t-1} g_{t-2}(\bh^{t-2}, \bX)  \Big \| =  \Big\|  \mathsf{b}_{t-1}  f_{t-2}(\bx^{t-2}) - \mathsf{c}_{t-1}  f_{t-2}(\bh^{t-2} +  \mu_{t-2}^n \bX)  \Big \| \\
 & \leq  |\mathsf{b}_{t-1}| \Big\|   f_{t-2}(\bx^{t-2}) -  f_{t-2}(\bh^{t-2} +  \mu_{t-2}^n \bX) \Big \| +  | \mathsf{b}_{t-1}   - \mathsf{c}_{t-1} |  \Big \|  f_{t-2}(\bh^{t-2} +  \mu_{t-2}^n \bX) \Big \| \\
  %
   %
 & \leq   \lambda_n \rho_n^2 \Big\|  \bx^{t-2} -  \bh^{t-2} - \mu_{t-2}^n \bX \Big \| +  | \mathsf{b}_{t-1}   - \mathsf{c}_{t-1} | \sqrt{n}.
  \end{split}
  \end{equation}
In the final step we use the lemma~\ref{lem:PL_cond} results 
 \begin{equation*}
 | \mathsf{b}_{t-1}   | \leq  \frac{1}{n} \sum_{i=1}^n \big| f'_{t-1}(x^{t-1}_i)  \big | \leq  \sqrt{\lambda_n} \rho_n, \qquad \text{ and } \qquad \Big \|  f_{t-2}(\bh^{t-2} + \mu_{t-2}^n \bX)  \Big \|^2 \leq n.
 \end{equation*}
 %

 We investigate the term $| \mathsf{b}_{t-1}   - \mathsf{c}_{t-1} |$ and recall from their definitions in \eqref{eq:AMP} and \eqref{eq:AMP_correct},
 \begin{equation}
 \begin{split}
& | \mathsf{b}_{t-1}   - \mathsf{c}_{t-1} | \leq  \frac{1}{n} \sum_{i=1}^n \big| f'_{t-1}(x^{t-1}_i) -  g'_{t-1}(h^{t-1}_i, X_i) \big |=\frac{1}{n} \sum_{i=1}^n \big|  f'_{t-1}(x^{t-1}_i) -    f'_{t-1}(h^{t-1}_i - \mu_{t-1}^n X_i ) \big | \\
&\overset{(a)}{=}   \frac{\sqrt{\lambda_n}}{n} \sum_{i=1}^n   \Big|f_{t-1}(x^{t-1}_i)(1- f_{t-1}(x^{t-1}_i)) -  f_{t-1}(h^{t-1}_i - \mu_{t-1}^n X_i)(1- f_{t-1}(h^{t-1}_i - \mu_{t-1}^n X_i) )\Big | \\
&\overset{(b)}{\leq} \frac{\sqrt{\lambda_n}}{n} \sum_{i=1}^n \Big|   f_{t-1}(x^{t-1}_i) -  f_{t-1}(h^{t-1}_i - \mu_{t-1}^n X_i)  \Big | \overset{(c)}{\leq} \frac{\lambda_n  \rho_n}{\sqrt{n}}  \Big\|   \bx^{t-1} - \bh^{t-1} - \mu_{t-1}^n \bX  \Big \|.
\label{eq:sfterms}
 \end{split}
 \end{equation}
In the above, step $(a)$ uses lemma~\ref{lem:PL_cond} for computing the derivative $f_t$, step $(b)$ uses the bound
  \begin{equation*}
 \begin{split}
 \big| f(a)(1- f(a)) -  f(b)(1- f(b)) \big | &\leq  \big| f(a) -  f(b)\big |   + \big |  [f(a)]^2 -  [f(b)]^2 \big | \leq  \kappa \big| f(a) -  f(b)\big |,
 \end{split} 
 \end{equation*}
 for $0 \leq f(a) \leq 1$ for all $a \in \mathbb{R}$, and step $(c)$ uses the Lipschitz property of $f_{t}$, namely $L_f = \sqrt{\lambda_n} \rho_n$, and Cauchy Schwarz to give $\sum_{i=1}^n |a_i| \leq \sqrt{n} ||\textbf{a}||$.   
Plugging the bound in \eqref{eq:sfterms} into \eqref{eq:termT3_1},
 \begin{equation*}
\begin{split}
 & \Big\|  \mathsf{b}_{t-1}  f_{t-2}(\bx^{t-2}) - \mathsf{c}_{t-1} g_{t-2}(\bh^{t-2}, \bX)  \Big \| \\
  & \leq   \lambda_n \rho_n^2 \Big\|  \bx^{t-2} -  \bh^{t-2} - \mu_{t-2}^n \bX \Big \| + \lambda_n \rho_n   \Big\|   \bx^{t-1} - \bh^{t-1} - \mu_{t-1}^n \bX  \Big \|.
  \end{split}
  \end{equation*}
 Now we have from lemma~\ref{sums} that
 \begin{equation*}
\begin{split}
T_3 &\leq 2 \mathbb{P} \Big(\frac{1}{n}  \Big\|   \bx^{t-1} - \bh^{t-1} - \mu_{t-1}^n \bX  \Big \|^2 \geq \frac{\kappa \epsilon}{L^2_{\psi}  \lambda^2_n \rho_n^2} \Big), 
  \end{split}
  \end{equation*}
and the final bound  follows from the inductive hypothesis for \eqref{eq:lemmabound2} using that $\lambda^2_n \rho_n^2 \widetilde{\gamma}_n^{t-1} \leq \lambda^2_n \rho_n \widetilde{\gamma}_n^{t-1} \leq   \widetilde{\gamma}_n^{t}$.


\paragraph{Result \eqref{eq:lemmaboundx}.} To complete the proof, we consider result \eqref{eq:lemmaboundx}. First notice that by the Triangle Inequality, $ \| \bx^{t}  \| \leq \| \bx^{t} - \bh^{t} - \mu_{t}^n \bX \| + \|  \bh^{t} + \mu_{t}^n \bX  \|$.  Then let $ \kappa_x = 2 \kappa_h + \frac{2\kappa}{ L_{\psi}}$ and therefore, by lemma~\ref{sums},
\begin{equation*}
\begin{split}
    \mathbb{P} \Big( \frac{1}{\sqrt{n}}   \| \bx^{t}  \| \geq \kappa_x \Big)  &=  \mathbb{P} \Big( \frac{1}{\sqrt{n}} \| \bx^{t} - \bh^{t} - \mu_{t}^n \bX  \| +   \frac{1}{\sqrt{n}} \|  \bh^{t} + \mu_{t}^n \bX  \|  \geq2 \kappa_h + \frac{2\kappa}{ L_{\psi}} \Big)  \\
&\leq  \mathbb{P} \Big( \frac{1}{\sqrt{n}}  \| \bx^{t} - \bh^{t} - \mu_{t}^n \bX  \|  \geq  \frac{\kappa}{ L_{\psi}} \Big)  +  \mathbb{P} \Big(   \frac{1}{\sqrt{n}} \|  \bh^{t} + \mu_{t}^n \bX \|  \geq  \kappa_h  \Big).
\end{split}
\end{equation*}
Then the bound follows by \eqref{eq:lemmaboundx} and \eqref{eq:lemmabound3}.  Note that with the bound on $   \mathbb{P} ( \frac{1}{\sqrt{n}}   \| \bx^{t}  \| \geq \kappa_x )$ established above, one can prove \eqref{eq:lemmabound1} and \eqref{eq:lemmabound2}, thus for the second result we would like to show in \eqref{eq:lemmaboundx}, namely the bound on $\mathbb{P}( \frac{1}{n}  \sum_{i=1}^n x^{t}_i \geq \frac{\mu_t^n}{2} )$, we note that we can employ the result \eqref{eq:lemmabound1}.  The proof proceeds then as in the same case for \eqref{eq:lemmaboundh}.
\end{proof}

\subsection{Useful lemmas} \label{appsec:useful_lemmas}

In this section we introduce a number of technical lemmas that are used to prove our main results. We include proofs only where the proof is non-standard.

 \begin{lemma}
 Recall from \cite{BLMConc},  that a random variable, $X$, is sub-gaussian with variance factor $\nu$ if $\log \mathbb{E}[e^{t (X - \mathbb{E}[X])}] \leq {t^2 \nu}/{2}$ for all $t \in \mathbb{R}$. When $X \sim p_{X}$, we have $\nu = 12 \rho$ for $p_{X} \sim {\rm Ber}(\rho)$.
 \label{lem:subgauss}
 \end{lemma}

 \begin{proof}
 The proof relies on an intermediate result: if for any $t >0$ it is true that
 \begin{equation}
 P(X > t) \leq \exp\left(\frac{-t^2}{2 \sigma^2}\right), \qquad \text{ and } \qquad  P(X < -t) \leq \exp\left(\frac{-t^2}{2 \sigma^2}\right),
 \label{eq:prob_bounds}
 \end{equation}
 then for any $s > 0$ it holds that $\mathbb{E}[e^{s X}] \leq e^{4 \sigma^2 s^2}$. This is easy to prove by showing that \eqref{eq:prob_bounds} implies bounds on the moments $\mathbb{E}[|X|^k]$ for any $k \in \mathbb{N}$, from which we can bound the moment generating function. 
 Thus, by  lemma~\ref{lem:chernoff}, we have that $\sigma^2 = (3/2) \rho_n$, giving the desired result: $\mathbb{E}[e^{t (X - \mathbb{E}[X])}] \leq e^{6 \rho_n t^2}$.
 \end{proof}


\begin{lemma}[Chernoff's bounds for Bernoulli's]
\label{lem:chernoff}
If $X_1, \ldots, X_n$ be i.i.d.\ ${\rm Ber}(\rho_n)$, then for $\ep \in [0,1]$ we have 
$\mathbb{P}(  \lvert \frac{1}{n}\sum_{i=1}^n (X_i -\rho_n)  \lvert \geq \ep ) \leq 2\exp\left\{  \frac{-n \ep^2}{3 \rho_n}\right\}.$
\end{lemma}

\begin{lemma}[Hoeffding's inequality]
\label{lem:hoeff_lem}
If $X_1, \ldots, X_n$ are independent bounded random variables such that $a_i \leq X_i \leq b_i$, then for $\nu = 2[\sum_{i} (b_i -a_i)^2]^{-1}$, we have 
$\mathbb{P}(  \lvert \frac{1}{n}\sum_{i=1}^n (X_i -\mathbb{E}\{X_i\})  \lvert \geq \ep ) \leq 2e^{ -\nu n^2 \ep^2}.$
\end{lemma}

 \begin{lemma}  \label{lem:PL}
 Recall the definition of pseudo-Lipschitz functions of order $2$ given in Definition~\ref{def:PL}. 
 The following functions $\psi: \mathbb{R}^2 \rightarrow \mathbb{R}$ are all pseudo-Lipschitz of order $2$ with  pseudo-Lipschitz constant $2$. 
 \begin{equation}
 \begin{split}
 \psi_1(a, b)  &= (a - b)^2, \qquad  \qquad \psi_2(a, b)  = b^2, \qquad  \qquad \psi_3(a, b)  = a b.
 \label{eq:the_phis}
 \end{split}
 \end{equation}
  \end{lemma}

 \begin{proof}
 Verifying the pseudo-Lipschitz property for the functions in \eqref{eq:the_phis} is straightforward, so we omit the details.
\end{proof}

 \begin{lemma}  \label{lem:PL_second}
 Recall the definition of pseudo-Lipschitz functions of order $2$ given in Definition~\ref{def:PL}. Let $f_t$ be the conditional expectation denoiser in \eqref{eq:denoiser} with Lipschitz constant $L_f^n$ and let  $\psi: \mathbb{R}^2 \rightarrow \mathbb{R}$ be a pseudo-Lipschitz of order $2$ function with constant $L_{\psi}$.
 The following functions $\phi: \mathbb{R}^2 \rightarrow \mathbb{R}$ are all pseudo-Lipschitz of order $2$ with  the stated pseudo-Lipschitz constants. 
 \begin{align}
 \phi_1(a, b)  &=   \psi(a+ \mu_{t}^n b, b) , \qquad  L_{\phi_1} =  2 L_{\psi} (1+\mu_{t}^n)^2,  \label{eq:phi1}\\
\phi_2(a, b)  &=  \psi(a, f_{t}(b)),  \qquad  L_{\phi_2} =  3  L_{\psi}  \max\{1, L_f\},  \label{eq:phi2} \\
\phi_3(a, b)  &=   (\mu_{t}^n)^{-1}a+b  , \qquad  L_{\phi_3} =  \sqrt{2} \max\{1, (\mu_{t}^n)^{-1}\} ,  \label{eq:phi3} \\
 \phi_4(a, b)  &=   (a + \mu_{t}^n b)^2  , \qquad  L_{\phi_4} =  2\max\{1, (\mu_{t}^n)^2\} .  \label{eq:phi4} 
 \end{align}
  \end{lemma}

  \begin{proof}
For function $\phi_1$ in \eqref{eq:phi1}, first notice
 \begin{equation}
 \begin{split}
 \label{eq:Lipschitz1}
&\lvert  \phi_1(a,b) -   \phi_1(\widetilde{a},\widetilde{b}) \lvert = \lvert   \psi(a+ \mu_{t}^n b, b) -     \psi(\widetilde{a}+ \mu_{t}^n \widetilde{b},\widetilde{b})  \lvert  \\
&\qquad \leq L_{\psi} (1 +  \| (a+ \mu_{t}^n b, b) \| +  \|  (\widetilde{a}+ \mu_{t}^n \widetilde{b},\widetilde{b})  \| )  \times \| (a+ \mu_{t}^n b, b) -  (\widetilde{a}+ \mu_{t}^n \widetilde{b},\widetilde{b})   \|.
\end{split}
 \end{equation}
 Next notice
 \begin{equation*}
 \begin{split}
&\| (a+ \mu_{t}^n b, b) - (\widetilde{a}+ \mu_{t}^n \widetilde{b},\widetilde{b})  \| \leq  | a- \widetilde{a}| + (1+\mu_{t}^n)  |b- \widetilde{b}| \leq \sqrt{2} (1+\mu_{t}^n)   \| (a, b) -    (\widetilde{a},\widetilde{b})  \|,
\end{split}
 \end{equation*}
and
$ \| (a+ \mu_{t}^n b, b) \|  \leq  |a+ \mu_{t}^n b| + |b | \leq  |a|+ (1+\mu_{t}^n) |b| \leq \sqrt{2} (1+\mu_{t}^n) \| (a, b) \|. $
Thus, from  \eqref{eq:Lipschitz1}, we have result \eqref{eq:phi1}:
 \begin{equation*}
 \begin{split}
&\lvert  \phi_1(a,b) -   \phi(\widetilde{a},\widetilde{b}) \lvert \leq 2 L_{\psi} (1+\mu_{t}^n)^2 (1 +  \| (a, b) \| + \| (\widetilde{a},\widetilde{b}) \|| )   \times  \| (a, b) -    (\widetilde{a},\widetilde{b})  \|.
\end{split}
 \end{equation*}

 For function $\phi_2$ in \eqref{eq:phi2}, first notice  
 \begin{equation}
 \begin{split}
 \label{eq:Lipschitz_func2}
\lvert    \phi_2(a, b) -   \phi_2(\widetilde{a}, \widetilde{b}) \lvert &= \lvert   \psi(a, f_{t}(b)) -     \psi(\widetilde{a}, f_{t}(\widetilde{b}))  \lvert \\
& \leq L_{\psi} (1 +  \| (a, f_{t}(b))  \| +  \|  (\widetilde{a}, f_{t}(\widetilde{b}))  \|) \| (a, f_{t}(b)) -    (\widetilde{a}, f_{t}(\widetilde{b}))  \|.
\end{split}
 \end{equation}
 Next, notice that since $f_t(\cdot)$ is a Lipschitz function with constant $L_f$,
 \begin{equation*}
 \begin{split}
\| (a, f_{t}(b)) -    (\widetilde{a}, f_{t}(\widetilde{b}))  \|^2 &= | f_{t}(b) - f_{t}(\widetilde{b})|^2+ | a  - \widetilde{a}|^2 \\
&\leq   L_f^2 | b- \widetilde{b}|^2 + | a  - \widetilde{a}|^2 \leq \max\{1, L_f^2\} \| (a, b) -    (\widetilde{a}, \widetilde{b})  \|^2,
\end{split}
 \end{equation*}
 and since our denoiser of interest $f_t$ in \eqref{eq:denoiser} is such that $|f_t(x)| \leq 1$,
   %
 \begin{equation*}
 \begin{split}
   \| (a, f_{t}(b))  \| \leq  |a| + |f_{t}(b)| \leq  |a|  + 1 &\leq (1 + |a| +  |b| )  \leq \sqrt{2} (1 + \|(a, b)\|).
\end{split}
 \end{equation*}
 Thus, from \eqref{eq:Lipschitz_func2},
 \begin{equation}
 \begin{split}
 \label{eq:Lipschitz_func2}
\lvert  \widetilde{\psi}(a,b) -   \widetilde{\psi}(\widetilde{a}, \widetilde{b}) \lvert   \leq 3  L_{\psi}  \max\{1, L_f\} (1 + \|(a, b)\| + \|( \widetilde{a},  \widetilde{b})\| )  \| (a, b) -    (\widetilde{a}, \widetilde{b})  \|.
\end{split}
 \end{equation}

Next, the bound for function $\phi_3$ in \eqref{eq:phi3} is straightforward:
 \begin{equation*}
 \begin{split}
\lvert    \phi_3(a, b) -   \phi_3(\widetilde{a}, \widetilde{b}) \lvert &= \Big \lvert    (\mu_{t}^n)^{-1} a +b  -   (\mu_{t}^n)^{-1}  \widetilde{a} - \widetilde{b} \Big  \lvert   \\
&\leq      (\mu_{t}^n)^{-1} \lvert a   -  \widetilde{a}   \lvert + \lvert b - \widetilde{b}  \lvert \leq  \sqrt{2} \max\{1, (\mu_{t}^n)^{-1}\} \| (a,b)   -  (\widetilde{a}, \widetilde{b}) \| .
\end{split}
 \end{equation*}

Finally, for function $\phi_4$ in \eqref{eq:phi4}, first notice  
 \begin{equation*}
 \begin{split}
&\lvert    \phi_4(a, b) -   \phi_4(\widetilde{a}, \widetilde{b}) \lvert = \Big \lvert    (a + \mu_{t}^n b)^2  -   (  \widetilde{a} + \mu_{t}^n   \widetilde{b})^2 \Big  \lvert   \\
&\leq \Big \lvert    (a + \mu_{t}^n b)  -   (  \widetilde{a} + \mu_{t}^n   \widetilde{b}) \Big  \lvert  \Big \lvert    (a + \mu_{t}^n b) +  (  \widetilde{a} + \mu_{t}^n   \widetilde{b}) \Big  \lvert  \\
&\leq   2\max\{1, (\mu_{t}^n)^2\} ( \| (a, b)\| + \| (\widetilde{a} , \widetilde{b} )\|)  ( \| (a, b) - (\widetilde{a} , \widetilde{b} )\|), 
\end{split}
 \end{equation*}
 where the final inequality uses that $ \lvert    a + \mu_{t}^n b  \lvert  \leq \sqrt{2} \max\{1, \mu_{t}^n\}   \| (a, b)\|$ giving
 $$ \Big \lvert    (a + \mu_{t}^n b) +  (  \widetilde{a} + \mu_{t}^n   \widetilde{b}) \Big  \lvert  \leq  \sqrt{2} \max\{1, \mu_{t}^n\}  ( \| (a, b)\| + \| (\widetilde{a} , \widetilde{b} )\|)$$
 and the fact that 
 \begin{align*}
 \Big \lvert    (a + \mu_{t}^n b)  -   (  \widetilde{a} + \mu_{t}^n   \widetilde{b}) \Big  \lvert \leq  \lvert a - \widetilde{a} \lvert + \mu_{t}^n \vert b -   \widetilde{b} \lvert  \leq \sqrt{2} \max\{1,  \mu_{t}^n\} \|(a,b) - (  \widetilde{a} ,   \widetilde{b})\|.
 \end{align*}

 \end{proof}

\begin{lemma}  \label{lem:PL_cond}
 Recall the definition of pseudo-Lipschitz functions of order $2$ given in Definition~\ref{def:PL}.  The conditional expectation denoiser in \eqref{eq:denoiser} is Lipschitz with constant $L_f = \sqrt{\lambda_n}$ when $X_0^n \sim P_{X,n}$ and $P_{X,n}$ is either ${\rm Ber}(\rho_n)$ or Bernoulli-Rademacher and $\frac{\partial}{\partial x} f_t(x)  = \sqrt{\lambda_n} f_t(x)(1- f_t(x))$.  Moreover, the Lipschitz constant can be strengthened to $\sqrt{\lambda_n}  \rho_n$ on $x \in (-\infty, \frac{\mu^n_t}{2})$ and $f_t(0) \leq \rho_n$.
  \end{lemma}
  
 \begin{proof}
 First, recall that $f_t(\cdot)$ is the conditional expectation denoiser given in \eqref{eq:denoiser},
\begin{equation*}
 f_t(x) = \mathbb{E}\big \{X_0^n  \mid  \sqrt{\lambda_n} \tau^n_{t}  X_0^n + \sqrt{\tau^n_t} Z = x\big\}.
\end{equation*}
Notice that for either the Bernoulli or Bernoulli-Rademacher case, we have that $|f_t(x)| \leq 1$ for all $x \in \mathbb{R}$ since $X_0^n  \in \{- 1, 0, 1\}$.
 
 First consider $ P_{X,n} \sim {\rm Ber}(\rho_n)$ and we show that $f_t(\cdot)$ is Lipschitz continuous with Lipschitz constant $\sqrt{\lambda_n}$. 
Let $\phi(x)$ denote the standard gaussian density evaluated at $x$. First, by Bayes' Rule,
\begin{equation}
\begin{split}
\label{eq:ft}
 f_t(x) &=\mathbb{E}\big \{X_0^n  \mid \sqrt{\lambda_n} \tau^n_{t}  X_0^n + \sqrt{\tau^n_t} Z = x\big\}  \\
&=\mathbb{P}\big (X_0^n =1 \mid \sqrt{\lambda_n} \tau^n_{t} X_0^n + \sqrt{\tau^n_t} Z = x\big) = \frac{\rho_n \phi \big(    \frac{x-\sqrt{\lambda_n} \tau^n_{t} }{\sqrt{\tau^n_t}}\big)   }{(1-\rho_n) \phi \big(   \frac{x}{\sqrt{\tau^n_t}}\big)  + \rho_n \phi \big(  \frac{x-\sqrt{\lambda_n} \tau^n_{t} }{\sqrt{\tau^n_t}}\big)  }.
%
%
\end{split}
\end{equation}
%
Now notice that $\frac{\partial}{\partial x} \phi (\frac{x-a }{b}) = - \frac{(x-a) }{b^2}\phi (\frac{x-a }{b}) $.  Using this and the representation above,
\begin{equation}
\begin{split}
\label{eq:deriv}
\frac{\partial}{\partial x} f_t(x) &= \frac{\partial}{\partial x}  \left[\frac{\rho_n \phi \big(    \frac{x-\sqrt{\lambda_n} \tau^n_{t} }{\sqrt{\tau^n_t}}\big)   }{(1-\rho_n) \phi \big(   \frac{x}{\sqrt{\tau^n_t}}\big)  + \rho_n \phi \big(  \frac{x-\sqrt{\lambda_n} \tau^n_{t} }{\sqrt{\tau^n_t}}\big)  } \right] \\
&=   \frac{-f_t(x)}{\tau^n_{t}}  \left[ (x-\sqrt{\lambda_n} \tau^n_{t})  - \frac{x (1-\rho_n) \phi \big(   \frac{x}{\sqrt{\tau^n_t}}\big)  + \rho_n  (x-\sqrt{\lambda_n} \tau^n_{t})  \phi \big(  \frac{x-\sqrt{\lambda_n} \tau^n_{t} }{\sqrt{\tau^n_t}}\big)   }{(1-\rho_n) \phi \big(   \frac{x}{\sqrt{\tau^n_t}}\big)  + \rho_n \phi \big(  \frac{x-\sqrt{\lambda_n} \tau^n_{t} }{\sqrt{\tau^n_t}}\big)  } \right] \\
&=   \frac{-f_t(x)(x-\sqrt{\lambda_n} \tau^n_{t}) }{\tau^n_{t}}  \left[  \frac {(1-\rho_n) \phi \big(   \frac{x}{\sqrt{\tau^n_t}}\big) \Big[1  -\frac{x}{(x-\sqrt{\lambda_n} \tau^n_{t}) } \Big]   }{(1-\rho_n) \phi \big(   \frac{x}{\sqrt{\tau^n_t}}\big)  + \rho_n \phi \big(  \frac{x-\sqrt{\lambda_n} \tau^n_{t} }{\sqrt{\tau^n_t}}\big)  } \right] \\
&=  \sqrt{\lambda_n} f_t(x)(1- f_t(x)).
\end{split}
\end{equation}
Therefore, using \eqref{eq:deriv}, we see that $\Big \lvert \frac{\partial}{\partial x} f_t(x) \Big \lvert \leq    \sqrt{\lambda_n} $ and it follows that $f_t(\cdot)$ is Lipschitz continuous with Lipschitz constant $  \sqrt{\lambda_n} $.

The fact  that $f_t(\cdot)$ is Lipschitz continuous with Lipschitz constant $\sqrt{\lambda_n}$ can be shown similarly for the case where $P_{X,n}$ is Bernoulli-Rademacher. 

Finally, notice that from \eqref{eq:ft} we have
\begin{equation}
\begin{split}
\label{eq:ft2}
 f_t(x) &= \frac{\rho_n   }{(1-\rho_n)  \exp\big\{ \frac{1}{2}  (\lambda_n \tau^n_{t}  - 2x \sqrt{\lambda_n} ) \big\}   +\rho_n  }.
\end{split}
\end{equation}
Then since $e^{x} \geq 1 + x$ (which can be seen by showing that $f(x) =e^{x} - (1 + x)$ has a minimum at $f(0) =0$), %
\begin{equation*}
\begin{split}
 f_t(x) &\leq \frac{\rho_n   }{(1-\rho_n)  (1 +  \frac{1}{2}  (\lambda_n \tau^n_{t}  - 2x \sqrt{\lambda_n} ))   +\rho_n  } =  \frac{\rho_n   }{1  +\frac{1}{2} (1-\rho_n)   (\lambda_n \tau^n_{t}  - 2x \sqrt{\lambda_n} )   }.
\end{split}
\end{equation*}
The above implies that $ f_t(0) \leq \rho_n$, and further, since
\begin{align*}
 (1-\rho_n)   (\lambda_n \tau^n_{t}  - 2x \sqrt{\lambda_n} ) \geq 0 \quad \text{ when } \quad x \leq  \frac{\sqrt{\lambda_n} \tau^n_{t}}{2} ,
\end{align*}
we find the bound $0 \leq  f_t(x) \leq \rho_n$ when $x \leq  \frac{\sqrt{\lambda_n} \tau^n_{t}}{2}$.

Therefore, by \eqref{eq:deriv}, we have $|\frac{\partial}{\partial x} f_t(x)| \leq  \sqrt{\lambda_n} f_t(x) \leq  \sqrt{\lambda_n}  \rho_n$ and it follows that $f_t(\cdot)$ is Lipschitz continuous with Lipschitz constant $  \sqrt{\lambda_n}  \rho_n$ on $x \in (-\infty, \frac{\mu^n_t}{2})$.
\end{proof}

The proof of the following two lemmas can be found in \cite[appendix A]{RushVenkataramanan}.
\begin{lemma}[Concentration of Sums]
\label{sums}
If random variables $X_1, \ldots, X_M$ satisfy $P(\abs{X_i} \geq \ep) \leq e^{-n\kappa_i \ep^2}$ for $1 \leq i \leq M$, then 
 \begin{equation*}
\mathbb{P}\Big(  \lvert \sum_{i=1}^M X_i  \lvert \geq \ep\Big) \leq \sum_{i=1}^M \mathbb{P}\left(|X_i| \geq \frac{\ep}{M}\right) \leq M e^{-n (\min_i \kappa_i) \ep^2/M^2}.
  \end{equation*}
\end{lemma}
\begin{lemma}[Concentration of Powers]
\label{powers}
Assume $c > 0$ and $0 < \ep \leq 1$.  Then, if\,
$
\mathbb{P}( \lvert X_n - c \lvert \geq \epsilon ) \leq e^{-\kappa n \epsilon^2},
$ 
it follows that
$
\mathbb{P}(\lvert X_n^2 - c^2  \lvert \geq \epsilon) \leq e^{ {-\kappa n \ep^2}/[1+2c]^2}.
$
\end{lemma}
%

\section{Algorithmic AMP phase transition regime}\label{app:AMPregime}

In this appendix we show that the right-hand side of the bound in theorem~\ref{AMP-theorem} for signal strength and sparsity scaling like $\lambda_n \rho_n^2= w$ 
and $\rho_n = \Omega((\ln n)^{-\alpha})$ with $w, \alpha\in \mathbb{R}_{+}$, tends to zero as $n\to +\infty$. We focus on the Bernoulli prior case but the arguments generalizes to Bernoulli-Rademacher prior.

Let us first upper bound $\gamma_n^{t}$ in terms of $\lambda_n$ and $\rho_n$ in the Bernoulli case. First we use the bound $\vert f_t^\prime(x)\vert \leq \sqrt{\lambda_n}$ (see lemma~\ref{lem:PL_cond}) to bound 
$$\max\{1, \hat{\textsf{b}}_1\}  \max\{1, \hat{\textsf{b}}_2\} \cdots  \max\{1, \hat{\textsf{b}}_{t-1}\} \leq \lambda_n^{\frac{t-1}{2}}.$$
From the explicit AMP iteration (see appendix G second formula for example) we have 
$\tau^n_{t} \leq \rho_n$. Since $\nu_n = 12\rho_n$ we get
$(\nu^n +  \tau^n_{1})  (\nu^n +  \tau^n_{1} +  \tau^n_{2})  \cdots  (\nu^n + \sum_{i=1}^{t} \tau^n_{i}) \leq (12\rho_n + \rho_n) (12\rho_n + 2\rho_n) \cdots (12\rho_n + t \rho_n) \leq \frac{1}{6} (12+t)! \rho_n^{t}$. 
Putting everything  together we get:
$$
\gamma_n^t \leq \frac{1}{6} (12+t)! \rho_n^{t} \lambda_n^{2t  -1+ \frac{t-1}{2}} . 
$$
Now we use the scaling (which is the correct scale for the phase transition to happen) $\lambda_n = w \rho_n^{-2}$ and get:
$$
\gamma_n^t \leq \frac{1}{6} (12+t)! \frac{w^{\frac{5t -3}{2}}}{\rho_n^{4t-3}} .
$$
Therefore 
$$
\textsf{bound}_t \leq  C C_{t} \exp\Big\{ - \frac{6c}{w^\frac{5t -3}{2}} \frac{c_{t} }{(12 + t)!} \rho_n^{4t-3} n \epsilon^2 \Big\}.
$$
Now the $t$ dependence in the constant $c_t=[C^t (t!)^C]^{-1}$ (from now on $C$ is a generic positive constant) and using Stirling's approximation $t! \approx \sqrt{2\pi t} \, t^t \, e^{-t}$ this scales at dominant order as $[C^t (t^t)^C]^{-1}$. So we have at dominant order
$$
\textsf{bound}_t \approx  C C_{t} \exp\{ - C e^{\pm Ct - C t \ln t} e^{(4t-3)\ln(\rho_n)} e^{\ln n} \epsilon^2 \}.
$$
Now set the number of iterations to $t = o(\frac{\ln n}{\ln\ln n})$. We get $t\ln t = o(\ln n)$ so $\pm Ct - C t \ln t = o(\ln n)$ and 
\begin{align}
\textsf{bound}_t \approx  C C_{t} \exp\{ - C e^{- o(\ln n)} e^{ o(\frac{\ln n}{\ln\ln n}) \ln(\rho_n)} e^{\ln n} \epsilon^2 \}.\label{worsen}
\end{align}
We set $\rho_n = \Theta\big( \frac{1}{(\ln n)^\alpha}\big) = \frac{C}{(\ln n)^\alpha}$. Then 
$\ln\rho_n = \ln C - \alpha\ln\ln n$ and we get 
$$
\textsf{bound}_t \approx  C C_{t} \exp\{ - C e^{- o(\ln n)} e^{ o(\frac{\ln n}{\ln\ln n}) (C - \alpha\ln\ln n)} e^{\ln n} \epsilon^2 \}.
$$
This leads to
\begin{align*}
\textsf{bound}_t & \approx  C C_{t} \exp\{ - C e^{- o(\ln n) + C o(\frac{\ln n}{\ln\ln n}) - \alpha o(\ln n)} e^{\ln n} 
\epsilon^2 \}
\nonumber \\ &
\approx  C C_{t} \exp\{ - C e^{(\ln n) - (1+ \alpha) o(\ln n) + C o(\frac{\ln n}{\ln\ln n})}
\epsilon^2 \}
\nonumber \\ &
\approx  C C_{t} \exp\{ - C e^{(\ln n)[ 1 - (1+ \alpha) o(1)] + C o(\frac{1}{\ln\ln n})}
\epsilon^2 \}
\nonumber \\ &
\approx  C C_{t} \exp\{ - C n^{1 - o_\alpha(1)}
\epsilon^2 \}.
\end{align*}
One can check that the prefactor $C_t=[C^t (t!)^C]$ does not change the dominant order for $t = o(\frac{\ln n}{\ln\ln n})$. This shows that the bound vanishes as $n\to+\infty$ for $\lambda = w \rho_n^{-2}$ and $\rho_n =\Theta(\frac{1}{(\ln n)^\alpha})$ for any $\alpha\ge 0$. {As seen from \eqref{worsen} the bound worsen with decreasing $\rho_n$. So the result extends to $\rho_n=\Omega(\frac{1}{(\ln n)^\alpha})$.}

Note also that in the case of the rescaled bound of remark 1 below theorem~\ref{AMP-theorem}, the previous derivation is unchanged, up to the constant appearing in the $o_\alpha(1)$ that is changed some other $o_\alpha(1)$ (for $n$ big enough). Indeed, because $\rho_n =\Omega(\frac{1}{(\ln n)^\alpha})$ the $\rho_n^2$ or $\rho_n^4$ appearing in the rescaled bound can be absorbed in the $o_\alpha(1)$ of the previous derivation, for $n$ large enough.

\end{document}